\documentclass[manyauthors]{fundam}
\usepackage{url} 
\usepackage{graphicx}
\usepackage{amsfonts}
\usepackage{amsmath}
\usepackage{amssymb}
\usepackage[T1]{fontenc} 
\usepackage{mathbbol}
\DeclareSymbolFontAlphabet{\mathbbl}{bbold}
\usepackage{enumitem}

\usepackage[numbers]{natbib}

\usepackage{tikz}
\usepackage{varioref,bbm,dsfont}
\usetikzlibrary{petri}
\usetikzlibrary{decorations,decorations.pathreplacing,
decorations.pathmorphing,decorations.markings}

\usepackage{tkz-graph}
\usetikzlibrary{shapes.geometric}
\usetikzlibrary{shapes.symbols}
\usetikzlibrary{arrows,arrows.meta,backgrounds,positioning,fit}
\usetikzlibrary{matrix,calc,shapes,patterns}

\tikzstyle{mypetristyle}=[
      place/.style={circle, draw=black, fill=white, thick, minimum size=4mm},
      transition/.style={rectangle, draw=black, fill=black!30, thick, minimum height=3mm, minimum width=3mm},
      token/.style={circle,draw=black,fill=black,inner sep=0pt,minimum size=1mm}
]

\usepackage[linesnumbered,ruled]{algorithm2e}

\newcommand{\eqdef}{\overset{\mathrm{def}}{=\joinrel=}}

\newcommand{\myarrow}[3]{#1 \overset{#3}{\longrightarrow} #2}

\newcommand{\qed}{} 

\newcommand{\WMGineq}{\texorpdfstring{WMG$_\le$}}

\newcommand{\es}{\emptyset}
\newcommand{\pminus}{
\mbox{\textrm{$-$}\!\!\!\!\!\;\raisebox{1.3mm}{$\scriptstyle \bullet$}}\:}
\newcommand{\dt}{\bullet}
\newcommand{\lbul}{^\bullet}
\newcommand{\emptyseq}{\varepsilon}
\newcommand{\support}{\mathit{supp}}

\newcommand{\Parikh}{{\mathbf P}}
\newcommand{\is}{\iota}
\newcommand{\one}{\mathbbl{1}}
\newcommand{\zero}{\mathbbl{0}}

\def\N{\mathbb{N}}

\def\projection#1#2{\mathchoice
              {\setbox1\hbox{${\displaystyle #1}_{\scriptstyle #2}$}
              \projectionaux{#1}{#2}}
              {\setbox1\hbox{${\textstyle #1}_{\scriptstyle #2}$}
              \projectionaux{#1}{#2}}
              {\setbox1\hbox{${\scriptstyle #1}_{\scriptscriptstyle #2}$}
              \projectionaux{#1}{#2}}
              {\setbox1\hbox{${\scriptscriptstyle #1}_{\scriptscriptstyle #2}$}
              \projectionaux{#1}{#2}}}
\def\projectionaux#1#2{{#1\,\smash{\vrule height .8\ht1 depth .85\dp1}}_{\,#2}}

\tikzstyle{ltsNode}=[circle,fill=black, inner sep=0, minimum width=4pt] 
\tikzstyle{ltsnst}=[circle,draw=black,fill=black!10, very thick, minimum width=1mm]
\tikzstyle{ltsNodePattern}=[circle,fill=black!10, very thick, minimum width=1mm]
\tikzstyle{petriNode}=[place,minimum size=3mm,very thick,fill=black!10]

\tikzset{every picture/.style={mypetristyle}}
\tikzset{every label/.style={black!90}}

\usepackage{lineno,hyperref}
\modulolinenumbers[5]

\begin{document}

\title{On the Petri Nets with a Single Shared Place and Beyond}

\author{Thomas Hujsa\thanks{Supported by the STAE foundation/project DAEDALUS, Toulouse, France.}\\  
LAAS-CNRS\\Universit\'e de Toulouse, CNRS, INSA\\Toulouse, France
({thomas.hujsa@laas.fr})
\and Bernard Berthomieu\\
LAAS-CNRS\\Universit\'e de Toulouse, CNRS, INSA\\Toulouse, France
({bernard.berthomieu@laas.fr})
\and Silvano Dal Zilio\\
LAAS-CNRS\\Universit\'e de Toulouse, CNRS, INSA\\Toulouse, France
({silvano.dalzilio@laas.fr})
\and Didier Le Botlan\\
LAAS-CNRS\\Universit\'e de Toulouse, CNRS, INSA\\Toulouse, France
({didier.le.botlan@laas.fr})
}

\maketitle
\runninghead{T. Hujsa, B. Berthomieu, S. Dal Zilio, D. Le Botlan}{On the Petri Nets with a Single Shared Place}

\vspace*{-2cm} 

\begin{abstract}
%
%
Petri nets proved useful to describe various real-world systems, but many of their properties are very hard to check.
To alleviate this difficulty, subclasses are often considered.
The class of weighted marked graphs with relaxed place constraint (\WMGineq{} for short),
in which each place has at most one input and one output,
and the larger class of choice-free (CF) nets,
in which each place has at most one output, have been extensively studied to this end, with various applications.

In this work, we develop new properties related to the fundamental and intractable problems of reachability, liveness and reversibility 
in weighted Petri nets.
We focus mainly on the homogeneous Petri nets with a single shared place (H$1$S nets for short), 
which extend the expressiveness of CF nets 
by allowing one shared place 
(i.e.\ a place with at least two outputs and possibly several inputs)
under the homogeneity constraint 
(i.e.\ all the output weights of the shared place are equal).
Indeed, this simple generalization already yields new challenging problems and is expressive enough 
for modeling existing use-cases, justifying a dedicated study. 

One of our central results is the first characterization of liveness in a subclass of H$1$S nets more expressive than \WMGineq{}
that is expressed by the infeasibility of an integer linear program (ILP) of polynomial size.
This trims down the complexity to co-NP, contrasting with the known EXPSPACE-hardness of liveness 
in the more general case of weighted Petri nets.
In the same subclass, we obtain a new reachability property related to the live markings,
which is a variant of the well-known Keller's theorem.
Another central result is a new reversibility characterization for the live H$1$S class, simplifying its checking.
Finally, we apply our results to use-cases, highlight their scalability and discuss their extensibility to more expressive classes.

\end{abstract}


\maketitle

{ 

\section{Introduction}\label{intro.sec}

Petri nets, or equivalently vector addition systems (VASs), have proved useful to model numerous artificial and natural systems.
Their weighted version allows weights (multiplicities) on arcs,
making possible the bulk consumption or production of tokens, 
hence a more compact representation of the systems.

Many fundamental properties of Petri nets are decidable, although often hard to check.
Given a bounded Petri net,
a naive analysis can be performed by constructing its finite reachability graph,
whose size may be considerably larger than the net size. 
In case the net is unbounded, the reachability graph is infinite, thus cannot be constructed,
requiring to relate the net structure to the behavioral properties of interest.
To address this problem and to limit the cost of the computation, subclasses are often considered,
allowing to derive more efficiently the behavior from the structure only.
This approach has led to various polynomial-time checking methods
dedicated to several subclasses,
the latter being defined by structural restrictions in many cases~\cite{DesEsp,STECS,LAT98,HDM2016,HD2017,HD2018}.\\

\noindent {\bf Persistent systems, applications and properties.}
A Petri net system is persistent if, from each reachable marking, no transition firing can disable any other transition.
Persistence~\cite{keller,PersistentBaryl} is generally required when designing asynchronous hardware, 
notably to implement Signal Transition Graphs (STGs)--specially interpreted Petri nets--specifications as speed-independent circuits, 
ensuring the absence of hazards~\cite{STG-BDD1995,Carmona2004}, as well as to arbiter-free synchronization of processes~\cite{Lamport2003}. 
This property also appears in several workflow models~\cite{Workflow2003}.

Choice-free (CF) nets, also called output-nonbranching nets, force each place to have at most one output,
hence have no shared place.
They are a special case of persistent systems
and have already been widely investigated, e.g.\ in~\cite{tcs97,HDM2014,Hujsa2015,ChoiceFreeBestDevil2015}.

The CF class contains the weighted marked graphs with relaxed place constraint (\WMGineq{} for short),
in which each place has at most one input and one output. 
Various studies have been devoted to \WMGineq{}, sometimes with further constraints~\cite{WTS92,Sauer2003,March09,DH18,DH19FI}.
\WMGineq{} can model Synchronous DataFlow graphs~\cite{LeeM87}, 
which have been fruitfully used to design and analyze many real-world systems
such as embedded applications, notably Digital Signal Processing (DSP) applications \cite{LM87a,Emb09,Pino95}.

Embedded systems, among others, must fulfill a set of fundamental properties.
Typically,
all their functionalities must be preserved after any evolution of the system,
they must use a limited amount of memory,
and 
they should ideally have a cyclic, regular behavior from the start, 
avoiding an highly time-consuming transient phase~\cite{PhDHujsa}.\\
The corresponding properties in Petri nets
are most often defined as liveness
(meaning the existence, from each reachable marking, of a fireable sequence containing all transitions), 
boundedness
(meaning the existence of an upper bound on the number of tokens in each place over all reachable markings) 
and 
reversibility
(meaning the strong connectedness of the reachability graph).

Various characterizations and polynomial-time sufficient conditions of structural and behavioral properties, 
including liveness, boundedness and reversibility,
have been developed for CF nets and their \WMGineq{} subclass~\cite{WTS92,March09,Hujsa2015,HDM2016}
as well as for larger classes~\cite{STECS,HDM2016}.\\

\noindent {\bf Generalizations of CF nets and related classes}.
In this work, we introduce an extension of CF nets, the homogeneous Petri nets with at most one shared place (H$1$S nets for short),
in which the output weights of the shared place are all equal.
This class makes more flexible the modeling of applications with an underlying CF structure,
allowing to slightly move away from the determinism induced by the persistence.

Previous works considered the addition of shared places to (unit-weighted) marked graphs, yielding notably Augmented Marked Graphs (AMG). 
The latter are unit-weighted and must satisfy a set of further restrictions on their structure and initial marking~\cite{ChuXie97}.
This class does not contain the H$1$S nets nor is generalized by them.
AMG benefit from various results relating their structure to their behavior,
notably on liveness, boundedness, reversibility and reachability~\cite{ChuXie97}.
Other classes allow shared places through place merging and are studied in~\cite{Jiao2004},
but do not include our class.

As we shall highlight, H$1$S nets deserve a dedicated study with new theoretical developments,
and embody a part of the frontier between well-analyzable CF-like nets and the more tricky ones.
As use-cases, we take the \emph{Swimming pool protocol} 
from the Model Checking Contest\footnote{\url{https://mcc.lip6.fr/models.php}} (MCC)
and
the emblem of the International Conference on Application and Theory of Petri Nets and Concurrency (\emph{Petri nets}).\\

\noindent {\bf Contributions.}
We develop new checking methods for reachability, liveness and reversibility in weighted Petri nets, focusing mainly on the H$1$S class.
In the H$1$S-\WMGineq{} subclass,
i.e.\ the H$1$S systems in which the deletion of the shared place (if any) yields a \WMGineq{},
we study reachability and liveness.
In the larger H$1$S class, we investigate reversibility.

Our first main result is that, 
for each live H$1$S-\WMGineq{} system $(N,M_0)$, 
every \emph{potentially reachable}\footnote{i.e.\ that appears in a solution $(M,Y)$ of the \emph{state equation} $M = M_0 + I \cdot Y$ of the system, where $I$ is the incidence matrix.} marking $M$ is live, 
and 
that, for each such $M$, some marking is reachable from both $M_0$ and $M$.
This result, combined with a previous liveness characterization known for a larger class,
yields a new characterization of liveness based on the state equation,
which amounts to check the infeasibility of an Integer Linear Program (ILP) 
whose number of linear inequalities may be exponential in the number of places.\\
In case the H$1$S-\WMGineq{} is strongly connected and structurally bounded with a strongly connected underlying WMG,
we derive an ILP whose number of linear inequalities is linear in the number of places and transitions
and
where each inequality has his length linear in the number of places, transitions
and
the number of bits in the binary encoding of the largest number used.
This trims down the complexity of liveness checking to co-NP for a given Petri net system with binary-encoded numbers. 
This improves upon the EXPSPACE-hardness of liveness checking in weighted Petri nets.

Our second main result is the first characterization of reversibility for the H$1$S systems under the liveness assumption.
More precisely, we show that the existence of a \emph{T-sequence} feasible from the initial marking, 
meaning a sequence containing all transitions and leading to the same (initial) marking,
is (obviously) necessary, but also sufficient for reversibility under the liveness assumption.
This characterization was known to hold in other subclasses of weighted Petri nets,
but not yet known to hold in the H$1$S class.
This avoids to build the whole reachability graph for checking its strong connectedness,
which is even impossible when infinite (in case the Petri net is unbounded).

Besides, we highlight the sharpness of our liveness, reachability and reversibility 
conditions by providing counter-examples when only few assumptions are relaxed.

As use-cases, we study the \emph{Swimming pool protocol} and the emblem of the 
\emph{International Conference on Application and Theory of Petri Nets and Concurrency},
which can both be modeled with H$1$S-\WMGineq{} systems.
We apply our methods to these use-cases, highlighting their scalability.

Finally, we discuss foreseen extensions of this work.\\ 

\noindent {\bf Organization of the paper.}
In Section~\ref{SectionDefs}, we introduce general definitions and notations.
In Section~\ref{SectionSubclasses}, we define the main subclasses studied in this paper;
then, in Section~\ref{SectionPNproperties}, we recall some known properties of Petri nets related to reachability
and also propose a new reachability property.

In Section~\ref{PotentialReachDirected}, we investigate properties of the \emph{potential reachability graphs},
in which each marking solution of the state equation is represented.
We introduce \emph{initial directedness} as well as \emph{strong liveness},
and prove a relationship between them.

In Section~\ref{SecLiveness}, we develop our new checking methods for liveness and reachability dedicated to the H$1$S-\WMGineq{} systems.
In Section~\ref{SecReversibility}, we characterize reversibility in the H$1$S systems under the liveness assumption.

We apply our results to the use-cases in Section~\ref{UseCase}
and
we discuss possible extensions of this work with a modular approach in Section~\ref{ModularH1C}.

Finally, Section~\ref{SecConclu} forms our conclusion with further perspectives.

\section{General Definitions and Notations}\label{SectionDefs}

\noindent {\bf Petri net, incidence matrix, pre-/post-set, shared place, choice transition.}
A~{\em (Petri) net} is a tuple $N=(P,T,W)$ such that 
$P$ is a finite set of {\em places}, 
$T$ is a finite set of {\em transitions}, with $P\cap T=\es$, 
and $W$ is a weight function $W\colon((P\times T)\cup(T\times P))\to\N$ 
setting the weights on the arcs.
A {\em marking} of the net $N$ is a mapping from $P$ to $\N$, i.e.\ a member of $\N^P$, defining the number of tokens in each place of $N$.

A {\em (Petri net) system} is a tuple $S=(N,M_0)$ where $N$ is a net and $M_0$ is a marking, often called {\em initial marking}. 
The {\em incidence matrix} $I$ of $N$ (and $S$) is the integer place-transition matrix 
with components $I(p,t)=W(t,p)-W(p,t)$, for each place $p$ and each transition~$t$.

A net or system is {\em non-trivial} if it contains at least one place and one transition.

The {\em post-set} $n\lbul$ and {\em pre-set} $\lbul n$ of a node $n \in P \cup T$ are defined
as $n\lbul=\{n'\in P\cup T\mid W(n,n'){>}0\}$ and $\lbul n=\{n'\in P\cup T \mid W(n',n){>}0\}$.
Generalizing to any set $A$ of nodes, $\lbul A = \cup_{n \in A} {\lbul n}$ and $A\lbul = \cup_{n \in A} n\lbul$. 

A place $p$ is \emph{enabled} by a marking $M$ if for each output transition $t$ of $p$,  $M(p) \ge W(p,t)$.

A place $p$ is \emph{shared} if it has at least two outputs. 

A transition $t$ is a \emph{choice} transition if it has at least one input place that is shared. 
If no input place of a transition $t$ is shared, 
$t$ is called a \emph{non-choice} transition.

Some of these notions are illustrated in Figure \ref{ex1.fig}.\\

\noindent {\bf Firings and reachability in Petri nets.}
Consider a system $S=(N,M_0)$ with $N=(P,T,W)$.
A transition $t$ is {\em enabled} at $M_0$ (i.e.\ in $S$) 
if for each $p$ in ${\lbul t}$, $M_0(p) \ge W(p,t)$, 
in which case $t$ is {\em feasible} or {\em fireable} from $M_0$,
which is denoted by $M_0[t\rangle$.
The firing of $t$ from $M_0$ leads to the marking $M = M_0 + I[P,t]$ 
where $I[P,t]$ is the column of $I$ associated to $t$: this is denoted by $M_0[t\rangle M$.

A finite {\em (firing) sequence} $\sigma$ of length $n \ge 0$ on the set $T$,
denoted by $\sigma = t_{1} \ldots t_{n}$ with $t_{1} \ldots t_{n} \in T$, is a mapping $\{1, \ldots, n\} \to T$.
Infinite sequences are defined similarly as mappings $\N\setminus\{0\} \to T$.
A sequence $\sigma$ of length $n$ is {\em enabled} (or {\em feasible}, {\em fireable}) in $S$ 
if the successive states obtained,
$M_0 [t_1\rangle M_1 \ldots [t_n\rangle M_n$,
satisfy 
$M_{k-1}[t_k\rangle M_k$, for each $k$ in $\{1,\ldots, n\}$,
in which case $M_n$ is said to be {\em reachable} from $M_0$: we denote this by $M_0 [\sigma\rangle M_n$.
If $n=0$, $\sigma$ is the {\em empty sequence} $\epsilon$, implying $M_0 [\epsilon\rangle M_0$.
The set of markings reachable from $M_0$ is denoted by $R(S)$ or $[S\rangle$;
when it is clear from the context, it is also denoted by $R(M_0)$ or $[M_0\rangle$.

The {\em reachability graph} of $S$, denoted by $RG(S)$, is the rooted directed graph $(V,A,\is)$
where $V$ represents the set of vertices labeled bijectively with the markings $[M_0\rangle$,
$A$ is the set of arcs labeled with transitions of $T$
such that the arc $M \xrightarrow[]{t} M'$
belongs to $A$ if and only if
$M [t\rangle M'$ and $M \in [M_0\rangle$, and $\is$ is the root, labeled with $M_0$.

In Figure \ref{ex1.fig}, a weighted system is pictured on the left.
Its reachability graph is pictured on the right, where $v^T$ denotes the transpose of vector $v$.\\

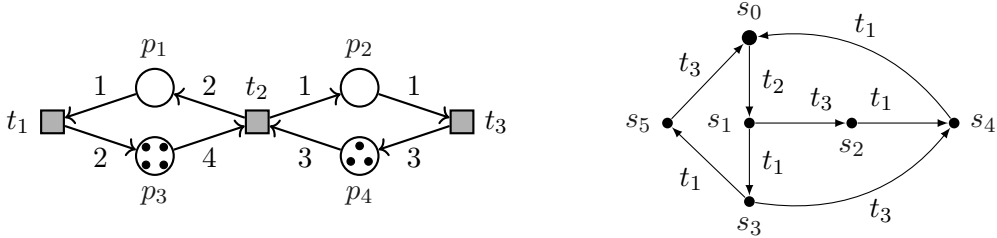
\begin{figure}[!h]
\centering
\raisebox{4mm}{
\begin{tikzpicture}[scale=0.9,place/.style={circle, draw=black, fill=white, thick, minimum size=4mm, inner sep=5pt}]

\node (p1) at (1.5,0.5) [place,label=below:$p_3$,tokens=4] {};
\node (p2) at (4.5,0.5) [place,label=below:$p_4$,tokens=3] {};
\node (p3) at (1.5,1.5) [place,label=above:$p_1$] {};
\node (p4) at (4.5,1.5) [place,label=above:$p_2$] {};

\node (t1) at (0,1) [transition] {};
\node (t2) at (3,1) [transition] {};
\node (t3) at (6,1) [transition] {};

\draw [->,thick, bend left=0] (t2) to node [above] {$2$} (p3);
\draw [->,thick, bend left=0] (t2) to node [above] {$1$} (p4);

\draw [->,thick, bend left=0] (p1) to node [below] {$4$} (t2);
\draw [->,thick, bend left=0] (p2) to node [below] {$3$} (t2);

\draw [->,thick, bend left=0] (p3) to node [above] {$1$} (t1);
\draw [->,thick, bend left=0] (p4) to node [above] {$1$} (t3);

\draw [->,thick, bend left=0] (t1) to node [below] {$2$} (p1);
\draw [->,thick, bend left=0] (t3) to node [below] {$3$} (p2);

 \node [anchor=east] at (t1.west) {$t_1$};
 \node [anchor=south] at (t2.north) {$t_2$};
 \node [anchor=west] at (t3.east) {$t_3$};

\end{tikzpicture}
}
\hspace*{1cm}
\begin{tikzpicture}[scale=0.9]

\node[ltsNode](s5)at(1.8,1.5)[label=left:$s_5$]{};
\node[ltsNode,minimum width=2mm](s0)at(3,2.75)[label=above:$s_0$]{}; 
\node[ltsNode](s1)at(3,1.5)[label=left:$s_1$]{}; 
\node[ltsNode](s3)at(3,0.35)[label=below:$s_3$]{};
\node[ltsNode](s2)at(4.5,1.5)[label=below:$s_2$]{};
\node[ltsNode](s4)at(6,1.5)[label=right:$s_4$]{};

\draw(s0)[]edge[-latex,bend left=0]node[ellipse,right,inner sep=2pt,pos=0.5]{$t_2$}(s1);
\draw(s1)[]edge[-latex,bend left=0]node[ellipse,right,inner sep=2pt,pos=0.5]{$t_1$}(s3);
\draw(s1)[]edge[-latex,bend left=0]node[ellipse,above,near end,inner sep=2pt]{$t_3$}(s2);
\draw(s2)[]edge[-latex,bend left=0]node[ellipse,above,near start,inner sep=2pt]{$t_1$}(s4);
\draw(s4)[]edge[-latex,bend right=30]node[ellipse,above,inner sep=2pt,pos=0.5]{$t_1$}(s0);
\draw(s5)[]edge[-latex,bend left=0]node[ellipse,above left,inner sep=2pt,pos=0.5]{$t_3$}(s0);
\draw(s3)[]edge[-latex,bend left=0]node[ellipse,below left,inner sep=2pt,pos=0.5]{$t_1$}(s5);
\draw(s3)[]edge[-latex,bend right=30]node[ellipse,below right,inner sep=2pt,pos=0.5]{$t_3$}(s4);

\end{tikzpicture}

\caption{ 
A system $S=(N,M_0)$ is pictured on the left. 
The pre-set $^\bullet t_2$ of $t_2$ is $\{p_3, p_4\}$
and
the post-set $t_2^\bullet$ of $t_2$ is $\{p_1, p_2\}$.
There is no shared place.
The reachability graph $RG(S)$ of $S$ is pictured on the right.
The initial marking is the state $s_0 = (0,0,4,3)^T$ in $RG(S)$.
The firing sequence $\sigma = t_2\, t_1\, t_3$ is feasible from $M_0$ and reaches the marking $(1,0,2,3)^T = s_4$.
} 
\label{ex1.fig}
\end{figure}

\noindent {\bf Relating reachability to the state equation.} 
Consider a system $S=(N,M_0)$ with its incidence matrix $I$, where $N=(P,T,W)$ is the underlying net.

The {\em state equation} associated to $S$ is expressed as $M = M_0 + I \cdot Y$,
where the variable $M$ takes its value in the set of markings 
and the variable $Y$ ranges over the set of vectors whose components are non-negative integers.

The set of markings {\em potentially reachable} in $S$ is defined as 
$PR(S) = \{ M \in \N^{\mid P\mid} \mid \exists Y \in \N^{\mid T\mid}, M = M_0 + I \cdot Y \}$;
this set is called the \emph{linearized reachability set} of $S$ in~\cite{LAT98}.
Potential reachability is a necessary condition for reachability, but it is not sufficient in general.

We denote by $PRG(S)$ the {\em potential reachability graph} of $S$, 
defined as the rooted directed graph $(V,A,\is)$
where $V$ represents the set of vertices $PR(S)$,
$A$ is the set of arcs labeled with transitions of $S$
such that, for each transition $t$, the arc $M \xrightarrow[]{t} M'$ belongs to $A$ if and only if
$M [t\rangle M'$ 
and 
$M \in PR(S)$, and $\is = M_0$ is the root.
$RG(S)$ is a subgraph of $PRG(S)$.\\

\noindent {\bf Vectors, semiflows, conservativeness and consistency.} 
The {\em support} of a vector $v$, denoted by $\support(v)$, is the set of the indices of its non-null components.
Consider any net $N=(P,T,W)$ with its incidence matrix~$I$.

A {\em T-vector} (respectively P-vector) is an element of $\N^{T}$ (respectively $\N^{P}$); 
it is called {\em prime} if 
the greatest common divisor of its components is one 
(i.e.\ its components do not have a common non-unit factor). 

The {\em Parikh vector} $\Parikh(\sigma)$ of a finite sequence $\sigma$ of transitions is the T-vector
counting the number of occurrences of each  transition in $\sigma$,
and the {\em support} of $\sigma$ is the support of its Parikh vector, 
i.e.
$\support(\sigma)=\support(\Parikh(\sigma))=\{t\in T\mid\Parikh(\sigma)(t)>0\}$.

We denote by $\zero^n$ (respectively $\one^n$) the column vector of size $n$ whose components are all equal to~$0$ (respectively~$1$).
We denote by $\one_t^n$ the column vector of size $n$ whose only non-null component has index $t$ and equals $1$.
The exponent $n$ may be omitted when it is clear from the context.

A {\em T-semiflow} (respectively P-semiflow) $\nu$ of the net is a non-null T-vector (respectively P-vector)
whose components are only non-negative integers (i.e.\ $\nu\gneqq\zero$)
and such that $I \cdot\nu=\zero$ (respectively $\nu^T \cdot I = \zero$). 
A T-(respectively P-)semiflow is called {\em minimal} 
when it is prime and its support 
is not a proper superset of the support of any other T-(respectively P-)semiflow.

$N$ is {\em conservative}, or {\em invariant}, if a P-semiflow $X \in \N^{|P|}$ exists for $I$
such that
$X \ge \one^{|P|}$, in which case $X$ is called a {\em conservativeness vector}.
In case such a P-vector $X$ exists and, in addition, $X=\one^{|P|}$, $N$ is called {\em $1$-conservative}, or {\em $1$-invariant}.

$N$ is {\em consistent} if a T-semiflow $Y \in \N^{|T|}$ exists for $I$
such that
$Y \ge \one^{|T|}$, in which case $Y$ is called a {\em consistency vector}.

Such vectors are frequently exploited in the structural and behavioral analysis of Petri nets~\cite{LAT98}.\\

\noindent {\bf Deadlockability, liveness, boundedness and reversibility.}
Consider any system $S=(N,M_0)$.
A transition $t$ is {\em dead} in $S$ if no marking of $[M_0\rangle$ enables $t$.
A {\em deadlock}, or {\em dead marking}, is a marking enabling no transition.
$S$ is {\em deadlock-free} if no deadlock belongs to $[M_0\rangle$; otherwise it is {\em deadlockable}. 
%

A transition $t$ is {\em live} in $S$ if for every marking $M$ in $[M_0\rangle$,
there is a marking $M' \in [M\rangle$ enabling~$t$.
$S$ is {\em live} if every transition is live in $S$.
$N$ is {\em structurally live} if a marking $M$ exists such that $(N,M)$ is live.

$S$ is {\em $k$-bounded} (or {\em $k$-safe}) if an integer $k$ exists such that: for each $M$ in $[M_0\rangle$, 
for each place~$p$, $M(p) \le k$.
It is {\em bounded} if an integer $k$ exists such that $S$ is $k$-bounded. 
$N$ is {\em structurally bounded} if $(N,M)$ is bounded for each $M$.

$N$ is {\em well-formed} if it is structurally bounded and structurally live.

A marking $M$ is a {\em home state} of $S$ if it can be reached from every marking in $[M_0\rangle$. 
$S$ is {\em reversible} if its initial marking is a home state, meaning that $RG(S)$ is strongly connected.

The underlying net $N$ in Figure~\ref{ex1.fig} is structurally live and bounded, thus well-formed.
In the same figure, the system $S$ is live, $4$-bounded and reversible, thus non-deadlockable,
which can be checked on its finite reachability graph.\\

\noindent{\bf Subnets and subsystems.} 
Let $N=(P,T,W)$ and $N'=(P',T',W')$ be two nets.
$N'$ is a {\em subnet} of $N$ if $P'$ is a subset of $P$, $T'$ is a subset of $T$, 
and $W'$ is the restriction of $W$ to $(P'\times T') \cup (T'\times P')$.
$S'=(N',M_0')$ is a {\em subsystem} of $S=(N,M_0)$ if $N'$ is a subnet of $N$ 
and its initial marking $M_0'$ is the restriction of $M_0$ to $P'$, denoted by $M_0'= \projection{M_0}{P'}$.

$N'$ is a {\em P-subnet} of $N$ if $N'$ is a subnet of $N$
and $T'=\mathit{^\bullet P'} \cup P'^\bullet$,
the pre- and post-sets being taken in $N$. 
$S'=(N',M_0')$ is a {\em P-subsystem} of $S=(N,M_0)$ if $N'$ is a P-subnet of $N$ 
and $S'$ is a subsystem of $S$.
We say that $N'$ and $S'$ are {\em induced} by the subset $P'$.


Subsystems play a fundamental role in the analysis of Petri nets,
typically leading to characterizations relating the system's behavior to properties of its subsystems;
this approach yielded polynomial-time checking methods in various subclasses, see e.g.~\cite{March09,HDM2016,HD2018}.
We exploit such subsystems in this paper to obtain some of our new results.\\

\noindent{\bf Siphons.} 
Consider a net $N = (P,T,W)$.
A 
subset $D \subseteq P$ of places is a {\em siphon} (sometimes also called a deadlock) 
if $\lbul D \subseteq D \lbul$.
%
%
For the sake of conciseness, in this paper, 
we assume siphons 
to be non-empty, unless emptiness is explicitly allowed.

There exist various studies relating the structure to the behavior with the help of siphons, 
see e.g.~\cite{DesEsp,STECS}.
Intuitively,
insufficiently marked siphons induce P-subsystems that cannot receive new tokens and thus block some transitions irremediably.

A siphon 
is {\em minimal} if it does not contain any proper siphon, 
i.e.\ there is no subset of the same type with smaller cardinality.

In Figure~\ref{ex1.fig}: on the left,
$\{p_1, p_2, p_3, p_4\}$ is a siphon 
and includes the smaller minimal ones $\{p_1, p_3\}$ and $\{p_2, p_4\}$,
while $\{p_1\}$ is not a siphon. 


\section{Petri net subclasses}\label{SectionSubclasses}

Let us define the Petri net subclasses studied in this work.

\subsection{Classical restrictions on the structure}

We define subclasses from restrictions on the structure of any weighted net $N$.
Most of them have been exploited in various works,
e.g.\ in~\cite{tcs97,Jiao2004,HDM2016}.\\

\noindent {$-$ Subclasses defined by restrictions on the weights:}

$N$ is {\it unit-weighted} (or {\it plain}, {\it ordinary}) if no arc weight exceeds~$1$; 
$N$ is {\em homogeneous} if for each place $p$, all output weights of $p$ are equal.
In particular, unit-weighted nets are homogeneous.
In this paper, for any class of nets C, we denote by HC the homogeneous subclass of C.
Examples are pictured in Figure~\ref{ExSubclasses2}.\\

\noindent {$-$ Subclasses without shared places:}

$N$ is {\it choice-free} (CF for short, also called place-output-nonbranching)
if each place has at most one output,
i.e.\ $\forall p\in P$, $|p^\dt|\leq 1$;
%
it is a {\em weighted marked graph with relaxed place constraint} (\WMGineq{})
if it is choice-free and, in addition, each place has at most one input,
i.e.
$\forall p\in P$, $|{}^\dt p|\leq 1$ and $|p{}^\dt|\leq 1$.
\WMGineq{} contain the {\em weighted T-systems} (WTS) of \cite{WTS92}, 
also known as {\em weighted marked graphs} (WMG) 
and 
{\em weighted event graphs} (WEG)~\cite{March09},
in which $\forall p\in P$, $|{}^\dt p| = 1$ and $|p{}^\dt| = 1$.
The system on the left of Figure~\ref{ex1.fig} is a WMG.
%
Well-studied subclasses encompass {\em marked graphs}~\cite{MDG71}, also known as {\em T-nets}~\cite{DesEsp},
which are WMG with unit weights.\\
%

\noindent {$-$ Subclasses with shared places:}

We define the class H$k$S as the set of homogeneous nets with at most $k$ shared places.
Figure~\ref{OneChoiceWMG} depicts an H$1$S system on the left and the CF net obtained by deleting the only shared place on the right.
We shall also consider H$k$S-\WMGineq{} systems, which are H$k$S systems with the further restriction that
the deletion of the $k$ shared places yields a \WMGineq{}.

A net $N$ is {\em asymmetric-choice} (AC) if it satisfies 
the following condition for any two input places $p_1$, $p_2$ of each synchronization~$t$,
$p_1^\bullet \subseteq p_2^\bullet$
or
$p_2^\bullet \subseteq p_1^\bullet$.
It is {\em free-choice} (FC) if for any two input places $p_1$, $p_2$ of each synchronization $t$,
$p_1^\bullet = p_2^\bullet$.
Thus, FC nets form a subclass of AC nets.
Examples are given in Figure~\ref{ExSubclasses2}.

A net $N$ is a {\em state machine} if it is unit-weighted 
and each transition has exactly one input and one output. 

\begin{figure}[!h]
 
 
\begin{minipage}{0.3\linewidth}

\centering

\hspace*{1cm}
\begin{tikzpicture}[mypetristyle,scale=1]

\node (p1) at (0,0) [place,thick] {};
\node (p2) at (1.25,0) [place,thick] {};

\node [anchor=east] at (p1.west) {$p_1$};
\node [anchor=west] at (p2.east) {$p_2$};

\node (t1) at (0,1.2) [transition,thick] {};
\node (t2) at (1.25,1.2) [transition,thick] {};

\node [anchor=east] at (t1.west) {$t_0$};
\node [anchor=west] at (t2.east) {$t_1$};

\draw [->,thick] (p1) to node [left] {$2$} (t1);
\draw [->,thick] (p1) to node [below,near start] {$2$} (t2.south west);
\draw [->,thick] (p2) to node [below,near start] {$3$} (t1.south east);
\draw [->,thick] (p2) to node [right] {$3$} (t2);

\end{tikzpicture}

\end{minipage}
\begin{minipage}{0.33\linewidth}

\centering

\begin{tikzpicture}[mypetristyle,scale=1]

\node (p1) at (0.5,0) [place,thick] {};
\node (p2) at (2,0) [place,thick] {};

\node [anchor=east] at (p1.west) {$p_1$};
\node [anchor=west] at (p2.east) {$p_2$};

\node (t0) at (0,1.2) [transition,thick] {};
\node (t1) at (1,1.2) [transition,thick] {};
\node (t2) at (2,1.2) [transition,thick] {};
\node (t3) at (3,1.2) [transition,thick] {};

\node [anchor=east] at (t1.west) {$t_1$};
\node [anchor=east] at (t2.west) {$t_2$};
\node [anchor=east] at (t0.west) {$t_0$};
\node [anchor=east] at (t3.west) {$t_3$};

\draw [->,thick] (p1) to node [below right, near start] {$2$} (t1);
\draw [->,thick] (p2) to node [above, near start] {$3$} (t1.south east);
\draw [->,thick] (p2) to node [right] {$3$} (t2);
\draw [->,thick] (p2) to node [below right, near end] {$3$} (t3);

\draw [->,thick] (p1) to node [left,near start] {$2$} (t0.south);
\draw [->,thick] (p2) to node [below,near start] {$3$} (t0.south east);

\end{tikzpicture}

\end{minipage}
\begin{minipage}{0.3\linewidth}

\centering

\begin{tikzpicture}[mypetristyle,scale=1]

\node (p1) at (0.25,0) [place,thick] {};
\node (p2) at (2.25,0) [place,thick] {};

\node [anchor=east] at (p1.west) {$p_1$};
\node [anchor=west] at (p2.east) {$p_2$};

\node (t0) at (0.25,1.2) [transition,thick] {};
\node (t1) at (1.25,1.2) [transition,thick] {};
\node (t2) at (2.25,1.2) [transition,thick] {};

\node [anchor=east] at (t1.west) {$~t_1$};
\node [anchor=west] at (t2.east) {$t_2$};
\node [anchor=east] at (t0.west) {$t_0$};

\draw [->,thick] (p1) to node [right, near start] {$2$} (t1);
\draw [->,thick] (p2) to node [left, near start] {$3$} (t1);
\draw [->,thick] (p2) to node [below left, very near end] {$3$} (t2);
\draw [->,thick] (p1) to node [below right, very near end] {$2$} (t0);

\end{tikzpicture}

\end{minipage}
%
%
%
%
%
%
%
%
%

\caption{The first net (on the left) is HFC, the second one is HAC. 
The third net is homogeneous, non-AC since $\mathit{\lbul t_1} = \{p_1,p_2\}$,
while 
$\mathit{p_1\lbul} \not\subseteq \mathit{p_2\lbul}$
and
$\mathit{p_2\lbul} \not\subseteq \mathit{p_1\lbul}$.
%
%
None of these nets is CF.
}

\label{ExSubclasses2}

\end{figure}
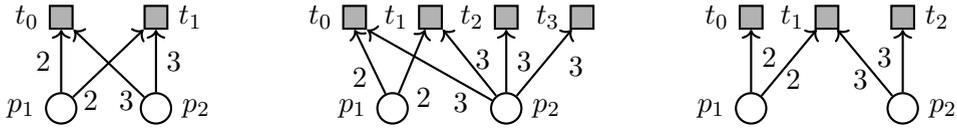

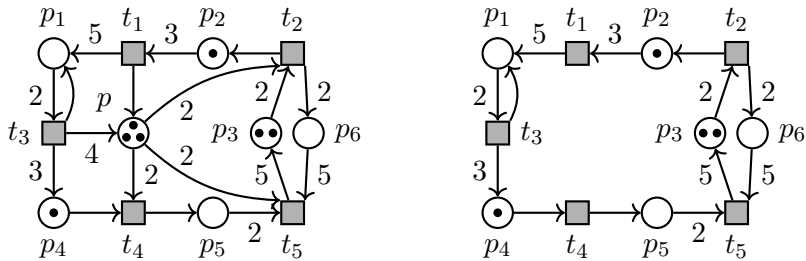
\begin{figure}
\centering
\begin{tikzpicture}[scale=0.7,mypetristyle]

\node (p) at (2.5,1.5) [place,tokens=3] {};
\node [anchor=south east] at (p.north west) {$p$};

\node (p1) at (1,3) [place] {};
\node [anchor=south] at (p1.north) {$p_1$};

\node (p2) at (4,3) [place,tokens=1] {};
\node [anchor=south] at (p2.north) {$p_2$};

\node (p3) at (5,1.5) [place,tokens=2] {};
\node [anchor=east] at (p3.west) {$p_3$};

\node (p4) at (1,0) [place,tokens=1] {};
\node [anchor=north] at (p4.south) {$p_4$};

\node (p5) at (4,0) [place] {};
\node [anchor=north] at (p5.south) {$p_5$};

\node (p6) at (5.8,1.5) [place] {};
\node [anchor=west] at (p6.east) {$p_6$};

\node (t1) at (2.5,3) [transition,thick] {};
\node (t2) at (5.5,3) [transition,thick] {};
\node (t3) at (1,1.5) [transition,thick] {};
\node (t4) at (2.5,0) [transition,thick] {};
\node (t5) at (5.5,0) [transition,thick] {};

\node [anchor=east] at (t3.west) {$t_3$};
\node [anchor=south] at (t1.north) {$t_1$};
\node [anchor=north] at (t5.south) {$t_5$};
\node [anchor=north] at (t4.south) {$t_4$};
\node [anchor=south] at (t2.north) {$t_2$};

\draw [->,thick] (t2.south east) to node [right] {$2$} (p6);
\draw [->,thick] (p6) to node [right] {$5$} (t5.north east);
\draw [->,thick] (t1) to node [left] {} (p);
\draw [->,thick] (p) to node [right] {$2$} (t4);
\draw [->,thick] (p1) to node [left] {$2$} (t3);
\draw [->,thick] (t3) to node [left] {$3$} (p4);
\draw [->,thick] (p4) to node [] {} (t4);
\draw [->,thick] (t4) to node [] {} (p5);
\draw [->,thick] (p5) to node [below] {$2$} (t5);
\draw [->,thick] (t5) to node [left] {$5$} (p3);
\draw [->,thick] (p3) to node [left] {$2$} (t2);
\draw [->,thick, bend right=20] (p) to node [above right, near start, inner sep=1pt] {$2$} (t5.north west);
\draw [->,thick, bend left=20] (p) to node [below right, near start, inner sep=1pt] {$2$} (t2.south west);
\draw [->,thick] (t2) to node [] {} (p2);
\draw [->,thick] (p2) to node [above] {$3$} (t1);
\draw [->,thick] (t1) to node [above] {$5$} (p1);
\draw [->,thick,bend right=30] (t3.north east) to node [above] {} (p1);
\draw [->,thick] (t3) to node [below] {$4$} (p);
\end{tikzpicture}
\hspace*{1cm}
\begin{tikzpicture}[scale=0.7,mypetristyle]


\node (p1) at (1,3) [place] {};
\node [anchor=south] at (p1.north) {$p_1$};

\node (p2) at (4,3) [place,tokens=1] {};
\node [anchor=south] at (p2.north) {$p_2$};

\node (p3) at (5,1.5) [place,tokens=2] {};
\node [anchor=east] at (p3.west) {$p_3$};

\node (p4) at (1,0) [place,tokens=1] {};
\node [anchor=north] at (p4.south) {$p_4$};

\node (p5) at (4,0) [place] {};
\node [anchor=north] at (p5.south) {$p_5$};

\node (p6) at (5.8,1.5) [place] {};
\node [anchor=west] at (p6.east) {$p_6$};

\node (t1) at (2.5,3) [transition,thick] {};
\node (t2) at (5.5,3) [transition,thick] {};
\node (t3) at (1,1.5) [transition,thick] {};
\node (t4) at (2.5,0) [transition,thick] {};
\node (t5) at (5.5,0) [transition,thick] {};

\node [anchor=west] at (t3.east) {$t_3$};
\node [anchor=south] at (t1.north) {$t_1$};
\node [anchor=north] at (t5.south) {$t_5$};
\node [anchor=north] at (t4.south) {$t_4$};
\node [anchor=south] at (t2.north) {$t_2$};

\draw [->,thick] (t2.south east) to node [right] {$2$} (p6);
\draw [->,thick] (p6) to node [right] {$5$} (t5.north east);
%
%
\draw [->,thick] (p1) to node [left] {$2$} (t3);
\draw [->,thick] (t3) to node [left] {$3$} (p4);
\draw [->,thick] (p4) to node [] {} (t4);
\draw [->,thick] (t4) to node [] {} (p5);
\draw [->,thick] (p5) to node [below] {$2$} (t5);
\draw [->,thick] (t5) to node [left] {$5$} (p3);
\draw [->,thick] (p3) to node [left] {$2$} (t2);
%
\draw [->,thick] (t2) to node [] {} (p2);
\draw [->,thick] (p2) to node [above] {$3$} (t1);
\draw [->,thick] (t1) to node [above] {$5$} (p1);
\draw [->,thick,bend right=30] (t3.north east) to node [above] {} (p1);
%
\end{tikzpicture}
\caption{
Deleting place $p$ in the H$1$S net on the left yields the CF net on the right.
}
\label{OneChoiceWMG}
\end{figure}

\section{Checking marking reachability in weighted Petri nets}\label{SectionPNproperties}

In this section, we recall the classical difficulties, some previous approaches and results 
for tackling reachability problems in weighted Petri nets.
We also propose a new simple related property.\\

\noindent {\bf Difficulty of analysis.}
In weighted Petri nets, the liveness checking problem is EXPSPACE-hard~\cite{Lipton76,esparza1994decidability,Complexity95},
the reversibility checking problem is PSPACE-hard~\cite{esparza1996decidability}
and
the marking reachability problem is non-elementary~\cite{ReachNonElementary2019}.
Besides, we recall next that most behavioral properties of Petri nets
are not preserved upon any increase of the initial marking, limiting the scalability of checking methods.\\

\noindent {\bf (Non-)Monotonicity.}
Consider any behavioral property $\mathcal{P}$. 
A system $(N,M_0)$ is said to satisfy $\mathcal{P}$ \emph{monotonically} if $(N,M)$ satisfies $\mathcal{P}$ for each marking $M \ge M_0$.
All the live choice-free systems are known to be \emph{m-live}, meaning they satisfy liveness monotonically~\cite{tcs97}.
However, generally, numerous fundamental properties, including liveness, reversibility, boundedness and deadlock-freeness,
are not preserved upon any increase of the initial marking, 
see e.g.~\cite{LAT98,DSSP98,Homothetic2012,HD2018}.
In Figure~\ref{NonMonotonicH1S}, we show that live, unit-weighted H$1$S systems are not always m-live nor m-reversible,
taking inspiration from Figure~$1$ in~\cite{ImportanceDeadTrap2010}.\\
\begin{figure}[!h]
 
%
\centering

\begin{tikzpicture}[mypetristyle,scale=0.85]

\node (p1) at (-0.25,0) [place,tokens=2] {};
\node (p2) at (1.18,0) [place,tokens=1] {};
\node (p3) at (1.82,0) [place,tokens=0] {};
\node (p4) at (2.9,0.4) [place] {};
\node (p5) at (2.9,-0.4) [place,tokens=0] {};
\node (p6) at (3.2,1.75) [place,tokens=0] {};

\node [anchor=east] at (p1.west) {$p_1$};
\node [] at (p2.west) {$p_2~~~~~$};
\node [] at (p3.east) {$~~~~~p_3$};
\node [anchor=south] at (p4.north) {$p_4$};
\node [anchor=north] at (p5.south) {$p_5$};
\node [anchor=west] at (p6.east) {$p_6$};

\node (t1) at (1.5,1) [transition] {};
\node (t2) at (1.5,-1) [transition] {};
\node (t3) at (4.25,0) [transition] {};
\node (t4) at (0.75,1.75) [transition] {};

\node [anchor=west] at (t1.east) {$t_1$};
\node [anchor=west] at (t2.east) {$t_2$};
\node [anchor=west] at (t3.east) {$t_3$};
\node [anchor=east] at (t4.west) {$t_4$};

\draw [->,thick] (p1) to node [above] {} (t1);
\draw [->,thick] (p1) to node [above] {} (t2);

\draw [->,thick] (t2) to node [above right] {} (p2);
\draw [->,thick] (p2) to node [above left] {} (t1);

\draw [->,thick] (t1) to node [above] {} (p3);
\draw [->,thick] (p3) to node [above] {} (t2);

\draw [->,thick] (t2.north east) to node [above] {} (p5);
\draw [->,thick] (p5) to node [above] {} (t3);

\draw [->,thick] (t1.south east) to node [right] {} (p4);
\draw [->,thick] (p4) to node [above] {} (t3);

\draw [->,thick] (t3) to node [left] {} (p6);
\draw [->,thick,bend left=100] (t3.south) to node [below] {} (p1.south);

\draw [->,thick] (p6) to node [below] {} (t4);
\draw [->,thick] (t4) to node [below] {} (p1);

\end{tikzpicture}
\hspace*{1.5cm}
\begin{tikzpicture}[mypetristyle,scale=0.85]

\node (p1) at (-0.25,0) [place,tokens=2] {};
\node (p2) at (1.18,0) [place,tokens=2] {};
\node (p3) at (1.82,0) [place,tokens=0] {};
\node (p4) at (2.9,0.4) [place] {};
\node (p5) at (2.9,-0.4) [place,tokens=0] {};
\node (p6) at (3.2,1.75) [place,tokens=0] {};

\node [anchor=east] at (p1.west) {$p_1$};
\node [] at (p2.west) {$p_2~~~~~$};
\node [] at (p3.east) {$~~~~~p_3$};
\node [anchor=south] at (p4.north) {$p_4$};
\node [anchor=north] at (p5.south) {$p_5$};
\node [anchor=west] at (p6.east) {$p_6$};

\node (t1) at (1.5,1) [transition] {};
\node (t2) at (1.5,-1) [transition] {};
\node (t3) at (4.25,0) [transition] {};
\node (t4) at (0.75,1.75) [transition] {};

\node [anchor=west] at (t1.east) {$t_1$};
\node [anchor=west] at (t2.east) {$t_2$};
\node [anchor=west] at (t3.east) {$t_3$};
\node [anchor=east] at (t4.west) {$t_4$};

\draw [->,thick] (p1) to node [above] {} (t1);
\draw [->,thick] (p1) to node [above] {} (t2);

\draw [->,thick] (t2) to node [above right] {} (p2);
\draw [->,thick] (p2) to node [above left] {} (t1);

\draw [->,thick] (t1) to node [above] {} (p3);
\draw [->,thick] (p3) to node [above] {} (t2);

\draw [->,thick] (t2.north east) to node [above] {} (p5);
\draw [->,thick] (p5) to node [above] {} (t3);

\draw [->,thick] (t1.south east) to node [right] {} (p4);
\draw [->,thick] (p4) to node [above] {} (t3);

\draw [->,thick] (t3) to node [left] {} (p6);
\draw [->,thick,bend left=100] (t3.south) to node [below] {} (p1.south);

\draw [->,thick] (p6) to node [below] {} (t4);
\draw [->,thick] (t4) to node [below] {} (p1);

\end{tikzpicture}


\caption{
On the left, a live and reversible H$1$S-\WMGineq{} system.
Adding a token to $p_2$ leads to the non-live and non-reversible system on the right.
}
\label{NonMonotonicH1S}
\end{figure}
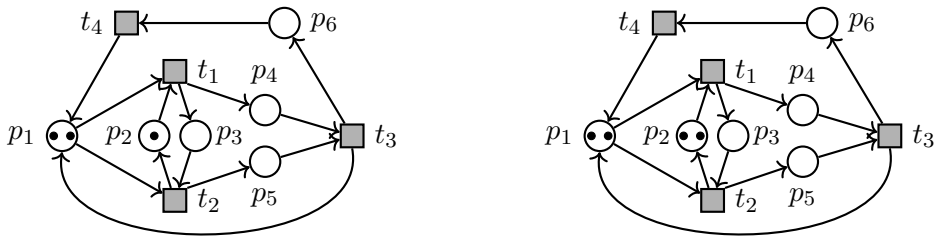

\noindent {\bf Deadlocked siphons.}
In various subclasses (see e.g.~\cite{STECS,HD2017}), 
non-liveness and deadlockability are tightly related to the existence 
of a reachable marking at which some \emph{siphon} is \emph{deadlocked}, defined as follows.
Formally, a siphon $D$ of a system $S=(N,M)$ is said to be \emph{deadlocked} if, 
for each place $p$ in $D$, for each $t \in p^\bullet$, $M(p) < W(p,t)$. 
Each transition having an input place in a deadlocked siphon can never be fired again.
Moreover, each transition having an output place in a deadlocked siphon necessarily has an input place in the same siphon,
hence can never be fired again.
%
Notice that a deadlocked siphon may have places without output that contain an arbitrary number of tokens.\\

In the rest of this section, 
we recall classical approaches for simplifying the analysis of reachability, liveness and reversibility,
and 
derive a new sufficient condition of siphon deadlockability.\\

\noindent {\bf Exploiting the state equation.}
For any weighted system $(N,M_0)$ with incidence matrix $I$,
an integer linear program (ILP) can be used as in~\cite{LHS1993} to check, 
for each minimal siphon $D$ of $N$, 
if a solution $(M,Y)$ to the state equation $M = M_0 + I \cdot Y$ exists
such that $D$ is deadlocked at $M$.

If such a solution exists, one cannot conclude that $M$ is indeed reachable.
If no such solution exists for any minimal siphon $D$,
then no siphon is deadlocked at any reachable marking,
which is necessary but not always sufficient for liveness;
in some subclasses, the non-deadlockability of all siphons is sufficient for liveness~\cite{STECS,Jiao2004}.
In such cases, solving the state equation with an ILP can be much more efficient than building the reachability graph,
even though the general ILP solving problem is NP-complete.

However, in addition to solving the (in)equalities describing the state equation,
it generally needs to compute a set of siphons, the cardinality of which grows exponentially 
in the number of places of the net, and, for each such siphon $D$, 
one must encode the fact that $D$ is deadlocked at $M$, which usually needs $|D|$ additional inequalities.
Thus, using a naive approach, one needs to solve an exponential number of ILPs, each of which has a number of inequalities 
at least equal to the number of places
and
a number of variables at least equal to the number of places and transitions.\\

\noindent {\bf Polynomial-time sufficient conditions.}
Despite intractability in the general case, 
polynomial-time sufficient conditions of monotonic liveness and reversibility exist for well-formed equal-conflict (EC, or HFC) nets
and
join-free (JF) nets (i.e.\ without synchronizations)~\cite{March09,HDM2014,HDM2016,HD2018},
simplifying the study of reachability.
Their efficiency comes from structural methods,
relating the structure (e.g.\ siphons, subnets, the incidence matrix, the initial marking...) 
to the behavior (i.e.\ to the reachability graph).\\

\noindent {\bf Dividing the marking in the unit-weighted case.}
For any system $(N,M_0)$, in case $M_0$ can be divided by some positive integer $k$,
the behavior of $(N,M_0)$ can sometimes be deduced from the behavior of $(N,M_0/k)$,
while checking the latter is generally much less costly.
This idea is related to \emph{separability}~\cite{BD2011,BHW18}
and 
to the \emph{homothetic markings} studied in~\cite{Homothetic2012}, 
a relationship with the fluidized version of the net being also investigated in the latter.

Before deriving a new simple sufficient condition for siphon deadlockability based on marking division,
we need the next notation.\\

\noindent {\bf Concatenation $\mathbf{\sigma^n}$:} 
For a sequence $\sigma$ and a positive integer~$n$, $\sigma^n$ denotes the concatenation of $\sigma$ taken $n$ times.

\begin{proposition}[Deadlocked siphons in unit-weighted systems]
Consider any system $S=(N,M_0)$ such that $M_0 = k \cdot M_0'$ for some positive integer $k$.
If some siphon $D$ of $N$ is deadlocked at some marking $M'$ reachable in $S'=(N,M_0')$,
then $D$ is deadlocked at $k \cdot M'$, which is reachable in $S$.
\end{proposition}

\begin{proof}
Suppose that a sequence $\sigma$ feasible in $S'$ leads to a marking $M'$
such that some siphon $D$ is deadlocked in $(N,M')$. 
It means that, for each place $p$ of $D$,
either $p$ has some output and is empty at $M'$, or $p$ has no output.
In $S$, $\sigma^k$ is clearly feasible, leading to the marking $M = k \cdot M'$ at which $D$ is deadlocked,
since each place empty at $M'$ remains empty at $M$ and each other place (if any) of $D$ has no output.
\end{proof}

When the marking is not divisible, we need to exploit other approaches.
In the next section, we introduce notions dealing with reachability and liveness in weighted Petri nets,
and prove a new related property.
Then, 
in the rest of this paper,
we develop new and more efficient checking techniques for the H$1$S class.

\section{Directedness and strong liveness}\label{PotentialReachDirected}

In this section, we first introduce in Subsection~\ref{SubsecDirected}
the notion of {\em directedness} of the potential reachability graph,
with variants,
taking inspiration from~\cite{STECS,DSSP98}; these can be seen as kinds of confluence.
Then, in Subsection~\ref{SubsecDirLiveSubclasses}, we introduce \emph{strong liveness}
and prove a new general property relating initial directedness to strong liveness.
%
We also show that the converse of this property does not hold in the unit-weighted asymmetric-choice class.
Finally, in Subsection~\ref{SubsecPersistent}, we define \emph{persistent} systems
and
recall Keller's theorem, which applies to this class and is also a kind of confluence.

\subsection{Directedness and variants}\label{SubsecDirected}

\begin{definition}[Directedness of the potential reachability graph]
Let us consider any system $S = (N,M_0)$ and its potential reachability graph $PRG(S)$:\\
$-$ $PRG(S)$ is \emph{directed} if every two potentially reachable markings have a common reachable marking.
More formally: $\forall M_1, M_2 \in PR(S)$: $R((N,M_1)) \cap R((N,M_2)) \neq \emptyset$.\\
$-$ $PRG(S)$ is \emph{initially directed} if $\forall M_1 \in PR(S): R(S) \cap R((N,M_1)) \neq \emptyset$.
\qed
\end{definition}

The directedness of $PRG(S)$ is called {\em structural directedness} in~\cite{LenderProc2006}.
We shall also consider the particular case of directedness restricted to the reachability graph,
i.e.\ when every two reachable markings have a common reachable marking.
Figure~\ref{FigDefDirectedness} illustrates these properties, where $M$ denotes some common reachable marking.

\begin{figure}[!h]

\centering

\begin{tabular}{ccc}
 
\begin{minipage}{0.3\linewidth}

\centering

\begin{tikzpicture}[scale=0.7]
\node[ltsNode,label=below:$M_0$](n0)at(1,0){};
\node[ltsNode,label=left:$M_1$](n1)at(0,1){};
\node[ltsNode,label=right:$M_2$](n2)at(2,1){};
\node[ltsNode,label=above:$M$](n3)at(1,2){};
\draw[-{>[scale=2.5,length=2,width=2]},dashed](n0)to node[auto,swap]{}(n1);
\draw[-{>[scale=2.5,length=2,width=2]},dashed](n0)to node[auto,swap]{}(n2);
\draw[-{>[scale=2.5,length=2,width=2]}](n1)to node[auto,swap]{}(n3);
\draw[-{>[scale=2.5,length=2,width=2]}](n2)to node[auto,swap]{}(n3);
\end{tikzpicture}\\
\centering
\footnotesize
Directedness of $PRG(S)$
\end{minipage}
&
\begin{minipage}{0.3\linewidth}
 \centering

\begin{tikzpicture}[scale=0.7]
\node[ltsNode,label=below:$M_0$](n0)at(1,0){};
\node[ltsNode,label=left:$M_1$](n1)at(0,1){};
\node[ltsNode,label=right:$M_2$](n2)at(2,1){};
\node[ltsNode,label=above:$M$](n3)at(1,2){};
\draw[-{>[scale=2.5,length=2,width=2]}](n0)to node[auto,swap]{}(n1);
\draw[-{>[scale=2.5,length=2,width=2]}](n0)to node[auto,swap]{}(n2);
\draw[-{>[scale=2.5,length=2,width=2]}](n1)to node[auto,swap]{}(n3);
\draw[-{>[scale=2.5,length=2,width=2]}](n2)to node[auto,swap]{}(n3);
\end{tikzpicture}\\
\centering
\footnotesize
Directedness of $RG(S)$
\end{minipage}
&
\begin{minipage}{0.3\linewidth}

\begin{tikzpicture}[scale=0.7]
\node[ltsNode,label=below:$M_0$](n0)at(1,0){};
\node[ltsNode,label=left:$M_1$](n1)at(0,1){};
\node[ltsNode,label=above:$M$](n3)at(1,2){};
\draw[-{>[scale=2.5,length=2,width=2]},dashed](n0)to node[auto,swap]{}(n1);
\draw[-{>[scale=2.5,length=2,width=2]}](n0)to node[auto,swap]{}(n3);
\draw[-{>[scale=2.5,length=2,width=2]}](n1)to node[auto,swap]{}(n3);
\end{tikzpicture}\\
\centering
\footnotesize
Initial directedness of $PRG(S)$
\end{minipage}
\end{tabular}

\caption{
Variants of directedness.
}
\label{FigDefDirectedness}
\end{figure}
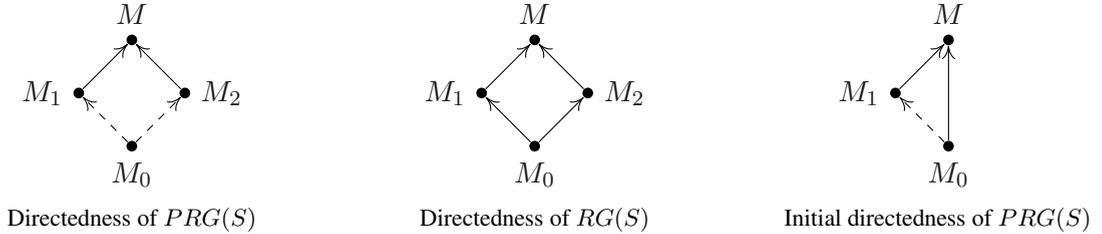

Directedness implies initial directedness, but the converse is not true, as shown in Figure~\ref{InitDirNotDir}.

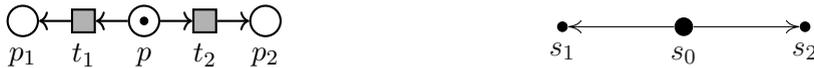
\begin{figure}[!h]
\centering
\begin{minipage}{0.45\linewidth}
\centering
\raisebox{0mm}{
\begin{tikzpicture}[scale=0.8,mypetristyle]

\node (p) at (0,0) [place,tokens=1] {};
\node (p1) at (-2,0) [place] {};
\node (p2) at (2,0) [place] {};

\node [anchor=north] at (p.south) {$p$};
\node [anchor=north] at (p1.south) {$p_1$};
\node [anchor=north] at (p2.south) {$p_2$};

\node (t1) at (-1,0) [transition,thick] {};
\node (t2) at (1,0) [transition,thick] {};

\node [anchor=north] at (t1.south) {$t_1$};
\node [anchor=north] at (t2.south) {$t_2$};

\draw [->,thick, bend left=0] (p) to node [] {} (t1);
\draw [->,thick, bend left=0] (t1) to node [] {} (p1);
\draw [->,thick, bend left=0] (p) to node [] {} (t2);
\draw [->,thick, bend right=0] (t2) to node [] {} (p2);

\end{tikzpicture}
}
\end{minipage}
%
%
\begin{minipage}{0.45\linewidth}
\centering

\begin{tikzpicture}[scale=0.8]
\node[ltsNode,minimum width=7pt,label=below:$s_0$](s0)at(0,0){};
\node[ltsNode,label=below:$s_1$](s1)at(-2,0){};
\node[ltsNode,label=below:$s_2$](s2)at(2,0){};
\draw[-{>[scale=2.5,length=2,width=2]}](s0)to node[auto,swap]{}(s1);
\draw[-{>[scale=2.5,length=2,width=2]}](s0)to node[auto,swap]{}(s2);
\end{tikzpicture}

\end{minipage}

\caption{
On the left, a system $S$. The LTS on the right represents both $RG(S)$ and $PRG(S)$.
The latter is initially directed but not directed.
}

\label{InitDirNotDir}

\end{figure}

\subsection{Directedness, (strong) liveness and HAC systems}\label{SubsecDirLiveSubclasses}

We introduce the new notion of \emph{strong liveness}.

\begin{definition}[Strong liveness]
A system $S=(N,M_0)$ is {\em strongly live} if, for each marking $M\in PR(S)$, $(N,M)$ is live.
\end{definition}

We derive next lemma. 

\begin{lemma}[Initial directedness and strong liveness]\label{LiveAndDirected}
Consider a live system $S$. 
If $PRG(S)$ is initially directed, then $S$ is strongly live.
\end{lemma}

\begin{proof}
Consider any potentially reachable marking $M_1$.
Suppose that $M_1$ is not live.
Consequently, there is a marking $M_1'$ reachable from $M_1$ at which some transition $t$ is dead.
Since $M_1'$ is potentially reachable from $M_0$
and 
since $PRG(S)$ is initially directed,
some marking $M$ is reachable from both $M_0$ and $M_1'$.
Since $M_0$ is live, $M$ is also live.
Since the live marking $M$ is reachable from $M_1'$, $t$ is not dead at the latter, a contradiction. Thus, $M_1$ is live.
\end{proof}

In the following, we recall that the converse holds in the class of live HFC systems
and
we show that it does not hold in live unit-weighted AC systems.
Later on, in Section~\ref{PropertiesH1SWMG}, we will show that the converse holds in the live H$1$S-\WMGineq{} class,
which is a subclass of the HAC systems.\\

\noindent {\bf Directedness and strong liveness of live HFC systems.} 
Each live HFC system is strongly live (Corollary~13 of~\cite{STECS}, where HFC nets are called Equal-Conflict).
Moreover, 
the potential reach\-ability graph of each live HFC system is directed (hence also initially directed), as proven in~\cite{LHS1993}.\\

\noindent {\bf Liveness of HAC systems.} 
Theorem~3.1 in~\cite{Jiao2004} states the following for any HAC system $S$:
a transition $t$ is non-live in $S$ iff some siphon $D$ is deadlocked at some reachable marking and $t$ has an input place in $D$.
Since the H$1$S systems form a subclass of the HAC systems, they benefit from this result.
Notice that if some siphon $D$ is deadlocked at some marking $M$, then $D$ includes a minimal siphon $D'$ deadlocked at $M$. 
We deduce the next reformulation.

\begin{proposition}[Liveness of HAC (and H$1$S) systems \cite{Jiao2004}]\label{LivenessOfHAC}
Consider a connected (and non-trivial) HAC system.
It is live iff no minimal siphon is deadlocked at any reachable marking.
\end{proposition}

\noindent {\bf A property of consistent systems.} 
We obtain the next simple lemma.

\begin{lemma}[Consistency and potential reachability]\label{LemPropConsistent}
For any consistent system $S=(N,M_0)$, let $M \in PR(S)$.
Then $M_0 \in PR((N,M))$.
\end{lemma}

\begin{proof}
Denote by $\mathcal{Y}$ any consistency vector of $N=(P,T,W)$ (implying $I\cdot \mathcal{Y} = 0$ and $\support(\mathcal{Y}) = T$) 
and by $I$ the incidence matrix of $N$.
Since $M \in PR(S)$, there exists a T-vector $Y$ such that $M = M_0 + I\cdot Y$.
We can choose a sufficiently large positive integer $k$ such that $k\cdot \mathcal{Y} \ge Y$.
We have: $M_0 = M - I\cdot Y + I\cdot \mathcal{Y}\cdot k = M + I\cdot (k\cdot \mathcal{Y} - Y)$.
Thus, $M_0 \in PR((N,M))$.
\end{proof}

Consequently, all the potentially reachable markings of a consistent system $S=(N,M_0)$
are mutually potentially reachable,
i.e.\ $\forall M,M' \in PR(S)$, $M' \in PR((N,M))$ and $M \in PR((N,M'))$.\\

\noindent {\bf Counter-example to the converse of Lemma~\ref{LiveAndDirected}.}
Using Figure~\ref{LBOACwoHS}, which is extracted from \cite{DesEsp} (Fig.\ 8.3, p.\ 174), 
we show that the converse of Lemma~\ref{LiveAndDirected} does not hold, even in the class of unit-weighted AC systems.
We denote the system by $S=(N,M_0)$ and proceed as follows.\\
In $S$, no minimal siphon is deadlocked at any potentially reachable marking,
allowing to use Proposition~\ref{LivenessOfHAC}, 
from which we deduce the liveness of each potentially reachable marking
hence strong liveness. Consequently, each potentially reachable marking is also strongly live.\\
We select a (strongly live) marking $M \in R(S)$, defining the strongly live system $S'=(N,M)$.
Since $S$ is consistent, Lemma~\ref{LemPropConsistent} applies: in particular, all its reachable markings are mutually potentially reachable.
Using this fact, we select a marking $M' \in R(S) \cap PR(S')$ 
such that $S'$ and $(N,M')$ do not have any reachable marking in common.
Thus, $PRG(S')$ is not initially directed.
Notice that $S'$ is even bounded and reversible.

\begin{figure}[!h]

\centering

\begin{tikzpicture}[mypetristyle,scale=0.90]

\node (p0) at (0,10.3) [place,thick,tokens=1,label=above:$p_0$] {};
\node (p1) at (0,9) [place,thick,label=above:$p_1$] {};
\node (p2) at (0,7.8) [place,thick,tokens=1,label=above:$p_2$] {};
\node (p3) at (0,7.2) [place,thick,label=below:$p_3$] {};
\node (p4) at (0,6) [place,thick,label=below:$p_4$] {};
\node (p5) at (0,4.8) [place,thick,tokens=1,label=below:$p_5$] {};
\node (p6) at (-2,2.6) [place,thick,label=below:$p_6$] {};
\node (p7) at (2,2.6) [place,thick,label=below:$p_7$] {};
\node (p8) at (-2,4.8) [place,thick,label=right:$p_8$] {};
\node (p9) at (2,4.8) [place,thick,label=left:$p_9$] {};
\node (p10) at (-2,7.5) [place,thick,label=above:$~~~~~~~p_{10}$] {};
\node (p11) at (2,7.5) [place,thick,label=above:$p_{11}~~~~~~~$] {};
\node (p12) at (-3,6) [place,thick,tokens=1,label=left:$p_{12}$] {};
\node (p13) at (3,6) [place,thick,tokens=1,label=right:$p_{13}$] {};

\node (t0) at (-2,9) [transition,thick] {};
\node (t1) at (2,9) [transition,thick] {};
\node (t2) at (-0.85,7.5) [transition,thick] {};
\node (t3) at (0.85,7.5) [transition,thick] {};
\node (t4) at (-2,3.6) [transition,thick] {};
\node (t5) at (2,3.6) [transition,thick] {};
\node (t6) at (-2,6) [transition,thick] {};
\node (t7) at (2,6) [transition,thick] {};

\node [anchor=south] at (t0.north) {$t_0$};
\node [anchor=south] at (t1.north) {$t_1$};
\node [anchor=north] at (t2.south) {$t_2$};
\node [anchor=north] at (t3.south) {$t_3$};
\node [anchor=east] at (t4.west) {$t_4$};
\node [anchor=west] at (t5.east) {$t_5$};
\node [anchor=east] at (t6.west) {$t_6$};
\node [anchor=west] at (t7.east) {$t_7$};

\draw [->,thick] (t0) to node [below] {} (p1);
\draw [->,thick] (t1) to node [below] {} (p1);
\draw [->,thick] (p0) to node [left] {} (t0);
\draw [->,thick] (p0) to node [right] {} (t1);
\draw [->,thick] (p1) to node [left] {} (t2);
\draw [->,thick] (p1) to node [right] {} (t3);
\draw [->,thick] (t2) to node [below] {} (p4);
\draw [->,thick] (t3) to node [below] {} (p4);
\draw [->,thick] (p2) to node [left] {} (t2);
\draw [->,thick] (t2) to node [left] {} (p3);
\draw [->,thick] (p3) to node [left] {} (t3);
\draw [->,thick] (t3) to node [left] {} (p2);
\draw [->,thick] (t0) to node [left] {} (p10);
\draw [->,thick] (p10) to node [left] {} (t6);
\draw [->,thick] (t6) to node [left] {} (p0);
\draw [->,thick] (t1) to node [left] {} (p11);
\draw [->,thick] (p11) to node [left] {} (t7);
\draw [->,thick] (t7) to node [left] {} (p0);
\draw [->,thick] (p4) to node [left] {} (t6);
\draw [->,thick] (p4) to node [left] {} (t7);
\draw [->,thick] (p5) to node [left] {} (t6);
\draw [->,thick] (p5) to node [left] {} (t7);
\draw [->,thick] (t4) to node [left] {} (p5);
\draw [->,thick] (t5) to node [left] {} (p5);
\draw [->,thick] (t6) to node [left] {} (p8);
\draw [->,thick] (p8) to node [left] {} (t4);
\draw [->,thick] (t7) to node [left] {} (p9);
\draw [->,thick] (p9) to node [left] {} (t5);
\draw [->,thick] (t4) to node [left] {} (p12);
\draw [->,thick] (p12) to node [left] {} (t0);
\draw [->,thick] (t5) to node [left] {} (p13);
\draw [->,thick] (p13) to node [left] {} (t1);
\draw [->,thick,bend left=90] (t0) .. controls (-5.5,7.5) and (-5,-1.5) .. (p7);
\draw [->,thick] (p7) to node [left] {} (t4);
\draw [->,thick,bend left=90] (t1) .. controls (5.5,7.5) and (5,-1.5) .. (p6);
\draw [->,thick] (p6) to node [left] {} (t5);
\draw [->,thick,bend right=30] (p6) to node [left] {} (t4);
\draw [->,thick,bend right=30] (t4) to node [left] {} (p6);
\draw [->,thick,bend right=30] (p7) to node [left] {} (t5);
\draw [->,thick,bend right=30] (t5) to node [left] {} (p7);

\end{tikzpicture}
\raisebox{2cm}{
\begin{tikzpicture}[scale=1]
\node[ltsNode,minimum width=7pt,label=right:${s_0=M_0}$](s0)at(0,5.5){};
\node[ltsNode,label=right:{$s_1=M'$}](s1)at(0,6.25){};
\node[ltsNode,label=right:{$s_2=M$}](s2)at(0,4.75){};
\node[ltsNode,label=left:$s_3$](s3)at(-1,7.25){};
\node[ltsNode,label=right:$s_4$](s4)at(1,3.75){};
\node[ltsNode,label=left:$s_5$](s5)at(-1.25,8){};
\node[ltsNode,label=right:$s_6$](s6)at(1.25,3){};
\node[ltsNode,label=left:$s_7$](s7)at(-1,8.75){};
\node[ltsNode,label=right:$s_8$](s8)at(1,2.25){};
\node[ltsNode,label=above:$s_9$](s9)at(0,9){};
\node[ltsNode,label=above:$s_{10}$](s10)at(0,9.75){};
\node[ltsNode,label=below:$s_{11}$](s11)at(0,2){};
\node[ltsNode,label=below:$s_{12}$](s12)at(0,1.25){};
\node[ltsNode,label=right:$s_{13}$](s13)at(1,8.75){};
\node[ltsNode,label=left:$s_{14}$](s14)at(-1,2.25){};
\node[ltsNode,label=right:$s_{15}$](s15)at(1.25,8){};
\node[ltsNode,label=left:$s_{16}$](s16)at(-1.25,3){};
\node[ltsNode,label=right:$s_{17}$](s17)at(1,7.25){};
\node[ltsNode,label=left:$s_{18}$](s18)at(-1,3.75){};
\node[ltsNode,label=below:$s_{19}$](s19)at(0,7){};
\node[ltsNode,label=above:$s_{20}$](s20)at(0,4){};
\draw[-{>[scale=2.5,length=2,width=2]}](s0)to node[left]{$t_1$}(s1);
\draw[-{>[scale=2.5,length=2,width=2]}](s0)to node[left]{$t_0$}(s2);
\draw[-{>[scale=2.5,length=2,width=2]}](s2)to node[right]{$~t_2$}(s4);
\draw[-{>[scale=2.5,length=2,width=2]}](s4)to node[left]{$t_6$}(s6);
\draw[-{>[scale=2.5,length=2,width=2]}](s6)to node[left]{$t_1$}(s8);
\draw[-{>[scale=2.5,length=2,width=2]}](s8)to node[right]{$~t_3$}(s12);
\draw[-{>[scale=2.5,length=2,width=2]}](s8)to node[above]{$t_4$}(s11);
\draw[-{>[scale=2.5,length=2,width=2]}](s11)to node[above]{$t_3$}(s14);
\draw[-{>[scale=2.5,length=2,width=2]}](s12)to node[left]{$t_4~$}(s14);
\draw[-{>[scale=2.5,length=2,width=2]}](s14)to node[right]{$t_7$}(s16);
\draw[-{>[scale=2.5,length=2,width=2]}](s16)to node[right]{$t_0$}(s18);
\draw[-{>[scale=2.5,length=2,width=2]}](s18)to node[below]{$t_2$}(s20);
\draw[-{>[scale=2.5,length=2,width=2]}](s20)to node[below]{$t_5$}(s4);
\draw[-{>[scale=2.5,length=2,width=2]}](s18)to node[left]{$t_5~$}(s2);
\draw[-{>[scale=2.5,length=2,width=2]}](s18)to node[below]{$t_2$}(s20);
\draw[-{>[scale=2.5,length=2,width=2]}](s20)to node[below]{$t_5$}(s4);

\draw[-{>[scale=2.5,length=2,width=2]}](s1)to node[left]{$t_2~$}(s3);
\draw[-{>[scale=2.5,length=2,width=2]}](s19)to node[above]{$t_4$}(s3);
\draw[-{>[scale=2.5,length=2,width=2]}](s17)to node[right]{$~t_4$}(s1);
\draw[-{>[scale=2.5,length=2,width=2]}](s17)to node[above]{$t_2$}(s19);
\draw[-{>[scale=2.5,length=2,width=2]}](s3)to node[right]{$t_7$}(s5);
\draw[-{>[scale=2.5,length=2,width=2]}](s5)to node[right]{$t_0$}(s7);
\draw[-{>[scale=2.5,length=2,width=2]}](s7)to node[below]{$t_5$}(s9);
\draw[-{>[scale=2.5,length=2,width=2]}](s7)to node[left]{$t_3~$}(s10);
\draw[-{>[scale=2.5,length=2,width=2]}](s10)to node[right]{$~t_5$}(s13);
\draw[-{>[scale=2.5,length=2,width=2]}](s9)to node[below]{$t_3$}(s13);
\draw[-{>[scale=2.5,length=2,width=2]}](s13)to node[left]{$t_6$}(s15);
\draw[-{>[scale=2.5,length=2,width=2]}](s15)to node[left]{$t_1$}(s17);
\end{tikzpicture}
}

\vspace*{-2cm}

\caption{
On the left, a unit-weighted, consistent, live and bounded asymmetric-choice system $(N,M_0)$,
on the right a labeled transition system representing its reachability graph, where $s_0$ represents $M_0$.
The system $(N,M_0)$ has no home state and is strongly live: 
the only minimal siphons that might be problematic are $D_1 = \{p_0, p_5, p_6\}$ and $D_2 = \{p_0, p_5, p_7\}$,
which cannot be unmarked at any potentially reachable marking.
Indeed, $(1,0,0,0,0,1,1,1,0,0,0,0,0,0)$ is a P-semiflow, thus for each marking $M$ in $PR(S)$, we have $M(p_0)+M(p_5)+M(p_6)+M(p_7) = 2$;
since $S$ is $1$-safe, we deduce that $M$ has at least one token in each siphon $D_1$ and $D_2$.
Denote by $M$ and $M'$ the markings reached by firing respectively $t_0$ or $t_1$ from $M_0$:
$M_0$ is not reachable from $M$ nor from $M'$, $[M\rangle \cap [M'\rangle = \emptyset$
but, for each pair $(M_1, M_2)$ of markings in $[M_0\rangle$, $M_2$ is potentially reachable from $M_1$ (by consistency);
in particular, $M'$ is potentially reachable from $M$.
We deduce that $(N,M)$ is not initially directed.
On the right, $s_2$ represents $M$ and $s_1$ represents $M'$.
}
\label{LBOACwoHS}
\end{figure}
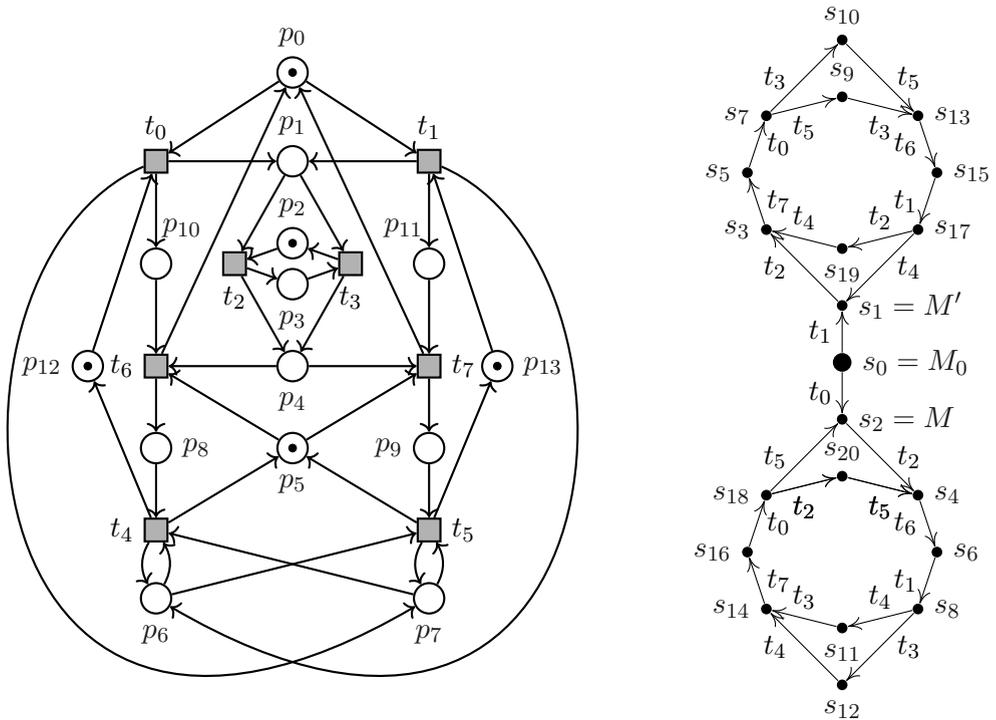

\subsection{Directedness in persistent systems}\label{SubsecPersistent}

Persistent systems have a strong restriction on their behavior: no transition firing can disable any other transition.
More formally: a system $S$ is \emph{persistent} if, for all distinct transitions $t_1,t_2$ 
and all reachable markings $M,M_1,M_2$ such that $M[t_1\rangle M_1$ and $M[t_2\rangle M_2$, 
we have $M_1[t_2\rangle$ (and~$M_2[t_1\rangle$).

Their reachability graph fulfills a kind of confluence, expressed by Keller's theorem below.
We first need to recall the notion of \emph{residues}, on which this theorem is based.

\begin{definition}[(Left) Residue]\label{residues.def}
Let $T$ be a set of labels (typically, transitions)	
and 
$\tau,\sigma\in T^*$
two sequences over this set.
The {\em (left) residue of $\tau$ with respect to $\sigma$}, 
denoted by $\tau\pminus\sigma$,
arises from cancelling successively
in $\tau$ the leftmost occurrences of all symbols from $\sigma$, read from left to right.

Inductively:
$\tau\pminus\emptyseq=\tau$;
$\tau\pminus t=\tau$ if $t\notin\support(\tau)$;
$\tau\pminus t$ is the sequence obtained by erasing the leftmost $t$ in $\tau$ if $t\in\support(\tau)$;
and 
$\tau\pminus(t\sigma)=(\tau\pminus t)\pminus\sigma$.
For instance:

$acbcacbc\pminus abbcb=cacc$ and $abbcb\pminus acbcacbc=b$.

Residues naturally extend to T-vectors as follows:
for any sequence $\tau$ and T-vector $Y$,
$\tau \pminus Y$ is~$\tau$ in which,
for each transition $t$ in $\support(Y)$,
the $\min\{P(\tau)(t),Y(t)\}$ leftmost occurrences of $t$ have been removed.
\end{definition}

\begin{theorem}[Keller \cite{keller}]\label{kellersTheorem}
Let $S$ be a persistent system. 
Let $\tau$ and $\sigma$ be two sequences feasible in $S$.
Then $\tau(\sigma\pminus\tau)$ and $\sigma(\tau\pminus\sigma)$ are both feasible in $S$ and lead to the same marking.
\end{theorem}

Keller's theorem applies to CF systems, since they are structurally persistent (each place having at most one output).
We will prove a variant of this theorem for H$1$S-\WMGineq{} systems in the sequel, embodied by Theorem~\ref{OnePlaceCommonReachability}.

Notice that each CF net is HFC, hence for each live CF system $S$, $PRG(S)$ is directed~\cite{STECS}.

\section{Liveness and reachability in \texorpdfstring{H$1$S systems}{}}\label{SecLiveness}

In Subsection~\ref{LiveReachCF}, we provide a general result on liveness and reachability for all weighted Petri nets, 
and recall previous related results in the CF subclass.

In Subsection~\ref{SecH1SWMG}, 
we provide new results on liveness and reachability in the H$1$S-\WMGineq{} subclass,
including a new (non-)liveness characterization based on the state equation and inequalities describing the potentially reachable deadlocks.
This trims down to co-NP the complexity of liveness checking in this class.
Moreover, with a counter-example, we show that this characterization does not extend to the whole H$1$S class and
is thus tightly related to the H$1$S-\WMGineq{} structure.

\subsection{Liveness and reachability in Petri nets and the CF subclass}\label{LiveReachCF} 

We shall need the next condition based on Dickson's lemma~\cite{Dickson1919}.


\begin{lemma}[One direction of Proposition~\ref{LivenessOfCFnets} for weighted Petri nets]\label{UsingDicksonsLemma}
Let $(N,M_0)$ be a Petri net system with incidence matrix $I$. 
If it is live, then there exists a marking $M \in R((N,M_0))$
such that
some firing sequence $\sigma$ is feasible in $(N,M)$ with $\Parikh(\sigma) \ge \one$ and $I \cdot \Parikh(\sigma) \ge \zero$.
\end{lemma}

\begin{proof}
By liveness, from each reachable marking some sequence is feasible that contains all transitions.
Thus, $(N,M_0)$ enables an infinite sequence of the form $M_0[\tau_0\rangle M_1[\tau_1\rangle \ldots M_i[\tau_i\rangle \ldots $
in which each $\tau_i$ is finite and contains all transitions.
Since all marking components take their values in the non-negative integers,
Dickson's lemma~\cite{Dickson1919} applies,
so that two distinct non-negative integers $i$ and $j$ exist such that $j > i$ and $M_j \ge M_i$.
We then have: $M=M_i$, $\sigma = \tau_i \ldots \tau_{j-1}$, $\Parikh(\sigma) \ge \one$ and $I \cdot \Parikh(\sigma) = M_j - M_i \ge \zero$.
\end{proof}

The other direction does not hold in all weighted Petri nets;
however, it does in the CF subclass, forming a liveness characterization, as recalled next
and 
given as Corollary~$4$ in~\cite{tcs97}.

\begin{proposition}[Liveness of CF systems \cite{tcs97}]\label{LivenessOfCFnets}
Let $(N,M_0)$ be a CF system with incidence matrix $I$. 
It is live iff there exists a marking $M \in R((N,M_0))$
such that
a firing sequence $\sigma$ is feasible in $(N,M)$ with $\Parikh(\sigma) \ge \one$ and $I \cdot \Parikh(\sigma) \ge~\zero$.
\end{proposition}

We shall also need the following property on fireable T-vectors in \WMGineq{} given as Corollary~$1$ in~\cite{DH18}.

\begin{proposition}[Fireable T-vectors in \WMGineq{} \cite{DH18}]\label{RealizableTvectors}
Let $N = (P,T,W)$ be a \WMGineq{} with incidence matrix $I$. 
Let $M_0$ and $M$ be markings over $P$ and $Y \in \N^T$ be a T-vector such that $M = M_0 + C \cdot Y \ge \zero$. 
Let $\sigma$ be a transition sequence such that $Y \le \Parikh(\sigma)$.
Then, if $M_0 [\sigma\rangle$, there is a firing sequence $M_0 [\sigma'\rangle M$ such that $\Parikh(\sigma') = Y$.
\end{proposition}

\subsection{Using the state equation in \texorpdfstring{H$1$S-\WMGineq{}}{} systems}\label{SecH1SWMG}

In this subsection, we develop a result on liveness and reachability in a subclass of the H$1$S systems,
namely the H$1$S-\WMGineq{} systems, 
in which the deletion of the shared place (if any) yields a \WMGineq{}.
With the help of a counter-example, we show that this result does not extend to the entire H$1$S class.

Then, we focus on the H$1$S-\WMGineq{} subclass under the following restrictions:
we assume the Petri net system to be strongly connected, structurally bounded with an underlying strongly connected \WMGineq{}
(the latter being thus a WMG).
We derive from our result a new characterization of liveness expressed in terms of potentially reachable deadlocks.
We deduce an ILP defined by the state equation and an additional set of linear inequalities,
the whole inequality system having a size polynomial in the initial Petri net encoding.
Thus, we obtain a new liveness checking method whose complexity lies in co-NP,
reducing the complexity drastically (comparing to EXPSPACE-hardness).

Notice that the H$1$S-\WMGineq{} class is expressive enough to model our use-cases in Section~\ref{UseCase}.

\subsubsection{Properties of the \texorpdfstring{H$1$S-\WMGineq{}} subclass}\label{PropertiesH1SWMG}

We obtain next theorem, illustrated in Figure~\ref{FigOnePlaceCommonReachabilityTheorem},
which applies to the H$1$S-\WMGineq{} class. 
It can be seen as a variant of Keller's theorem.
Its proof is illustrated in Figure~\ref{FigOnePlaceCommonReachability}.
We then show it does not extend to the H$1$S class.

\begin{theorem}[Liveness and reachability in H$1$S-\WMGineq{}]\label{OnePlaceCommonReachability}
Consider a live H$1$S-\WMGineq{} $S=(N,M_0)$. 
For any T-vector $Y$ and marking $M$ forming a solution to the state equation $M_0 + I \cdot Y = M$,
there exist a marking $M'$ and a firing sequence $M_0[\sigma\rangle M'
$
such that 
$\Parikh(\sigma) \ge Y$
and
$M'$ is also reached by firing $\sigma \pminus Y$ from $M$.
Consequently, $PRG(S)$ is initially directed and $(N,M)$ is live, thus $S$ is strongly live.
%
%
\end{theorem}

\begin{figure}[!h]

\centering


\hspace*{1cm}
\begin{tikzpicture}[scale=0.6,mypetristyle]
\node[ltsNode,label=below left:Live~$M_0$](n0)at(0,0){};
\node[ltsNode,label=below right:$M$~is~thus~live](n1)at(3,0){};
\node[ltsNode,label=above:$M'$](n2)at(1.5,2.5){};
\draw[-{>[scale=2.5,length=2,width=2]},bend left=0](n0)to node[auto,swap,above left,near start]{$\sigma$}(n2);
\draw[-{>[scale=2.5,length=2,width=2]},bend left=0](n1)to node[auto,swap,above right,near start]{$\sigma \pminus Y$}(n2);
\draw[-{>[scale=2.5,length=2,width=2]},dashed](n0)to node[auto,swap,below]{$Y$}(n1);
\end{tikzpicture}


\caption{
Illustration of the claim of Theorem~\ref{OnePlaceCommonReachability}: for any such T-vector $Y$, there exists $\sigma$ leading to some $M'$ such that $\Parikh(\sigma) \ge Y$
and $\sigma \pminus Y$ leads to $M'$ from $M$.
}
\label{FigOnePlaceCommonReachabilityTheorem}
\end{figure}
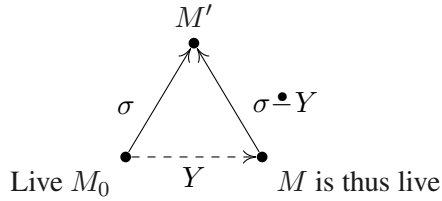

\begin{proof}
In case there is no shared place, $S$ is a live \WMGineq{}, hence one is able to build a sequence $\sigma_1$ feasible in $S$
such that $\Parikh(\sigma_1) \ge Y$.
By applying Proposition~\ref{RealizableTvectors}, we get a sequence $\sigma'$ such that $\Parikh(\sigma') = Y$ and $M_0[\sigma'\rangle M$.
We conclude by taking $\sigma = \sigma'$ and $M' = M$.
%
%
In the rest of the proof, let us assume the existence of one shared place $p$.

Consider a Parikh vector $Y$ such that $M_0 + I \cdot Y = M$,
where $I$ is the incidence matrix of the system.
We prove the first claim by induction on the sum $n$ of the components of~$Y$.

\noindent $-$ Base case: $n=0$, thus $M_0 = M$ is a common reachable marking,
with $\sigma = \epsilon$,
$\sigma \pminus Y = \epsilon$ and $\Parikh(\sigma) \ge Y = \zero$.

\noindent $-$ Inductive case: $n > 0$. Suppose the property to be true for $n-1$.
If some transition $t$ in $\support(Y)$ is enabled at $M_0$, leading to $M_0'$, which is also live,
we have
$M_0' + I \cdot (Y-\one_t) = M$.
Hence the induction hypothesis applies to $M_0'$, $M$ and $Y-\one_t$:
a sequence $\sigma'$ is feasible at $M_0'$, leading to a marking $M'$
such that $\sigma' \pminus (Y-\one_t)$ is feasible at $M$
and leads to $M'$, with $\Parikh(\sigma') \ge Y-\one_t$.
Denoting $\sigma = t \sigma'$, which is feasible at $M_0$,
the sequence $\sigma \pminus Y = \sigma' \pminus (Y-\one_t)$
is feasible at $M$ and leads to $M'$,
hence the first claim.

Now, suppose that no transition of $\support(Y)$ is feasible at $M_0$.
Since $S$ is live, there exist $t \in \support(Y)$ and a sequence $\sigma_t$ feasible from $M_0$ 
that leads to $M_t$ enabling $t$,
such that no transition in $\sigma_t$ belongs to $\support(Y)$.
We show next that such a $\sigma_t$ exists that contains only non-choice transitions (i.e.\ transitions that are not outputs of $p$),
from which we deduce that $\sigma_t$ is also feasible from $M$.

Since $S$ is live and by Lemma~\ref{UsingDicksonsLemma}, 
a sequence $\tau \tau'$ is feasible in $S$, leading to $M_{\tau'}$, $\tau$ leading to $M_\tau$,
such that $\Parikh(\tau') \ge \one$ and $I \cdot \Parikh(\tau') \ge \zero$.
Denote by $S_{\textrm{WMG}}$ the \WMGineq{} P-subsystem obtained by deleting $p$ in $S$.
Since $S_{\textrm{WMG}}$ is less constrained than $S$ with the same set of transitions,
every sequence feasible in $S$ is also feasible in $S_{\textrm{WMG}}$.
In particular, the sequence $\tau \tau'$ is feasible in $S_{\textrm{WMG}}$, which is consequently live by Proposition~\ref{LivenessOfCFnets}.

Now, let us apply Proposition~\ref{RealizableTvectors} to $S_{\textrm{WMG}}$:
it is live,
thus enables a sequence $\gamma$ such that $\Parikh(\gamma) \ge Y$,
hence
a sequence $\sigma_Y$ is feasible in $S_{\textrm{WMG}}$ such that $\Parikh(\sigma_Y) = Y$.
But we assumed that no transition of $\support(Y)$ is feasible in $S$.
Let us denote by $t_0$ the first transition in $\sigma_Y$.
We deduce that
$t_0$ has all its input places enabled at $M_0$ except $p$,
and that no output of $p$ is enabled by $M_0$ since $S$ is homogeneous.
Now, let us consider any (finite) feasible sequence of the form $\sigma_t t$ (where $t$ belongs to $\support(Y)$, as defined earlier)
such that $\sigma_t$ contains only non-choice transitions, none of which belongs to $\support(Y)$:
such a sequence exists since either $t$ equals $t_0$ (once $p$ is enabled, by liveness and homogeneity) 
or
$t$ is another transition enabled before enabling $p$.

We proved that $t \in \support(Y)$ and $\sigma_t$ exist
such that $\sigma_t$ is feasible in $S$,
contains only non-choice transitions
that are not in $\support(Y)$,
and leads to $M_t$ which enables $t$, leading to $M_1$.
Thus,
for each input place $p_i$ of each transition of $\sigma_t$, $M(p_i) \ge M_0(p_i)$,
from which we deduce that $\sigma_t$ is also feasible from $M$, 
leading to $M_t'$.

Now,
we apply the induction hypothesis to $(N,M_1)$ (which is live),
to $Y-\one_t$ and to $M_t'$:
there exist $\sigma'$ feasible from $M_1$ leading to $M'$
such that
$\sigma' \pminus (Y-\one_t)$ is feasible at $M_t'$ and leads to $M'$,
with 
$\Parikh(\sigma') \ge (Y-\one_t)$.
We deduce that $\sigma = \sigma_t t \sigma'$ (with $\Parikh(\sigma) \ge Y$) is feasible at $M_0$ and leads to $M'$,
while
$\sigma \pminus Y = (\sigma_t t \sigma') \pminus Y = \sigma_t (\sigma' \pminus (Y-\one_t))$
is feasible at $M$ and leads to $M'$, with $\Parikh(\sigma) \ge Y$.

We proved the inductive step. Hence the first claim.
Initial directedness is deduced immediately,
and
liveness of each potentially reachable marking $M$ comes from Lemma~\ref{LiveAndDirected}.
\end{proof}

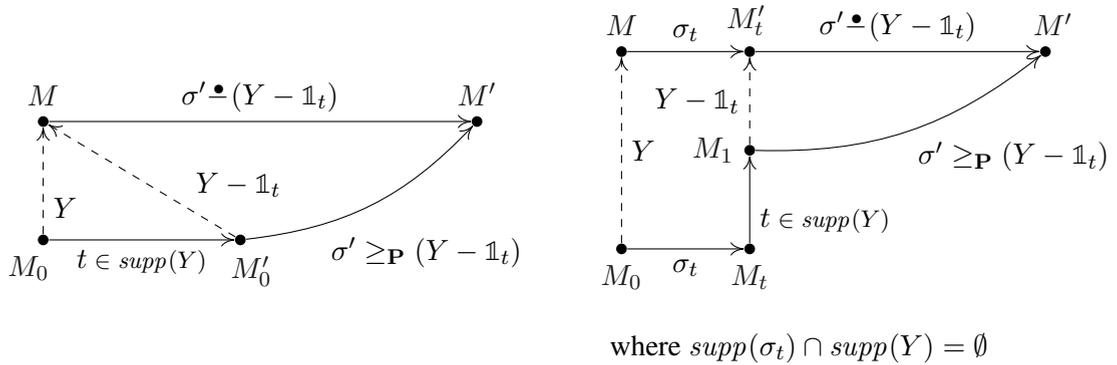
\begin{figure}[!h]
\centering

\begin{minipage}{0.48\linewidth}
\centering
\begin{tikzpicture}[scale=1.3,mypetristyle]
%
\node[ltsNode,label=below:$M_0~~~~$](m0)at(-1,0){};
\node[ltsNode,label=above:$M$](m)at(-1,1.2){};
\node[ltsNode,label=below:$~~~M_0'$](m0')at(1,0){};
\node[ltsNode,label=above:$M'$](m')at(3.4,1.2){};
%
%
\draw[-{>[scale=2.5,length=2,width=2]},dashed](m0)to node[right, near start]{$Y$}(m);
\draw[-{>[scale=2.5,length=2,width=2]}](m0)to node[below]{$t$ {\footnotesize $\in \support(Y)$}}(m0');
\draw[-{>[scale=2.5,length=2,width=2]},dashed,bend right=0](m0')to node[above right, near start]{$Y-\one_t$}(m);
\draw[-{>[scale=2.5,length=2,width=2]},bend right=20](m0')to node[below right, near start]{$~\sigma' \ge_\Parikh (Y-\one_t)$}(m');
\draw[-{>[scale=2.5,length=2,width=2]},bend left=0](m)to node[above]{$\sigma' \pminus (Y-\one_t)$}(m');
\end{tikzpicture}

\end{minipage}
\hspace*{1mm}
\begin{minipage}{0.48\linewidth}
\centering
\begin{tikzpicture}[scale=1.3,mypetristyle]
%
\node[ltsNode,label=below:$M_0$](m0)at(0.2,0){};
\node[ltsNode,label=above:$M$](m)at(0.2,2){};
\node[ltsNode,label=below:$M_t$](mt)at(1.5,0){};
\node[ltsNode,label=left:$M_1$](m1)at(1.5,1){};
\node[ltsNode,label=above:$M_t'~$](mt')at(1.5,2){};
\node[ltsNode,label=above:$~~M'$](m')at(4.5,2){};
\node[](lSuppSigma)at(2,-1){where $\support(\sigma_t) \cap \support(Y) = \emptyset$};
%
%
\draw[-{>[scale=2.5,length=2,width=2]},dashed](m0)to node[right]{$Y$}(m);
\draw[-{>[scale=2.5,length=2,width=2]},dashed](m1)to node[left]{$Y-\one_t$}(mt');
\draw[-{>[scale=2.5,length=2,width=2]}](m0)to node[below]{$\sigma_t$}(mt);
\draw[-{>[scale=2.5,length=2,width=2]}](m)to node[above]{$\sigma_t$}(mt');
\draw[-{>[scale=2.5,length=2,width=2]}](mt)to node[right,near start]{$t$ {\footnotesize $\in \support(Y)$}}(m1); 
\draw[-{>[scale=2.5,length=2,width=2]},bend right=20](m1)to node[below right]{$\sigma' \ge_\Parikh (Y-\one_t)$}(m');
\draw[-{>[scale=2.5,length=2,width=2]},bend left=0](mt')to node[above]{$\sigma' \pminus (Y-\one_t)$}(m');
\end{tikzpicture}

\end{minipage}
%


\caption{
Illustration of the proof of Theorem~\ref{OnePlaceCommonReachability}. The notation $\sigma \ge_\Parikh Y$ denotes the sequence $\sigma$ with $\Parikh(\sigma) \ge Y$.
For two cases considered in the proof,
part of the potential reachability graph of the system is depicted.
On the left, we consider the simple case in which
a transition $t$ of $\support(Y)$ is enabled by $M_0$.
We have $\sigma' \pminus (Y-\one_t) = (t\sigma') \pminus Y = \sigma \pminus Y$. 
On the right,
we suppose that $M_0$ does not enable any transition in $\support(Y)$,
and 
we depict the inductive step.
}
\label{FigOnePlaceCommonReachability}
\end{figure}

\noindent {\bf Non-extensibility of Theorem~\ref{OnePlaceCommonReachability} to the H$1$S class.}
The proof of Theorem~\ref{OnePlaceCommonReachability} uses Proposition~\ref{RealizableTvectors} 
which is not true in CF systems (as detailed in~\cite{DH18}).
Actually, we show that Theorem~\ref{OnePlaceCommonReachability} cannot be extended to the entire H$1$S class:
we build a counter-example in Figure~\ref{CEH1SCF}.

\begin{figure}[!h]
\centering
\begin{tikzpicture}[scale=0.6,mypetristyle]

\node (p1) at (0,3) [place,tokens=2] {};
\node [anchor=east] at (p1.west) {$p_1$};

\node (p2) at (3,3) [place,tokens=2] {};
\node [anchor=south] at (p2.north) {$p_2$};

\node (p3) at (1.5,1.5) [place] {};
\node [anchor=east] at (p3.west) {$p_3$};

\node (p4) at (3,1.5) [place,tokens=1] {};
\node [anchor=north] at (p4.south) {$p_4$};

\node (p5) at (4.5,0) [place] {};
\node [anchor=west] at (p5.east) {$p_5$};

\node (t1) at (1.5,3) [transition,thick] {};
\node (t2) at (4.5,3) [transition,thick] {};
\node (t3) at (0,0) [transition,thick] {};

\node [anchor=south] at (t1.north) {$t_1$};
\node [anchor=west] at (t2.east) {$t_2$};
\node [anchor=east] at (t3.west) {$t_3$};

\draw [->,thick] (t3) to node [left] {} (p1);
\draw [->,thick] (t3) to node [left] {} (p4);
\draw [->,thick] (p4) to node [left] {} (t2.south west);
\draw [->,thick,bend right=20] (p2) to node [] {} (t2.west);
\draw [->,thick,bend right=20] (t2) to node [above] {$2$} (p2);
\draw [->,thick] (t2) to node [right] {} (p5);
\draw [->,thick] (p5) to node [right] {} (t3);
\draw [->,thick] (p2) to node [above] {} (t1);
\draw [->,thick,bend left=20] (p1) to node [above] {$2$} (t1);
\draw [->,thick,bend left=20] (t1) to node [above] {} (p1);
\draw [->,thick] (t1) to node [above] {} (p3);
\draw [->,thick] (p3) to node [above] {} (t3);
\end{tikzpicture}
\hspace*{1cm}
\begin{tikzpicture}[scale=0.6,mypetristyle]

\node (p1) at (0,3) [place] {};
\node [anchor=east] at (p1.west) {$p_1$};

\node (p2) at (3,3) [place] {};
\node [anchor=south] at (p2.north) {$p_2$};

\node (p3) at (1.5,1.5) [place,tokens=2] {};
\node [anchor=east] at (p3.west) {$p_3$};

\node (p4) at (3,1.5) [place,tokens=1] {};
\node [anchor=north] at (p4.south) {$p_4$};

\node (p5) at (4.5,0) [place] {};
\node [anchor=west] at (p5.east) {$p_5$};

\node (t1) at (1.5,3) [transition,thick] {};
\node (t2) at (4.5,3) [transition,thick] {};
\node (t3) at (0,0) [transition,thick] {};

\node [anchor=south] at (t1.north) {$t_1$};
\node [anchor=west] at (t2.east) {$t_2$};
\node [anchor=east] at (t3.west) {$t_3$};

\draw [->,thick] (t3) to node [left] {} (p1);
\draw [->,thick] (t3) to node [left] {} (p4);
\draw [->,thick] (p4) to node [left] {} (t2.south west);
\draw [->,thick,bend right=20] (p2) to node [] {} (t2.west);
\draw [->,thick,bend right=20] (t2) to node [above] {$2$} (p2);
\draw [->,thick] (t2) to node [right] {} (p5);
\draw [->,thick] (p5) to node [right] {} (t3);
\draw [->,thick] (p2) to node [above] {} (t1);
\draw [->,thick,bend left=20] (p1) to node [above] {$2$} (t1);
\draw [->,thick,bend left=20] (t1) to node [above] {} (p1);
\draw [->,thick] (t1) to node [above] {} (p3);
\draw [->,thick] (p3) to node [above] {} (t3);
\end{tikzpicture}
\caption{
The H$1$S system $(N,M_0)$ on the left is live and reversible, with $M_0 = (2,2,0,1,0)$.
Let us define $M = I \cdot Y + M_0$ where $I$ is the incidence matrix of $N$, $Y=(2,0,0)^T$ and $M=(0,0,2,1,0)$.
The H$1$S system $(N,M)$ is pictured on the right: it is deadlocked.
Intuitively, allowing a place to have several inputs yields for example the situation in place $p_1$:
the effect of $t_1$ on $p_1$ is $-1$, but it needs $2$ tokens in $p_1$ to fire,
an information non encoded in the incidence matrix nor in the state equation.
}
\label{CEH1SCF}
\end{figure}
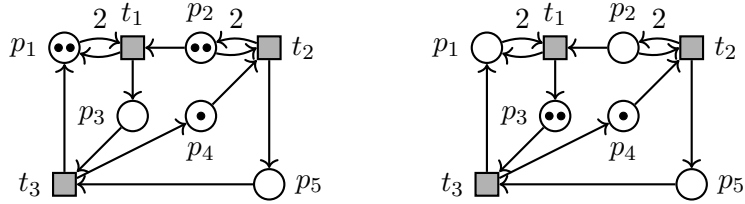

We deduce the following liveness characterization,
which is a variant of Proposition~\ref{LivenessOfHAC} in the H$1$S-\WMGineq{} subclass
exploiting the state equation instead of the reachable markings.

\begin{corollary}[Liveness characterization in H$1$S-\WMGineq{}]\label{LiveCharH1SWMG}
Let $S=(N,M_0)$ be an H$1$S-\WMGineq{} with incidence matrix $I$.
$S$ is live iff there is no solution $(M,Y)$ to its state equation $M = I \cdot Y + M_0$ such that some siphon of $N$ is deadlocked at $M$.
\end{corollary}

\begin{proof}
\noindent $(\Rightarrow)$ If $S$ is live, then for each solution $(M,Y)$, $(N,M)$ is live by Theorem~\ref{OnePlaceCommonReachability},
thus no siphon of $N$ is deadlocked at $M$.

\noindent $(\Leftarrow)$ Suppose there is no solution $(M,Y)$ to the state equation $M = I \cdot Y + M_0$ 
such that some siphon of $N$ is deadlocked at $M$.
In particular, no siphon is deadlocked at any marking reachable in $S=(N,M_0)$.
Since $S$ is an HAC system, it is live by Proposition~\ref{LivenessOfHAC}.
\end{proof}

We need the next property of strongly connected \WMGineq{}, which are WMG.
It is deduced from Theorem~4.10.3 in~\cite{WTS92}.

\begin{proposition}[Deadlocks, liveness and feasible sequences in WMG]\label{WMGnoInfiniteExec}
Let $S=(N,M_0)$ be a strongly connected \WMGineq{}. 
$S$ is deadlock-free iff it is live.
Moreover, if $S$ enables an infinite sequence $\sigma$, then $\sigma$ contains an infinite number of occurrences of each transition.
\end{proposition}

\begin{proof}
The first claim is Theorem~4.10.3 in~\cite{WTS92}.
Now, if some transition $t$ is fired an infinite number of times in some feasible sequence~$\sigma$, 
then each transition in $\lbul (\lbul t)$ is fired an infinite number of times,
and by strong connectedness each transition of $N$ is fired an infinite number of times in $\sigma$, hence the second claim.
\end{proof}

We derive next corollary.

\begin{corollary}[Checking liveness in a strongly connected H$1$S-\WMGineq{}]\label{CheckLivenessH1SWMG}
Let $S=(N,M_0)$ be a (non-trivial) strongly connected H$1$S-\WMGineq{} with incidence matrix $I$
such that the deletion of the shared place (if any) yields a strongly connected~\WMGineq{}.
Then, $S$ is live iff there is no solution $(M,Y)$ to its state equation $M = I \cdot Y + M_0$ such that $(N,M)$ is deadlocked.
\end{corollary}

\begin{proof}
We show next that for each marking $M$, if some (non-empty) siphon $D$ is deadlocked at $(N,M)$, then $(N,M)$ is deadlockable.

The transition(s) with an input place or an output place in $D$ cannot be fired anymore 
(by definition of deadlocked siphons),
and
there is at least one such transition since $N$ is strongly connected with at least one place and one transition.

If $N$ is a (strongly connected) \WMGineq{}, Proposition~\ref{WMGnoInfiniteExec} applies and a deadlock is thus reachable in $(N,M)$. 
Otherwise, suppose in the following that $N$ has one shared place $p$.
We have two cases:

\noindent $-$ $p$ belongs to $D$: 
we construct a new system $(N',M')$ where $N'=(P',T,W')$, as follows.
Denote by $t_{1,p} \ldots t_{m,p}$ the outgoing transitions of $p$ 
and
by $t_{1,p}' \ldots t_{n,p}'$ the ingoing transitions of~$p$;
define $u = \min(m,n) \ge 1$ and $v = \max(m,n)\ge 1$,
delete $p$ (with its adjacent arcs) and add $v$ places $p_{1,p} \ldots p_{v,p}$ with the following unit-weighted arcs:
for each $i\in \{1,v\}$, we set $W'(x,p_{i,p})=1$ where $x = t_{i,p}'$ if $i \le n$, $x = t_{n,p}'$ otherwise;
for each $i\in \{1,v\}$, we set $W'(p_{i,p},x)=1$ where $x = t_{i,p}$ if $i \le m$, $x = t_{m,p}$ otherwise.
Let us define $NewPlaces = \{p_{1,p} \ldots p_{v,p}\}$.
For each place $p'$, we set $M'(p') = 0$ if $p' \in NewPlaces$ and $M'(p') = M(p')$ otherwise.
Hence, $P' = (P\setminus \{p\})\cup NewPlaces$.

Clearly, $(N',M')$ is a strongly connected \WMGineq{} with a deadlocked siphon $D' = (D \setminus \{p\}) \cup NewPlaces$. 
By Proposition~\ref{WMGnoInfiniteExec}, 
$(N',M')$ enables some sequence $\sigma_d$ that leads to some deadlock $M_d''$
and 
that does not contain any transition having an input or an output in $D'$.
Consequently, by construction, $\sigma_d$ is feasible in $(N,M)$ since it contains no occurrence of transitions 
having $p$ as an input or an output,
leading to the marking $M_d$ such that
for each place $p_D$ in $D$, $M_d(p_D) = M(p_D)$ and for each other place $p_{\overline{D}}$,
$M_d(p_{\overline{D}}) = M_d''(p_{\overline{D}})$.
We deduce that $M_d$ is a deadlock reachable in $(N,M)$.

\noindent $-$ $p$ does not belong to $D$:
thus, the strongly connected \WMGineq{} P-subsystem obtained from $(N,M)$
by deleting $p$ contains the deadlocked siphon $D$
and has no infinite feasible sequence by Proposition~\ref{WMGnoInfiniteExec}.
Consequently,
since $(N,M)$ is more constrained than this P-subsystem, it does not enable any infinite sequence either,
hence has some reachable deadlock.

Under the given assumptions,
we deduce that checking the existence of a solution $(M_d,Y)$ to the state equation $M_d = I \cdot Y + M_0$ 
such that $M_d$ is a deadlock
is equivalent to checking the existence of a solution $(M,Y)$ to $M = I \cdot Y + M_0$ 
such that some siphon of $N$ is deadlocked at~$M$.
Applying Corollary~\ref{LiveCharH1SWMG}, we deduce the claim.
\end{proof}

The class of nets studied in Corollary~\ref{CheckLivenessH1SWMG} is expressive enough to model our use-cases in Section~\ref{UseCase}.

\subsubsection{ILP of poly-size for checking liveness of structurally bounded \texorpdfstring{H$1$S-WMG}{}}\label{SubsecSmallILP}

In what follows, 
exploiting the work of~\cite{LHS1993}, we provide an ILP that encodes the non-liveness of any structurally bounded H$1$S-WMG system 
satisfying the conditions of Corollary~\ref{CheckLivenessH1SWMG}.
We consider the decision version (without objective function) of the classical ILP problem, 
i.e. the problem of deciding the existence of a solution to a given integer linear system.
More precisely, 
we show that the system is live if and only if the associated ILP is infeasible (i.e.\ has no solution).
The number of inequalities of this ILP, as well as 
the length of each of its inequalities and its number of variables,
are linear in the number of places, transitions 
and the number of bits in the largest binary-encoded number (among the weights and upper bounds on structural bounds, defined in what follows).
This trims down the complexity of liveness checking to co-NP, 
the problem description being given as the Petri net system or as the corresponding ILP.
Notice that structural boundedness can be checked in polynomial-time using its following characterization, 
extracted from~\cite{LAT98}:

\begin{proposition}[Characterization of (non-)structural boundedness (Corollary~16 in~\cite{LAT98})]
A net with incidence matrix $I$ is not structurally bounded iff there exists a vector $Y \gneq \zero$ of rational numbers 
such that $I \cdot Y \gneq \zero$.
\end{proposition}

By Corollary~\ref{CheckLivenessH1SWMG}, to check liveness, we only need to check that no potentially reachable marking is a deadlock,
under the given assumptions.

So as to design the ILP, we extract several notions and results from~\cite{LHS1993} and adapt them to the H$1$S class.
Notably, we propose a variant of a transformation given in~\cite{LHS1993},
transforming the given H$1$S system into another H$1$S system on which the ILP is defined.
We will show that the system obtained has a size polynomial in the size of the original one
and preserves several of its properties, including deadlockability and liveness.
Hence, checking liveness of the new system provides the answer for the original system, the problem remaining in co-NP.\\

Consider any weighted system $S=(N,M_0)$. The following expression describes the fact that $M$ 
is a deadlock potentially reachable from $M_0$:
\begin{equation}\label{SystemEq1} 
(M=M_0+I\cdot Y) \land \left(\bigwedge\limits_{t\in T} \left(\bigvee\limits_{p\in \lbul t} M(p) < W(p,t)\right)\right)
\end{equation}
 
Due to the disjunction, condition~\eqref{SystemEq1} is a set of $\Pi_{t\in T} |\lbul t|$ linear systems.
In the case of HFC nets, this number can be reduced to a single system of linear inequalities,
the number of which is linear in the number of places and transitions,
each inequality length being also linear in the number of places and transitions as well as the encoding size of numbers~\cite{LHS1993}.
To achieve it, the authors of~\cite{LHS1993} provide two system transformations that preserve deadlockability
and simplify the expression of condition~\eqref{SystemEq1},
exploiting the notion of \emph{structural bound}, defined as follows.

The \emph{structural bound} of a place $p$ in any weighted system $S$ is defined as follows:
$$SB(p,S) \eqdef \max\{M(p) \, | \, M \in PR(S)\}.$$

Notice that in each structurally bounded system, each place has a structural bound.
Let us recall the next sufficient condition of linearity for the non-fireability condition of a transition.

\begin{proposition}[Sufficient condition of linearity for non-fireability: Theorem~5.4 in~\cite{LHS1993}]\label{PropTh5.4}
Consider any system $S$ in which each place has a structural bound.
Let $t$ be a transition such that $\lbul t = \pi \cup p'$ and $\forall p \in \pi$: $SB(p,S) \le W(p,t)$.
The non-enabledness of $t$ at some marking $M \in PR(S)$ can be written as the following integer linear inequality:
$$SB(p',S) \sum_{p\in\pi} M(p) + M(p') < SB(p',S) \sum_{p \in \pi} W(p,t)+W(p',t).$$
\end{proposition}

When this sufficient condition applies, condition~\eqref{SystemEq1} 
can be rewritten as an expression of length linear in $|P|\cdot |T|\cdot m$, 
where $m$ is the length of the binary encoding of the largest value among the structural bounds and the weights.

\noindent {\bf Using upper bounds encoded with a polynomial number of bits.}
For our purpose, it is sufficient to replace each structural bound with an upper bound whose binary encoding  
has a number of bits polynomial in the input length.
Such upper bounds can be computed in weakly polynomial-time using linear programming over the rational numbers~\cite{LAT98},
hence we suppose in the following that such upper bounds belong to the input of the problem.\\

\noindent {\bf Transformation $\Theta$ for weighted systems.} 
Taking inspiration from~\cite{LHS1993}, we propose a transformation $\Theta$ applying to any weighted system $S$, 
yielding another system $S^\Theta$ with polynomial increase in size 
while preserving deadlockability, structural boundedness and other properties.
We will apply it to $1$S systems, i.e.\ H$1$S without the homogeneity constraint.
In particular, if $S$ is a structurally bounded $1$S system, then $S^\Theta$ fulfills all conditions of Proposition~\ref{PropTh5.4}.
This allows to express deadlocks of $PR(S^\Theta)$ 
with an integer linear inequality system of size linear in $|P^\Theta|\cdot |T^\Theta|\cdot m$, 
with $N^\Theta = (P^\Theta,T^\Theta,W^\Theta)$ where $T^\Theta$ has size at most $|T|+|P|$.
Hence the encoding size of the ILP is linear in $|P|\cdot (|T|+|P|) \cdot m$.\\ 
This transformation is described by Algorithm~\ref{AlgoTheta}, which clearly terminates.
It consists in applying transformation $\Theta_2$ of~\cite{LHS1993} to each pair $(p,t)$ such that $p\lbul = \{t\}$
and $t$ has at least two inputs;
it is illustrated in Figures~\ref{FigAlgoTheta},~\ref{FigAlgoThetaH1S} and~\ref{FigAlgoThetaSwimming}.
Properties of the transformation are then stated in Theorem~\ref{ThTransfoDeadlock}.

\begin{algorithm}[!h]\label{AlgoTheta}
   \KwData{A system $S=(N,M_0)$.}
   \KwResult{The transformed system $S^\Theta=(N^\Theta,M_0^\Theta)$ where $N^\Theta=(P^\Theta,T^\Theta,W^\Theta)$.}
	$(N^\Theta,M_0^\Theta) := (N,M_0)$\;
    \ForAll{$(p,t)$ such that $p\lbul = \{t\}$ and $t$ has at least two inputs}{
		Add to $P^\Theta$ two places $p_a^{(p,t)}$ and $p_b^{(p,t)}$ with 
		$M_0^\Theta(p_a^{(p,t)}) := 0$ and  $M_0^\Theta(p_b^{(p,t)}) := 1$\;
		Add to $T^\Theta$ the transition $t_p^{(p,t)}$\;
		Add to $W^\Theta$ the following arcs:
		$W^\Theta(p,t_p^{(p,t)}) := W^\Theta(p,t)$\;
		$W^\Theta(t_p^{(p,t)},p_a^{(p,t)}) := W^\Theta(p_a^{(p,t)},t) := W^\Theta(t,p_b^{(p,t)}) := W^\Theta(p_b^{(p,t)}, t_p^{(p,t)}) := 1$\;
		Remove from $W^\Theta$ the arc $(p,t)$\;
    }
    \Return{$S^\Theta=(N^\Theta,M_0^\Theta)$}

    \caption{Transformation of the given system into another system.}
\end{algorithm}

\begin{figure}[!h]
 
\begin{minipage}{0.4\linewidth}

\centering

\begin{tikzpicture}[mypetristyle,scale=1]

\node (p) at (0,0) [place,thick,inner sep=1pt] {\scriptsize $m$};
\node [anchor=north] at (p.south) {$p$};

\node (p') at (1.5,1) [place,thick,inner sep=1pt] {\scriptsize $m'$};
\node [anchor=west] at (p'.east) {$p'$};

\node (t) at (0,1) [transition,thick] {};
\node [anchor=south] at (t.north) {$t$};

\node (t') at (1.5,2) [transition,thick] {};
\node [anchor=east] at (t'.west) {$t'$};

\draw [->,thick] (p) to node [left] {$w$} (t);
\draw [->,thick] (p') to node [above] {$w'$} (t);
\draw [->,thick] (p') to node [right] {$w''$} (t');

\end{tikzpicture}

\end{minipage}
$\xrightarrow{\textrm{Transformation }\Theta}{}{}$
\begin{minipage}{0.4\linewidth}

\centering

\begin{tikzpicture}[mypetristyle,scale=1]

\node (p) at (0,0) [place,thick,inner sep=1pt] {\scriptsize $m$};
\node (p') at (1.5,2.5) [place,thick,inner sep=1pt] {\scriptsize $m'$};
\node (pa) at (1,1.75) [place,thick,inner sep=1pt] {};
\node (pb) at (-1,1.75) [place,thick,tokens=1] {};

\node [anchor=west] at (p.east) {$p$};
\node [anchor=west] at (p'.east) {$p'$};
\node [anchor=east] at (pb.west) {$p_b^{(p,t)}$};
\node [anchor=west] at (pa.east) {$p_a^{(p,t)}$};

\node (tp) at (0,1) [transition,thick] {};
\node (t) at (0,2.5) [transition,thick] {};
\node (t') at (1.5,3.5) [transition,thick] {};

\node [anchor=east] at (t.west) {$t$};
\node [anchor=east] at (t'.west) {$t'$};
\node [anchor=west] at (tp.south east) {$t_p^{(p,t)}$};

\draw [->,thick] (p) to node [left] {$w$} (tp);
\draw [->,thick] (p') to node [above] {$w'$} (t);
\draw [->,thick] (p') to node [right] {$w''$} (t');
\draw [->,thick] (tp) to node [left] {} (pa);
\draw [->,thick] (pa) to node [left] {} (t);
\draw [->,thick] (t) to node [left] {} (pb);
\draw [->,thick] (pb) to node [left] {} (tp);

\end{tikzpicture}

\end{minipage}

\caption{Application of transformation $\Theta$ to the system on the left, yielding the system on the right. 
On the left, we assumed that $p\lbul = \{t\}$ (thus $p$ is not shared) and $t$ has at least two inputs (exactly two here); 
also $p$ and $p'$ are allowed to have several inputs,
$t$ and $t'$ are allowed to have several outputs,
and the arcs from $p'$ to $t$ and~$t'$ remain unchanged since $p'$ is shared.}

\label{FigAlgoTheta}

\end{figure}
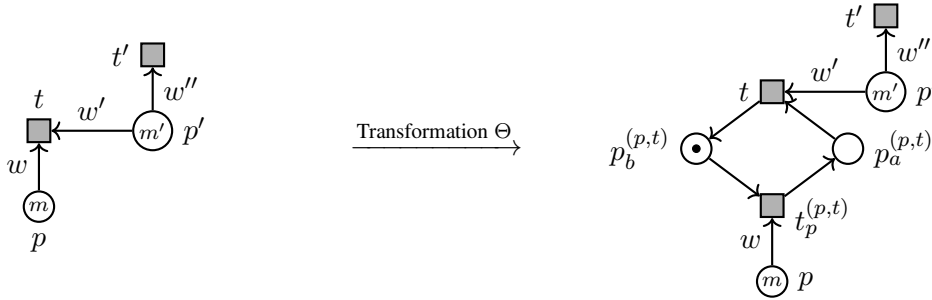

\begin{figure}[!h]
 
\begin{minipage}{0.31\linewidth}

\centering

\begin{tikzpicture}[mypetristyle,scale=1.2]

\node (p) at (1,1) [place,thick,inner sep=1pt] {\scriptsize $m$};

\node [anchor=north east] at (p.south west) {$p$};

\node (p1) at (0,0) [place,thick,inner sep=1pt] {\scriptsize $m_1$};
\node (p2) at (0,2) [place,thick,inner sep=1pt] {\scriptsize $m_2$};
\node (p3) at (2,2) [place,thick,inner sep=1pt] {\scriptsize $m_3$};
\node (p4) at (2,0) [place,thick,inner sep=1pt] {\scriptsize $m_4$};

\node [anchor=north east] at (p1.south west) {$p_1$};
\node [anchor=south east] at (p2.north west) {$p_2$};
\node [anchor=south west] at (p3.north east) {$p_3$};
\node [anchor=north west] at (p4.south east) {$p_4$};

\node (t1) at (0,1) [transition,thick] {};
\node (t2) at (1,2) [transition,thick] {};
\node (t3) at (2,1) [transition,thick] {};
\node (t4) at (1,0) [transition,thick] {};

\node [anchor=east] at (t1.west) {$t_1$};
\node [anchor=south] at (t2.north) {$t_2$};
\node [anchor=west] at (t3.east) {$t_3$};
\node [anchor=north] at (t4.south) {$t_4$};

\draw [->,thick] (p1) to node [left] {$2$} (t1);
\draw [->,thick] (t1) to node [left] {} (p2);
\draw [->,thick] (t1) to node [left] {} (p);

\draw [->,thick] (p2) to node [left] {} (t2);
\draw [->,thick] (p) to node [right] {$2$} (t2);
\draw [->,thick] (t2) to node [above] {$3$} (p3);

\draw [->,thick] (p3) to node [right] {$3$} (t3);
\draw [->,thick] (t3) to node [right] {} (p);
\draw [->,thick] (t3) to node [right] {$2$} (p4);

\draw [->,thick] (p) to node [right] {$5$} (t4);
\draw [->,thick] (p4) to node [below] {$2$} (t4);
\draw [->,thick] (t4) to node [below] {$4$} (p1);

\end{tikzpicture}

\end{minipage}
$\xrightarrow{\textrm{\scriptsize Transformation }\Theta}{}{}$
\begin{minipage}{0.55\linewidth}

\centering

\begin{tikzpicture}[mypetristyle,scale=1.5]

\node (p) at (1.5,1.5) [place,thick,inner sep=1pt] {\scriptsize $m$};

\node [anchor=north] at (p.south) {$p$};

\node (p1) at (0,0) [place,thick,inner sep=1pt] {\scriptsize $m_1$};
\node (p2) at (0,3) [place,thick,inner sep=1pt] {\scriptsize $m_2$};
\node (p3) at (3,3) [place,thick,inner sep=1pt] {\scriptsize $m_3$};
\node (p4) at (3,0) [place,thick,inner sep=1pt] {\scriptsize $m_4$};

\node [anchor=north east] at (p1.south west) {$p_1$};
\node [anchor=south east] at (p2.north west) {$p_2$};
\node [anchor=south west] at (p3.north east) {$p_3$};
\node [anchor=north west] at (p4.south east) {$p_4$};



\node (pa2) at (1.5,2.75) [place,thick,tokens=0] {};
\node (pb2) at (1.5,3.25) [place,thick,tokens=1] {};

\node [anchor=north,inner sep=1pt] at (pa2.south) {\scriptsize $p_a^{(p_2,t_2)}$};
\node [anchor=south,inner sep=1pt] at (pb2.north) {\scriptsize $p_b^{(p_2,t_2)}$};

%

\node (pa4) at (1.5,-0.25) [place,thick,tokens=0] {};
\node (pb4) at (1.5,0.25) [place,thick,tokens=1] {};

\node [anchor=north,inner sep=1pt] at (pa4.south) {\scriptsize $p_a^{(p_4,t_4)}$};
\node [anchor=south,inner sep=1pt] at (pb4.north) {\scriptsize $p_b^{(p_4,t_4)}$};

\node (t1) at (0,1.5) [transition,thick] {};
\node (t2) at (2,3) [transition,thick] {};
\node (t3) at (3,1.5) [transition,thick] {};
\node (t4) at (1,0) [transition,thick] {};

\node [anchor=east] at (t1.west) {$t_1$};
\node [anchor=south] at (t2.north) {$t_2$};
\node [anchor=west] at (t3.east) {$t_3$};
\node [anchor=north] at (t4.south) {$t_4$};

\node (tp2) at (1,3) [transition,thick] {};
\node (tp4) at (2,0) [transition,thick] {};

\node [anchor=north] at (tp2.south west) {$t_{p_2}^{(p_2,t_2)}~~~$};
\node [anchor=south] at (tp4.north east) {$~~~~t_{p_4}^{(p_4,t_4)}$};

\draw [->,thick] (p1) to node [left] {$2$} (t1);
\draw [->,thick] (t1) to node [left] {} (p2);
\draw [->,thick,bend left=30] (t1) to node [left] {} (p);

\draw [->,thick] (p2) to node [left] {} (tp2);
\draw [->,thick,bend right=30] (p) to node [right] {$2$} (t2);
\draw [->,thick] (t2) to node [above] {$3$} (p3);

\draw [->,thick] (p3) to node [right] {$3$} (t3);
\draw [->,thick,bend left=30] (t3) to node [right] {} (p);
\draw [->,thick] (t3) to node [right] {$2$} (p4);

\draw [->,thick,bend right=30] (p) to node [left] {$5$} (t4);
\draw [->,thick] (p4) to node [below, near start] {$2$} (tp4);
\draw [->,thick] (t4) to node [below] {$4$} (p1);


\draw [->,thick] (t2) to node [below] {} (pb2);
\draw [->,thick] (pb2) to node [below] {} (tp2);
\draw [->,thick] (tp2) to node [below] {} (pa2);
\draw [->,thick] (pa2) to node [below] {} (t2);


\draw [->,thick] (t4) to node [below] {} (pa4);
\draw [->,thick] (pa4) to node [below] {} (tp4);
\draw [->,thick] (tp4) to node [below] {} (pb4);
\draw [->,thick] (pb4) to node [below] {} (t4);

\end{tikzpicture}

\end{minipage}

\caption{
Application of transformation $\Theta$ to the $1$S system on the left, yielding the $1$S system on the right. 
One system is deadlockable iff the other one is.
}

\label{FigAlgoThetaH1S}

\end{figure}
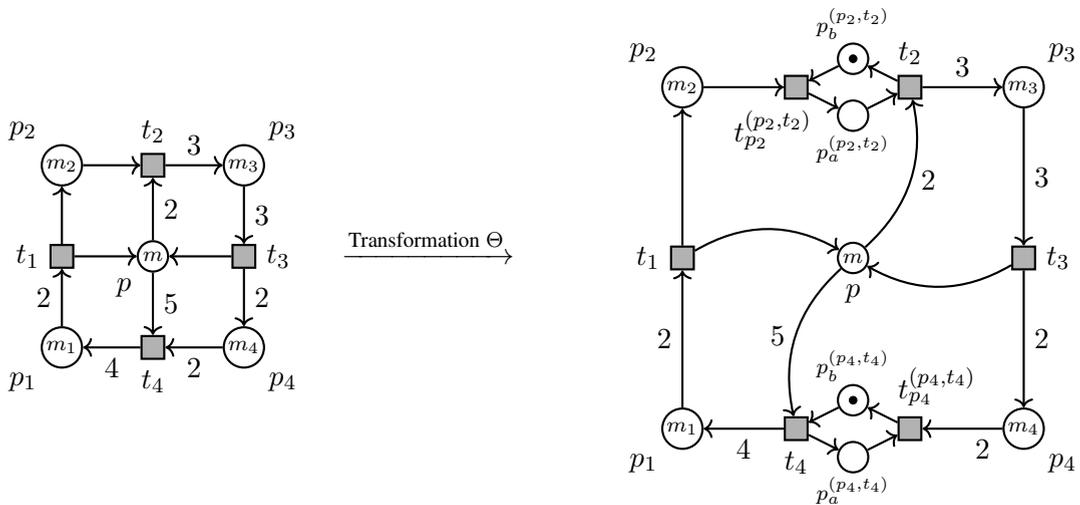

\begin{figure}[!h]
\centering
\begin{tikzpicture}[scale=0.7,mypetristyle]

\scriptsize

\node (p1) at (2.5,1.5) [place,thick,tokens=2] {};
\node [anchor=west] at (p1.east) {$p_1$};

\node (p2) at (4,3) [place,tokens=0] {};
\node [anchor=south] at (p2.north) {$p_2$};

\node (p3) at (7,3) [place] {};
\node [anchor=south] at (p3.north) {$p_3$};

\node (p4) at (8.5,1.5) [place,tokens=0] {};
\node [anchor=east,inner sep=1pt] at (p4.west) {$p_4$};

\node (p5) at (7,0) [place] {};
\node [anchor=north] at (p5.south) {$p_5$};

\node (p6) at (4,0) [place] {};
\node [anchor=north] at (p6.south) {$p_6$};

\node (p7) at (4,1.5) [place,tokens=3] {};
\node [anchor=north] at (p7.south) {$p_7$};

\node (p8) at (7,1.5) [place,tokens=4] {};
\node [anchor=east] at (p8.west) {$p_8$};

\node (t1) at (2.5,3) [transition,thick] {};
\node (t2) at (5.5,3) [transition,thick] {};
\node (t3) at (8.5,3) [transition,thick] {};
\node (t4) at (8.5,0) [transition,thick] {};
\node (t5) at (5.5,0) [transition,thick] {};
\node (t6) at (2.5,0) [transition,thick] {};

\node [anchor=south] at (t1.north) {$t_1$};
\node [anchor=south] at (t2.north) {$t_2$};
\node [anchor=south] at (t3.north) {$t_3$};
\node [anchor=north] at (t4.south) {$t_4$};
\node [anchor=north] at (t5.south) {$t_5$};
\node [anchor=north] at (t6.south) {$t_6$};

\draw [->,thick] (t6) to node [] {} (p6);
\draw [->,thick] (p6) to node [] {} (t5);
\draw [->,thick] (t5) to node [] {} (p5);
\draw [->,thick] (p5) to node [] {} (t4);
\draw [->,thick] (t4) to node [] {} (p4);
\draw [->,thick] (p4) to node [] {} (t3);
\draw [->,thick] (t3) to node [] {} (p3);
\draw [->,thick] (p7) to node [] {} (t6);
\draw [->,thick] (p3) to node [] {} (t2);
\draw [->,thick] (t2) to node [] {} (p2);
\draw [->,thick] (p2) to node [] {} (t1);
\draw [->,thick] (t1) to node [] {} (p7);
\draw [->,thick] (p7) to node [] {} (t3.south west);
\draw [->,thick] (t4.north west) to node [] {} (p7);
\draw [->,thick] (t2) to node [] {} (p8);
\draw [->,thick] (p8) to node [] {} (t5);
\draw [->,thick] (t1) to node [] {} (p1);
\draw [->,thick] (p1) to node [] {} (t6);
\end{tikzpicture}\\
\vspace*{5mm}
\begin{tikzpicture}[scale=1,mypetristyle]
\hspace*{1.47cm}\draw [->,thick] (0,1) to node [right] {Transformation $\Theta$} (0,0);
\end{tikzpicture}\\
\vspace*{5mm}
\begin{tikzpicture}[scale=0.7,mypetristyle]

\scriptsize

\node (p1) at (2.5,4.5) [place,thick,tokens=2] {};
\node [anchor=east] at (p1.west) {$p_1$};

\node (p2) at (6,6) [place,tokens=0] {};
\node [anchor=south] at (p2.north) {$p_2$};

\node (p3) at (13.25,6) [place] {};
\node [anchor=south] at (p3.north) {$p_3$};

\node (p4) at (17,1.5) [place,tokens=0] {};
\node [anchor=east,inner sep=1pt] at (p4.west) {$p_4$};

\node (p5) at (13.25,0) [place] {};
\node [anchor=north] at (p5.south) {$p_5$};

\node (p6) at (4.5,0) [place] {};
\node [anchor=north] at (p6.south) {$p_6$};

\node (p7) at (7,3) [place,tokens=3] {};
\node [anchor=north] at (p7.south) {$p_7$};



\node (p8) at (12,3) [place,tokens=4] {};
\node [anchor=east] at (p8.west) {$p_8$};

\node (tp8) at (12,1.8) [transition,thick] {};
\node [anchor=west] at (tp8.east) {$t_{p_8}$};

\node (t1) at (2.5,6) [transition,thick] {};
\node (tp1) at (2.5,3) [transition,thick] {};
\node (t2) at (9.5,6) [transition,thick] {};
\node (t3) at (17,6) [transition,thick] {};
\node (t4) at (17,0) [transition,thick] {};
\node (tp4) at (17,3) [transition,thick] {};
\node (t5) at (9.5,0) [transition,thick] {};

\node (t6) at (2.5,0) [transition,thick] {};
\node (tp6) at (6.5,0) [transition,thick] {};

\node (pa1) at (1.5,1.5) [place,thick,tokens=0] {};
\node [anchor=east] at (pa1.west) {$p_a^{(p_1,t_6)}$};
\node (pb1) at (2.5,1.5) [place,thick,tokens=1] {};
\node [anchor=south west,inner sep=1pt] at (pb1.north) {$p_b^{(p_1,t_6)}$};


\node (pa4) at (17,4.5) [place,thick,tokens=0] {};
\node [anchor=east,inner sep=1pt] at (pa4.west) {$p_a^{(p_4,t_3)}$};
\node (pb4) at (18,4.5) [place,thick,tokens=1] {};
\node [anchor=west,inner sep=1pt] at (pb4.east) {$p_b^{(p_4,t_3)}$};


\node (pa6) at (8,0) [place,thick,tokens=0] {};
\node [anchor=north,inner sep=2pt] at (pa6.south east) {$p_a^{(p_6,t_5)}$};
\node (pb6) at (8,-1.5) [place,thick,tokens=1] {};
\node [anchor=north,inner sep=1pt] at (pb6.south) {$p_b^{(p_6,t_5)}$};

%

\node (pa8) at (10,1.8) [place,thick,tokens=0] {};
\node [anchor=east,inner sep=1pt] at (pa8.west) {$p_a^{(p_8,t_5)}$};
\node (pb8) at (10.75,0.85) [place,thick,tokens=1] {};
\node [anchor=west,inner sep=1pt] at (pb8.east) {$p_b^{(p_8,t_5)}$};


\node [anchor=south] at (t1.north) {$t_1$};
\node [anchor=east] at (tp1.west) {$t_{p_1}$};
\node [anchor=south] at (t2.north) {$t_2$};
\node [anchor=south] at (t3.north) {$t_3$};
\node [anchor=north] at (t4.south) {$t_4$};
\node [anchor=west] at (tp4.east) {$t_{p_4}$};
\node [anchor=north] at (t5.south) {$t_5$};
\node [anchor=north] at (t6.south) {$t_6$};
\node [anchor=north] at (tp6.south) {$t_{p_6}$};

\draw [->,thick] (tp1) to node [] {} (pa1);
\draw [->,thick] (pb1) to node [] {} (tp1);
\draw [->,thick] (t6) to node [] {} (pb1);
\draw [->,thick] (pa1) to node [] {} (t6);



\draw [->,thick] (tp4) to node [] {} (pa4);
\draw [->,thick] (pb4) to node [] {} (tp4);
\draw [->,thick] (t3) to node [] {} (pb4);
\draw [->,thick] (pa4) to node [] {} (t3);


\draw [->,thick] (tp6) to node [] {} (pa6);
\draw [->,thick] (pb6) to node [] {} (tp6);
\draw [->,thick] (t5) to node [] {} (pb6);
\draw [->,thick] (pa6) to node [] {} (t5);




\draw [->,thick] (p4) to node [] {} (tp4);
\draw [->,thick] (t6) to node [] {} (p6);

\draw [->,thick] (t5) to node [] {} (pb8);
\draw [->,thick] (pb8) to node [] {} (tp8);
\draw [->,thick] (tp8) to node [] {} (pa8);
\draw [->,thick] (pa8) to node [] {} (t5);

\draw [->,thick] (p2) to node [] {} (t1);

\draw [->,thick] (p6) to node [] {} (tp6);
\draw [->,thick] (t5) to node [] {} (p5);
\draw [->,thick] (p5) to node [] {} (t4);
\draw [->,thick] (t4) to node [] {} (p4);
%
\draw [->,thick] (p7) to node [] {} (t6);
\draw [->,thick] (p7.north east) to node [] {} (t3.south west);
\draw [->,thick] (p8) to node [] {} (tp8);

\draw [->,thick] (p3) to node [] {} (t2);
\draw [->,thick] (t2) to node [] {} (p2);
\draw [->,thick] (p2) to node [] {} (t1);
\draw [->,thick,bend right=15] (t1) to node [] {} (p7);
%
\draw [->,thick] (t3) to node [] {} (p3);
\draw [->,thick,bend right=15] (t4.north west) to node [] {} (p7);
\draw [->,thick] (t2) to node [] {} (p8);
%
\draw [->,thick] (t1) to node [] {} (p1);
\draw [->,thick] (p1) to node [] {} (tp1);

\end{tikzpicture}
\caption{
Below, the H$1$S system obtained by applying transformation $\Theta$ to the H$1$S system above.
}
\label{FigAlgoThetaSwimming}

\end{figure}
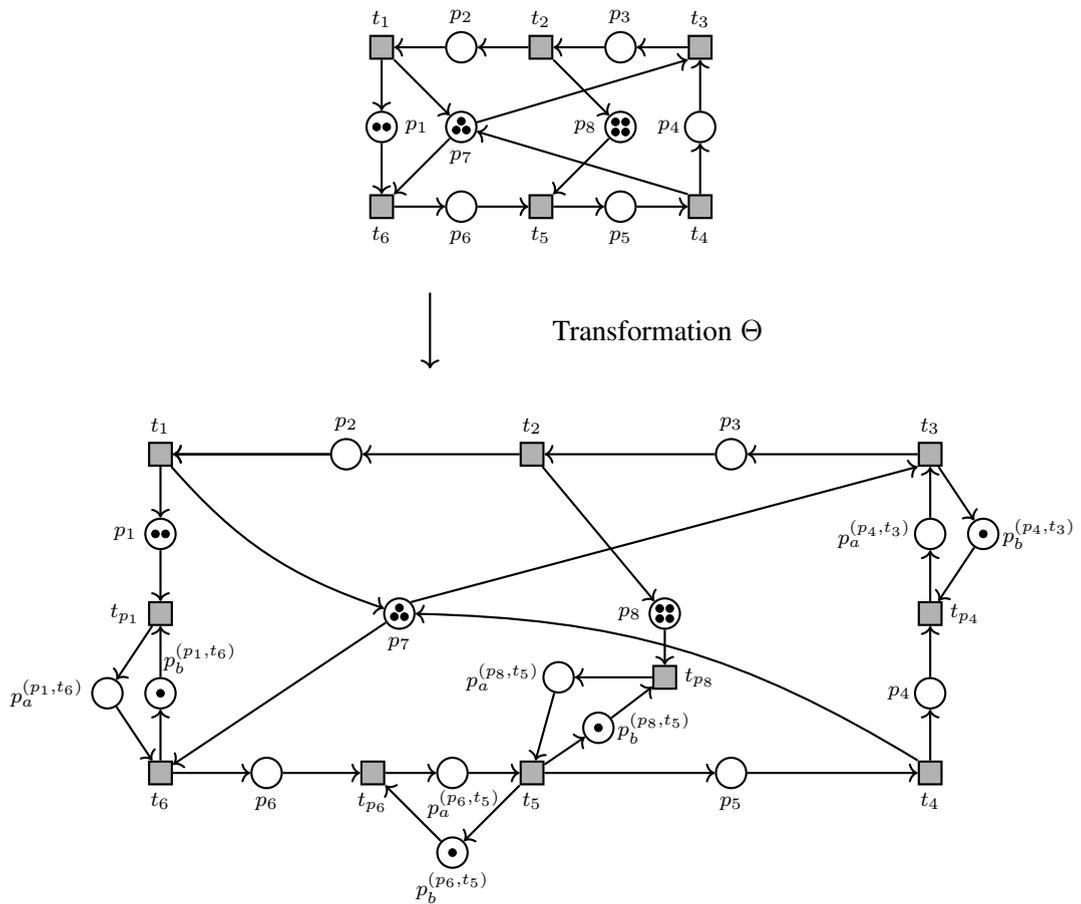

\clearpage

\noindent {\bf ILP for Figure~\ref{FigAlgoThetaSwimming}.}
Let us consider the H$1$S system $S^\Theta = (N^\Theta,M_0^\Theta)$ obtained in the figure.
The associated ILP is described by the state equation $M = M_0 + I \cdot Y$, where $M$ and $Y$ describe the variables,
in conjunction with the following linear inequalities enforcing non-fireability of~$t_{p_i}$'s: 
%

For $t_{p_1}$: ~~$SB(p_1,S) \cdot M(p_b^{(p_1,t_6)}) + M(p_1) < SB(p_1,S) \cdot W(p_b^{(p_1,t_6)},t_{p_1}) + W(p_1,t_{p_1})$\\
%
%
%
For $t_{p_4}$: ~~$SB(p_4,S) \cdot M(p_b^{(p_4,t_3)}) + M(p_4) < SB(p_4,S) \cdot W(p_b^{(p_4,t_3)},t_{p_4}) + W(p_4,t_{p_4})$\\
%
%
For $t_{p_6}$: ~~$SB(p_6,S) \cdot M(p_b^{(p_6,t_5)}) + M(p_6) < SB(p_6,S) \cdot W(p_b^{(p_6,t_5)},t_{p_6}) + W(p_6,t_{p_6})$\\
For $t_{p_8}$: ~~$SB(p_8,S) \cdot M(p_b^{(p_8,t_5)}) + M(p_8) < SB(p_8,S) \cdot W(p_b^{(p_8,t_5)},t_{p_8}) + W(p_8,t_{p_8})$

Transitions $t_1$, $t_2$ and $t_4$ have only one input place each, which must contain zero token:

For $t_1$: ~~$M(p_2) = 0$\\
For $t_2$: ~~$M(p_3) = 0$\\
For $t_4$: ~~$M(p_5) = 0$

Transition $t_5$ has only input places with structural bound at most $1$, hence we get the next simple non-fireability condition:

%
For $t_5$: ~~$M(p_a^{(p_6,t_5)}) + M(p_a^{(p_8,t_5)}) < 2$

The remaining transitions are outputs of the shared place $p_7$, which is their only input whose structural bound might exceed~$1$:


For $t_3$: ~~$SB(p_7,S) \cdot M(p_a^{(p_4,t_3)}) + M(p_7) < SB(p_7,S) \cdot W(p_a^{(p_4,t_3)},t_3) + W(p_7,t_3)$\\
For $t_6$: ~~$SB(p_7,S) \cdot M(p_a^{(p_1,t_6)}) + M(p_7) < SB(p_7,S) \cdot W(p_a^{(p_1,t_6)},t_6) + W(p_7,t_6)$.\\

Recall that each structural bound $SB(p_i,S)$, where $p_i$ is any place of $S$, can be replaced by an upper bound of polynomial length,
given as an input of the problem.\\

\noindent {\bf Notation.}
We denote by $NewPlaces(S,\Theta)$, respectively $NewTrans(S,\Theta)$, the places, respectively transitions, 
of $S^\Theta$ added to $S$ by transformation $\Theta$, i.e.\ those added in Algorithm~\ref{AlgoTheta}.
We denote by $NewTransPre(S,\Theta,t)$ the set $NewTrans \cap {\lbul(\lbul t)}$
and
by $Seq(A)$ the firing sequence containing one occurrence of each transition of $A$ in increasing label order
(that is, any natural ordering of transition names).
For instance, if $A=\{t_3,t_5,t_8\}$ then $Seq(A) = t_3\,t_5\,t_8$. If $A$ is empty, then $Seq(A) = \epsilon$.

\begin{definition}[Expanded sequences of $S$]
For each feasible sequence~$\alpha$ of~$S$, let us define inductively the feasible sequence $\theta(\alpha)$ of $S^\Theta$ as follows:\\
$-$ if $\alpha = \epsilon$ then $\theta(\alpha) = \epsilon$;\\
$-$ otherwise $\alpha$ is of the form $t \, \alpha'$,
and 
$\theta(\alpha) = Seq(NewTransPre(S,\Theta,t)) \, t \, \theta(\alpha')$.\\
Intuitively, $\theta$ inserts before each occurrence of any transition $t$ in $\alpha$ 
the firing sequence $Seq(A)$ containing one occurrence of each new transition whose firing enables some input place of $t$.
We say that $\theta(\alpha)$ is an \emph{expanded sequence}. 
\end{definition}

\begin{definition}[Reduced sequences of $S^\Theta$]
Let $\beta$ be any feasible sequence of $S^\Theta$.
We denote by $\hat\theta(\beta)$ the sequence obtained by removing--when it exists--the occurrence 
of each $t_{p_i} \in NewTrans(S,\Theta)$ from $\beta$ that fulfills the following condition:
denoting $\{t\} = (t_{p_i}\lbul)\lbul$, no occurrence of $t$ appears after this occurrence of $t_{p_i}$.
For instance, 
let 
$\beta = t_{p_2} t_{p_3} t_{p_1} t_1 t_{p_1} t_2 t_{p_3} t_{p_2}$,
where
$\{t_1\} = (t_{p_1}\lbul)\lbul$
and
$\{t_2\} = (t_{p_2}\lbul)\lbul = (t_{p_3}\lbul)\lbul$.
The reduced version of $\beta$ is $\hat\theta(\beta) = t_{p_2} t_{p_3} t_{p_1} t_1 t_2$.
We say that $\hat\theta(\beta)$ is a \emph{reduced sequence}. 
\end{definition}

\begin{theorem}[Properties of transformation $\Theta$]\label{ThTransfoDeadlock}
Let $S$ be any weighted system and $S^\Theta$ be the system obtained from $S$ 
through transformation $\Theta$, i.e.\ via Algorithm~\ref{AlgoTheta}.
Then:
\begin{enumerate}
\item If $S$ is a $1$S system then $S^\Theta$ is a $1$S system. 
If $S$ is homogeneous, then $S^\Theta$ is homogeneous. 
If $S$ is a H$1$S-\WMGineq{}, then $S^\Theta$ is a H$1$S-\WMGineq{}.
If $S$ is a strongly connected H$1$S-\WMGineq{} whose shared place deletion (if any) yields a strongly connected~\WMGineq{},
then $S^\Theta$ also has these properties.
\item For each sequence $\alpha$ feasible in $S$, the expanded sequence $\theta(\alpha)$ is feasible in $S^\Theta$.

\item For each sequence $\beta'$ feasible in $S^\Theta$, the reduced sequence $\beta=\hat\theta(\beta')$ is also feasible in $S^\Theta$
and
there is a sequence $\alpha$ feasible in $S$
such that 
$\theta(\alpha)$ is feasible in $S^\Theta$ and has the same Parikh vector as $\beta$.

\item If each place of $S$ has a structural bound, then each place of $S^\Theta$ has a structural bound.

\item $S^\Theta$ is deadlockable iff $S$ is deadlockable.

\item $S^\Theta$ is live iff $S$ is live.

\end{enumerate}

\end{theorem}

\begin{proof}
\begin{enumerate}
\item The transformation modifies and adds only non-shared places.
The other claims are deduced immediately from the construction.

\item It is clear, by construction. 

\item Reducing a sequence $\beta'$ feasible in $S^\Theta$ deletes at most one occurrence of each new transition:
since such transitions are non-choice and since their deleted occurrences appear rightmost in $\beta'$, 
the tokens produced by the latter ones are not used by other transitions in $\beta'$.
Thus, these occurrences can always be postponed to the end of the sequence,
implying that the reduced sequence $\beta=\hat\theta(\beta')$ is also feasible in $S^\Theta$.\\
Now, let us prove by induction on the length $n$ of any reduced sequence $\beta$ feasible in $S^\Theta$
that 
a sequence $\alpha$ is feasible in $S$ 
such that 
$\theta(\alpha)$ is feasible in $S^\Theta$ and has the same Parikh vector as $\beta$.

Base case: $n=0$, trivial: $\beta=\alpha=\epsilon$.

Inductive case: $n>0$. We suppose the claim to be true for every smaller reduced sequence.
Denote $\beta = \tau t$.
Since $\beta$ is reduced, transition $t$ is not one of the new transitions $t_{p_i}$ and thus also exists in $S$.
Denote by $\tau'$ the reduced version of $\tau$: each deleted occurrence of any $t_{p_i}$ (at most one is deleted for each $t_{p_i}$) 
fulfills $(t_{p_i}\lbul)\lbul = \{t\}$;
denote by $\rho$ the sequence of these deleted occurrences, in their increasing index order.
Applying the inductive hypothesis, a sequence $\gamma$ is feasible in $S$ 
such that 
$\theta(\gamma)$ is feasible in $S^\Theta$ and has the same Parikh vector as $\tau'$.
%
%
Clearly, the sequence $\alpha = \gamma t$ is feasible in $S$ 
such that
$\theta(\alpha) = \theta(\gamma) \rho t$ is feasible in $S^\Theta$ and has the same Parikh vector as $\beta$
(these interleavings are allowed since each $t_{p_i}$ is non-choice, so that $\rho$ remains feasible after $\theta(\gamma)$).
Hence the inductive step.

We proved the base case and the inductive case, hence the claim is true for every $n$.

\item Suppose that each place of $S$ has a structural bound and that some place $p_\infty$ of $S^\Theta$ has no structural bound.
It means that for each positive integer $k$, there is a solution $(M_k,Y_k)$ to the state equation of $S^\Theta$ 
such that $M_k(p_\infty) > k$.
Since each of the new places of $S^\Theta$ has structural bound at most $1$ by construction, 
$p_\infty$ is a place belonging to $S$.

By construction of $S^\Theta$, the only places having a new transition as input are the new places.
Thus, for each place $p'$ of $S$,
each input transition $t'$ of $p'$ in $S^\Theta$ is also an input of $p'$ in $S$.
Consequently, 
there exists a solution $(M_k',Y_k')$ to the state equation of $S^\Theta$ such that, for each place $p'$ of $S$, we have both following properties:\\
$-$ for each input transition $t'$ of $p'$ in $S^\Theta$, $Y_k'(t') = Y_k(t')$;\\
$-$ for each output transition $t'$ of $p'$ in $S^\Theta$, 
if $t'$ is not new then $Y_k'(t') = Y_k(t')$, otherwise $t'$ is new, hence the only output of $p'$;
in this second case, denote by $p$ the only output of $t'$:
if $M_k(p) = 0$ then $Y_k'(t') = Y_k(t')$, otherwise $M_k(p) = 1$ with $Y_k'(t') = Y_k(t') - 1$.\\
By construction of $S^\Theta$,  
it is clear that $(M_k',Y_k')$ is also a solution to the state equation of $S^\Theta$
such that, for each place $p$ of $S$, $M'_k(p) \ge M_k(p)$.
Moreover, the projection of $(M_k', Y_k')$ to the places and transitions of $S$, denoted by $(M_k'', Y_k'')$,
is a solution to the state equation of $S$,
with the property that $M_k''(p_\infty) = M_k'(p_\infty) \ge M_k(p_\infty) > k$.
This means that $p$ has no structural bound in $S$, a contradiction.
We deduce that each place of~$S^\Theta$ has a structural bound.
\item Clearly, if some sequence $\alpha$ feasible in $S$ leads to a deadlock,
from the above we deduce that $\theta(\alpha)\tau$ is feasible in $S^\Theta$, leading to a deadlock,
where $\tau$ is a possibly empty sequence of new transitions, 
each of which appears at most once in $\tau$.

For the converse, suppose that some sequence $\beta$ feasible in $S^\Theta=(N^\Theta,M_0^\Theta)$ leads to a deadlock. 
Define $\beta'$ as the reduced version of $\beta$.
Applying the claim $3$ proved above, some sequence $\alpha$ is feasible in $S$ such that $\theta(\alpha)$ is feasible in $S^\Theta$
and fulfills $\Parikh(\theta(\alpha)) = \Parikh(\beta')$.
Thus, the projection on the places of $S$ of the marking reached in $S^\Theta$ by firing $\beta'$ 
is a deadlock reached in $S$ by firing $\alpha$. We get the claim.
%
%
\item The liveness equivalence is directly deduced from the above (recall that no shared place is an input of any new transition $t_{p_i}$).
\end{enumerate}
\end{proof}

\noindent {\bf Exploiting Proposition~\ref{PropTh5.4}.}
Consider any $1$S (i.e.\ with at most one shared place) system $S$ in which each place has a structural bound;
let $S^\Theta$ be the system obtained from $S$ through transformation $\Theta$.
By~Theorem~\ref{ThTransfoDeadlock}, $S^\Theta$ is $1$S and each of its places has a structural bound;
besides,
checking the liveness of $S^\Theta$ is equivalent to checking the liveness of $S$.

Denote by $p_s$ the shared place of $S^\Theta$, if any.
For each transition $t$ in $p_s\lbul$, the only input place of~$t$ whose structural bound 
might be (strictly) greater than its output weights is $p_s$ by construction.
For each other transition $t'$, for each input place $p'$ of $t'$, the only output of $p'$ is $t'$;
besides, each such $t'$ has at most one input place whose structural bound 
might be (strictly) greater than its single output weight:
this can be the case only if this place already exists in $S$,
the new places having structural bound at most $1$.

Thus, 
for each transition $t$ of the system $S^\Theta$ obtained, 
for any marking $M \in PR(S^\Theta)$,
Proposition~\ref{PropTh5.4} provides a single integer linear inequality expressing the non-enabledness of $t$ at $M$.
The conjunction of these inequalities with the constraints of the state equation 
provides the system describing the (set of) deadlocks potentially reachable from $S^\Theta$.
This allows to check the liveness of the H$1$S-\WMGineq{} that fulfill the conditions of Corollary~\ref{CheckLivenessH1SWMG},
the latter conditions being preserved by the transformation (by~Theorem~\ref{ThTransfoDeadlock}.1).

Denoting $S^\Theta=(N^\Theta,M_0^\Theta)$, where $N^\Theta=(P^\Theta,T^\Theta,W^\Theta)$,
the size of this system is clearly linear in $|P^\Theta|\cdot |T^\Theta|\cdot m$, 
where 
$m$ is the length of the largest binary-encoded integer among the arc weights and the given upper bounds on structural bounds.

By definition of the transformation $\Theta$, 
the size of this ILP 
is linear in $|P|\cdot (|P|+|T|)\cdot m$, 
where 
$P$ and $T$ are the places and transitions of $S$.

\section{Reversibility of live \texorpdfstring{H$1$S}{} systems}\label{SecReversibility}

In weighted Petri nets, the reversibility checking problem is PSPACE-hard~\cite{esparza1996decidability}.
However, the notion of a \emph{T-sequence}, recalled next,
is exploited in several studies to reduce this complexity in some subclasses.

\begin{definition}[T-sequence \cite{Hujsa2015,HDM2016}]
Consider a system $S$ whose set of transitions is $T$ and denote by $I$ its incidence matrix.
A firing sequence $\sigma$ of $S$ is a T-sequence 
if it contains all transitions of $T$ (i.e.\ $\support(\sigma) = T$) and $I \cdot \Parikh(\sigma) = 0$ 
(i.e.\ $\Parikh(\sigma)$ is a consistency vector).
\end{definition}

\noindent {\bf A known necessary condition.}
In all weighted Petri nets, the existence of a feasible T-sequence is a known necessary condition 
of liveness and reversibility, taken together~\cite{Hujsa2015}.
Under the liveness assumption, the existence of a T-sequence characterizes reversibility in HFC systems,
allowing to stop the exploration of the reachability graph when such a sequence is found~\cite{Hujsa2015}.
%
%
Based on this result, a wide-ranging linear-time sufficient condition of liveness and reversibility 
in well-formed HFC nets is given by Theorem~$6.6$ in~\cite{HDM2016}.
%
%
Polynomial-time sufficient conditions of liveness and reversibility also exist 
for well-formed join-free (JF) nets (i.e.\ without synchronizations)~\cite{March09,HDM2016,HD2018}.\\

\noindent {\bf Toward a necessary and sufficient condition.}
Under the liveness assumption, in the H$1$S class, 
we show that the existence of a feasible T-sequence implies reversibility,
extending the result shown for the HFC class in~\cite{Hujsa2015},
from which we extract and adapt the proof scheme.
To achieve~it, we proceed as follows.

In Subsection~\ref{SubSecDefRev}, we formalize preliminary concepts related to sequences,
which we then use in Subsection~\ref{SubSecPrevRevChar}, 
where we recall the general idea behind the proof 
given in~\cite{Hujsa2015} for the live HFC class
and highlight the main difference with the H$1$S case.
In Subsection~\ref{SubsecAddNotions}, we present further notions related to transition firings,
which we exploit in our adapted proof in Subsection~\ref{SubsecSuffCondRliveJF}.
Finally, we construct in Subsection~\ref{NonExtensibility2places} a live H$2$S system
(with two shared places) for which the characterization does not work.
This shows that the characterization is tightly related to the H$1$S class.

\subsection{Definitions and notations}\label{SubSecDefRev}

\noindent {\bf Subsequences and projections.}
The sequence $\sigma'$ is a {\em subsequence} of the sequence $\sigma$ 
if $\sigma'$ is obtained from $\sigma$ by removing some occurrences of its transitions.
The {\em projection of $\sigma$ on the set $T' \subseteq T$ of transitions} 
is the subsequence of $\sigma$ with maximal length whose transitions belong to $T'$, denoted by $\projection{\sigma}{T'}$.
For example, the projection of 
the sequence $\sigma = t_1 \, t_2 \, t_3 \, t_2$
on the set $\{t_1 ,\, t_2\}$ is the sequence  $t_1 \, t_2 \, t_2$.\\

\noindent {\bf Infinite concatenation:} 
For a sequence $\sigma$, we denote by $\mathbf{\sigma^\infty}$ its infinite concatenation.\\

\noindent {\bf Local ordering:} 
Let $T'$ be a subset of transitions. The \emph{local ordering of} $T'$ \emph{induced by $\sigma$} 
is the sequence $\projection{\sigma^\infty}{T'}$ (obtained by projection).

In the system in the left of Figure \ref{HFCvsJF}, consider the set $p_2\lbul  = \{t_2,t_3\}$ and the feasible T-sequence 
$\sigma_r = t_2 \, t_1 \, t_3 \, t_4 \, t_5 \, t_2$.
The local ordering of $p_2\lbul$ induced by $\sigma_r$ 
is defined by $\sigma_2 = (t_2 \, t_3 \, t_2)^\infty$, 
which is the projection of $\sigma_r^\infty$ on the post-set $p_2\lbul$.
Note that this is not an order in the usual sense, 
but it specifies which sequence of transitions must be considered.\par\vspace{\baselineskip}

\subsection{Reversibility: from HFC to \texorpdfstring{H$1$S}{} systems}\label{SubSecPrevRevChar}

\noindent {\bf A previous method for reaching the initial marking in live HFC systems.}  
%
The proof of the characterization in~\cite{Hujsa2015} in the HFC case uses the following idea,
assuming the existence of a feasible T-sequence $\sigma_r$:
for each set $E$ of conflicting transitions (i.e.\ sharing an input place),
the projection of $\sigma_r$ on $E$
induces a local ordering that solves the associated conflicts;
then,
for any firing sequence $\sigma$ feasible at the initial marking and leading to a marking $M$,
one can reach the initial marking from $M$ by firing transitions
according to these orderings.\\
These proofs use the assumption of liveness and the structure of HFC nets:
by liveness,
one can always fire a transition;
by homogeneity 
and by the free-choice structure,
two conflicting transitions are either both enabled or both disabled by a marking.
These assumptions
ensure that any conflict resolution policy is achievable,
in particular a policy leading to the initial marking.\\
To illustrate, an HJF system, i.e.\ an HFC system without synchronizations, is pictured on the left of Figure~\ref{HFCvsJF}.
Denoting by $E$ the set $\{t_2,t_3\}$ of conflicting transitions of this homogeneous system,
and by $\sigma_r$ the T-sequence $t_2 \, t_1 \, t_3 \, t_4 \, t_5 \, t_2$,
the associated local ordering is $(t_2 \, t_3 \, t_2)^\infty$.
Then, if $t_3$ is fired first, the idea is to fire transitions until $p_2$ becomes enabled
and $t_2$ is fired,
implying that the prefix $t_2 \, t_3$ of the local ordering $(t_2 \, t_3 \, t_2)^\infty$ 
has been fired in a permuted fashion.
Thereafter,
the ordering can be used again as a policy solving the conflicts in $E$:
the next transition to be fired in $E$ is $t_2$.\par\vspace{\baselineskip}

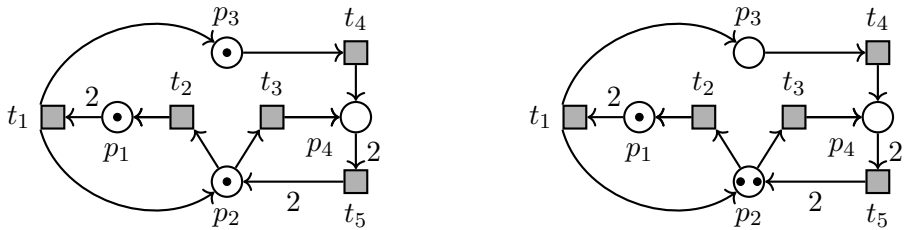
\begin{figure}[!h]
%
%

\centering

\begin{tikzpicture}[mypetristyle,scale=0.85]

\node (p1) at (1.3,1) [place,tokens=1] {};
\node (p2) at (3,0) [place,tokens=1] {};
\node (p3) at (3,2) [place,tokens=1] {};
\node (p4) at (5,1) [place] {};

\node [anchor=north] at (p1.south) {$p_1$};
\node [anchor=north] at (p2.south) {$p_2$};
\node [anchor=south] at (p3.north) {$p_3$};
\node [anchor=north east] at (p4.south west) {$p_4$};

\node (t1) at (0.3,1) [transition] {};
\node (t2) at (2.3,1) [transition] {};
\node (t3) at (3.7,1) [transition] {};
\node (t4) at (5,2) [transition] {};
\node (t5) at (5,0) [transition] {};

\node [anchor=east] at (t1.west) {$t_1$};
\node [anchor=south] at (t2.north) {$t_2$};
\node [anchor=south] at (t3.north) {$t_3$};
\node [anchor=south] at (t4.north) {$t_4$};
\node [anchor=north] at (t5.south) {$t_5$};

\draw [->,thick] (p1) to node [above, near start] {$2$} (t1);
\draw [->,thick] (t2) to node [above] {} (p1);

\draw [->,thick] (p2) to node [above right] {} (t2.south east);
\draw [->,thick] (p2) to node [above left] {} (t3.south west);

\draw [->,thick,bend left=55] (t1.north west) to node [above] {} (p3);
\draw [->,thick] (t2) to node [above] {} (p1);
\draw [->,thick,bend right=55] (t1.south west) to node [below] {} (p2);

\draw [->,thick] (t3) to node [above] {} (p4);
\draw [->,thick] (p3) to node [above] {} (t4);

\draw [->,thick] (t4) to node [right] {} (p4);

\draw [->,thick] (t3) to node [above] {} (p4);

\draw [->,thick] (p4) to node [right] {$2$} (t5);

\draw [->,thick] (t5) to node [below] {$2$} (p2);

\end{tikzpicture}
\hspace*{1.5cm}
\begin{tikzpicture}[mypetristyle,scale=0.85]

\node (p1) at (1.3,1) [place,tokens=1] {};
\node (p2) at (3,0) [place,tokens=2] {};
\node (p3) at (3,2) [place,tokens=0] {};
\node (p4) at (5,1) [place,tokens=0] {};

\node [anchor=north] at (p1.south) {$p_1$};
\node [anchor=north] at (p2.south) {$p_2$};
\node [anchor=south] at (p3.north) {$p_3$};
\node [anchor=north east] at (p4.south west) {$p_4$};

\node (t1) at (0.3,1) [transition] {};
\node (t2) at (2.3,1) [transition] {};
\node (t3) at (3.7,1) [transition] {};
\node (t4) at (5,2) [transition] {};
\node (t5) at (5,0) [transition] {};

\node [anchor=east] at (t1.west) {$t_1$};
\node [anchor=south] at (t2.north) {$t_2$};
\node [anchor=south] at (t3.north) {$t_3$};
\node [anchor=south] at (t4.north) {$t_4$};
\node [anchor=north] at (t5.south) {$t_5$};

\draw [->,thick] (p1) to node [above, near start] {$2$} (t1);
\draw [->,thick] (t2) to node [above] {} (p1);

\draw [->,thick] (p2) to node [above right] {} (t2.south east);
\draw [->,thick] (p2) to node [above left] {} (t3.south west);

\draw [->,thick,bend left=55] (t1.north west) to node [above] {} (p3);
\draw [->,thick] (t2) to node [above] {} (p1);
\draw [->,thick,bend right=55] (t1.south west) to node [below] {} (p2);

\draw [->,thick] (t3) to node [above] {} (p4);
\draw [->,thick] (p3) to node [above] {} (t4);

\draw [->,thick] (t4) to node [right] {} (p4);

\draw [->,thick] (t3) to node [above] {} (p4);

\draw [->,thick] (p4) to node [right] {$2$} (t5);

\draw [->,thick] (t5) to node [below] {$2$} (p2);

\end{tikzpicture}

\caption{
The HFC system on the left is live and enables the T-sequence $\sigma_r = t_2 \, t_1 \, t_3 \, t_4 \, t_5 \, t_2$.
At each reachable marking, some place is enabled, while $t_2$ and $t_3$ are either both enabled or both disabled.
Notice that only one occurrence of $t_2$ appears before the first occurrence of $t_3$ in $\sigma_r$:
after firing $t_3$ on the left, $t_4 \, t_5$ can be fired, leading to the system on the right, 
in which $p_2$ is enabled; firing $t_2$ fills the gap relative to $p_2\lbul$ in the largest prefix of $\sigma_r$ before $t_3$;
then, $t_1 \, t_2$  contains the rest of $\sigma_r$ and leads to the initial marking.
}
\label{HFCvsJF}
\end{figure}

\noindent {\bf Adapting the method to live H$1$S systems.}
We show that,
after the firing of any outgoing transition of $p$,
it is always possible to enable and fire the outgoing transitions of $p$ according to their occurrence order in $\sigma_r^\infty$,
in such a way that the initial marking is reached.
To prove it, we use the fact that the other input places (if any) of these outgoing transitions have exactly one output,
hence keep the tokens produced by ingoing transitions until their unique output is fired.

In next subsection, we provide further notions useful to our purpose. 
Then, in Subsection~\ref{SubsecSuffCondRliveJF},
we develop our characterization of reversibility using the notion of T-sequence.

\subsection{Additional notions and notations}\label{SubsecAddNotions}

To investigate the reversibility property, we borrow the following concepts from \cite{Hujsa2015}.
Let $S = (N,M_0)$ be any weighted system with $N=(P,T,W)$ and let $\sigma$ be a sequence feasible in $S$.\par\vspace{\baselineskip}

\noindent {\bf The next transition function $\mathbf{tnext}$:}
Consider some place $p$ and sequences $\sigma$, $\kappa$ 
such that $\Parikh(\sigma) \lneq \Parikh(\kappa)$.
Assume there exists 
a transition $t'$ in $p\lbul$ for which $\Parikh(\sigma)(t') < \Parikh(\kappa)(t')$. 
The transition $t'$ in $p\lbul$, 
among the ones satisfying $\Parikh(\sigma)(t') < \Parikh(\kappa)(t')$,
whose $(\Parikh(\sigma)(t')+1)$-th occurrence is the first to appear in $\kappa$,
is returned by a function,
called the \emph{next transition function} and denoted by ${tnext}(p\lbul, \sigma, \kappa)$.

Consider  
$\kappa = \sigma_r = t_2 \, t_1 \, t_3 \, t_4 \, t_5 \, t_2$ 
and 
$\sigma = t_2$ on the left of Figure~\ref{HFCvsJF}.
Then,
${tnext}(p_2\lbul, \sigma, \kappa) = t_3$,
where $p_2\lbul = \{t_2,t_3\}$.
For $\sigma' = t_3$,
we have 
${tnext}(p_2\lbul, \sigma', \kappa) = t_2$.\par\vspace{\baselineskip}

\noindent {\bf Prefix sequence $\mathbf{K_{t_i}^n(\sigma)}$:} 
Assuming $t_i$ occurs at least $n$ times in $\sigma$, with $n\geq 1$, 
the largest prefix sequence of $\sigma$ preceding the $n$-th occurrence of $t_i$ in $\sigma$,
thus containing $n-1$ occurrences of $t_i$,
is denoted by $K_{t_i}^n(\sigma)$, $n \ge 1$, or more simply $K_i^n(\sigma)$. 
For example, 
if $\sigma = t_1 \, t_2 \, t_1 \, t_3 \, t_1 \, t_2 \, t_3$, 
then $K_{t_1}^3(\sigma) = t_1 \, t_2 \, t_1 \, t_3$ 
and $K_{t_3}^1(\sigma) =  t_1 \, t_2 \, t_1$.\par\vspace{\baselineskip}

\noindent {\bf Delayed occurrences:} 
Consider a subset  $T' \subseteq T$ of transitions and a transition $t \in T'$.
Denote by $\tau = \projection{\sigma^\infty}{T'}$ 
the local ordering of $T'$ induced by $\sigma$.
An \emph{occurrence} of $t$ is \emph{delayed} by the firing of a sequence $\alpha$
relatively to $\tau$ if there exists $t' \in T'$, $t' \neq t$, such that,
noting $n = \Parikh(\alpha)(t')$ and $K = K_{t'}^n(\tau)$,
we have
$\Parikh(\alpha)(t) < \Parikh(K)(t)$.
In other words,
an occurrence of the transition $t$ is delayed by $\alpha$ relatively to the local ordering $\tau$ if 
$t$ occurred (strictly) fewer times in $\alpha$ 
than in the largest (finite) prefix sequence $K$ of $\tau$ preceding the $n$-th occurrence of $t'$ in $\tau$.

To illustrate,
in the system on the left of Figure~\ref{HFCvsJF}, let $\sigma_r = t_2 \, t_1 \, t_3 \, t_4 \, t_5 \, t_2$
be a feasible T-sequence;
consider the set
$p_2\lbul  = \{t_2,t_3\}$ and the associated local ordering $\tau_2 =  (t_2 \, t_3 \, t_2)^\infty$.
If the sequence $\alpha = t_3$ is fired first, then the local ordering defined by $\tau_2$ is broken 
and 
one occurrence of $t_2$ is delayed\footnote{Note that if $p_2\lbul$ contained more than two transitions,
we could have various delays for several output transitions.}: 
denoting 
$n = \Parikh(\alpha)(t_3) = 1$ 
and 
$K = K_{t_3}^n(\tau_2) = K_{t_3}^1(\tau_2) = t_2$,
we have
$\Parikh(\alpha)(t_2) = 0 < 1 = \Parikh(K)(t_2)$.
Consequently, after the initial firing of $\alpha = t_3$,
if
one aims at removing the delay(s) as soon as possible, 
possibly by following the local orderings in other places first,
the next transition to be fired in $p_2\lbul$ 
is ${tnext}(p_2\lbul, \alpha, \tau_2) = t_2$.

These notions are also illustrated on the H$1$S system of Figure~\ref{TsequenceHAC}.

\begin{figure}[!h]
\centering

\begin{tikzpicture}[mypetristyle,scale=0.85]

\node (p1) at (1,1.3) [place,tokens=1] {};
\node (p2) at (3,-0.5) [place,tokens=1] {};
\node (p3) at (3,2.5) [place,tokens=2] {};
\node (p4) at (5,1) [place] {};
\node (p5) at (3,1.3) [place,tokens=2] {};
\node (p6) at (3,0.7) [place,tokens=2] {};
\node (p7) at (1,0.7) [place,tokens=1] {};

\node [anchor=south] at (p1.north) {$p_1$};
\node [anchor=north] at (p2.south) {$p_2$};
\node [anchor=south] at (p3.north) {$p_3$};
\node [anchor=west] at (p4.east) {$p_4$};
\node [anchor=south] at (p5.north) {$p_5$};
\node [anchor=north] at (p6.south) {$p_6$};
\node [anchor=north] at (p7.south) {$p_7$};

\node (t1) at (0,1) [transition] {};
\node (t2) at (2,1) [transition] {};
\node (t3) at (4,1) [transition] {};
\node (t4) at (5,2.5) [transition] {};
\node (t5) at (5,-0.5) [transition] {};

\node [anchor=east] at (t1.west) {$t_1$};
\node [anchor=south] at (t2.north) {$t_2$};
\node [anchor=south] at (t3.north) {$t_3$};
\node [anchor=west] at (t4.east) {$t_4$};
\node [anchor=west] at (t5.east) {$t_5$};

\draw [->,thick] (p1) to node [above] {$2$} (t1);
\draw [->,thick] (t2) to node [above] {} (p1);

\draw [->,thick] (p2) to node [above right] {} (t2);
\draw [->,thick] (p2) to node [above left] {} (t3);

\draw [->,thick,bend left=40] (t1.north west) to node [above] {} (p3);
\draw [->,thick] (t2) to node [above] {} (p1);
\draw [->,thick,bend right=40] (t1.south west) to node [below] {} (p2);

\draw [->,thick] (t3) to node [above] {} (p4);
\draw [->,thick] (p3) to node [above] {} (t4);

\draw [->,thick] (t4) to node [right] {} (p4);
\draw [->,thick] (t3) to node [above] {} (p4);

\draw [->,thick] (p4) to node [right] {$2$} (t5);
\draw [->,thick] (t5) to node [below] {$2$} (p2);

\draw [->,thick] (t1) to node [below] {$2$} (p7);
\draw [->,thick] (p7) to node [below] {} (t2);

\draw [->,thick] (t3) to node [below] {} (p5);
\draw [->,thick] (p5) to node [below] {} (t2);

\draw [->,thick] (t2) to node [below] {} (p6);
\draw [->,thick] (p6) to node [below] {} (t3);

\end{tikzpicture}

\caption{
This live H$1$S system enables the T-sequence 
$\sigma_r = t_2 \, t_1 \, t_3 \, t_4 \, t_5 \, t_2 \, t_3 \, t_5 \, t_3 \, t_2 \, t_1 \, t_2 \, t_4 \, t_5 \, t_3$.
It is also reversible.
Consider $p_2\lbul  = \{t_2,t_3\}$ and the associated local ordering $\tau_2 = (t_2 \, t_3 \, t_2 \, t_3 \, t_3 \, t_2 \, t_2 \, t_3)^\infty$.
Suppose $\alpha = t_3$ is fired first. Then, the local ordering defined by $\tau_2$ is broken 
and 
one occurrence of $t_2$ is delayed. Denoting $n = \Parikh(\alpha)(t_3) = 1$ 
and 
$K = K_{t_3}^n(\tau_2) = K_{t_3}^1(\tau_2) = t_2$,
we have
$\Parikh(\alpha)(t_2) = 0 < 1 = \Parikh(K)(t_2)$.
Thus, ${tnext}(p_2\lbul, \alpha, \tau_2) = t_2$.
}
\label{TsequenceHAC}
%
\end{figure}
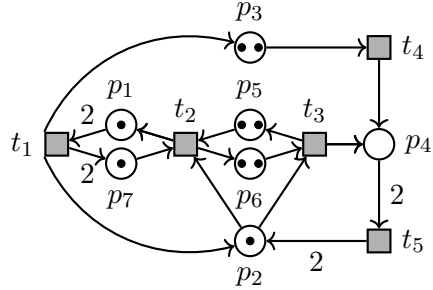

\subsection{Reversibility of live \texorpdfstring{H$1$S}{} systems}\label{SubsecSuffCondRliveJF}

We show that the existence of a feasible T-sequence is sufficient for reversibility in live H$1$S systems. 
Since this result is already known to hold in \WMGineq{}~\cite{WTS92}, we suppose the existence of one shared place $p$.
This result is embodied by Corollary~\ref{H1Srev} in Section~\ref{CharRev}.
To achieve it, we derive variants of the proofs of Section $4$ in \cite{Hujsa2015},
pointing out the differences between the HFC case and our H$1$S case.

In \cite{Hujsa2015}, for any live HFC system with a feasible T-sequence $\sigma_r$,
two algorithms are presented that construct, after any single firing of any transition,
a firing sequence leading to the initial marking. These algorithms form two consecutive steps: 
\begin{enumerate}
\item After the firing of some transition $t$ from the initial marking,
the first algorithm, called Algorithm~\ref{AlgoSC1} in the sequel,
fires transitions by following local orderings 
until all delayed occurrences are fired. The sequence obtained is denoted by $\sigma_t$.
If no occurrence is delayed by the firing of $t$, $\sigma_t$ is empty.
\item Then, the second algorithm, called Algorithm~\ref{AlgoSC2} in the sequel,
applies to the sequence $t \sigma_t$ resulting from Algorithm~\ref{AlgoSC1}
and completes this sequence to reach the initial marking.
The sequence obtained is denoted by $\sigma_t'$.
\end{enumerate}

At the end, we obtain the sequence $t \sigma_t \sigma_t'$ 
whose Parikh vector is a multiple of the Parikh vector of the initial T-sequence $\sigma_r$,
hence it is also a T-sequence.
These two steps are depicted in Figure~\ref{FigSigmaC}.

\begin{figure}[!h]
\begin{minipage}{1\linewidth}

\centering

\begin{tikzpicture}[scale=1.1]

\node (M1) at (0.9,1.25) {$M_t'$};
\node (M0) at (0,2) {$M_0$};
\node (M) at (2,2) {$M_t$};

\draw [->,thick,above] (M0) to node {$t$} (M);
\draw [->,thick, loop,below] (M0) to node {$\sigma_r$} (M0);

\draw [->,thick,below right, very near start, bend left=35] (M) to node {$\sigma_t$} (M1);

\draw [->,thick, left, bend left=35] (M1) to node {$\sigma_t' \;$} (M0);

\end{tikzpicture}
		
\end{minipage}

\caption{If the T-sequence $\sigma_r$ is feasible and $t$ is fired,   
then Algorithm~\ref{AlgoSC1} builds the sequence $\sigma_t$
and  Algorithm~\ref{AlgoSC2} computes the sequence $\sigma_t'$,
which returns to the initial marking. 
}

\label{FigSigmaC}
\end{figure}
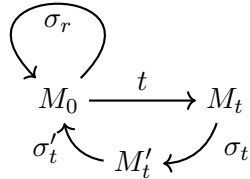

Once the existence of such sequences is proved,
the method can be generalized to arbitrary sequences, constructing,
from any reachable marking, some sequence that leads to the initial marking.

In the sequel, we consider any live H$1$S system $S=(N,M_0)$ enabling a T-sequence.
For each algorithm, we design a variant applying to $S$.

In \cite{STECS,Hujsa2015,HDM2016},
an \emph{equal-conflict set} $E$ of a net $N=(P,T,W)$ is defined as a maximal subset of~$T$ 
such that, for all distinct transitions $t_i, t_j \in E$,
$t_i$ and $t_j$ are in conflict (i.e.\ share some input place),
have the same pre-set and, for each place $p \in {\lbul t_i}$, $W(p,t_i)=W(p,t_j)$.
Notice that each singleton formed of a transition without input defines an equal-conflict set.

In the case of H$1$S systems with shared place $p$,
the equal-conflict sets considered are (disjoint) subsets of the post-set of~$p$ and singletons.

\subsubsection{First part of the sequence construction}

We deduce Algorithm~\ref{AlgoSC1} below from Algorithm $1$ in \cite{Hujsa2015},
where $p$ is the unique shared place in the system.
We suppose w.l.o.g. that $t$ belongs to the post-set of $p$, since otherwise 
we only have to fire the sequence $\sigma_r \pminus t$ to reach the initial marking, which is done by Algorithm~\ref{AlgoSC2}:
in this case, we leave $\sigma_t$ empty.
Notice that the inner loop terminates when $tnext(p\lbul,\alpha,\kappa_0)$ becomes enabled.

\begin{algorithm}\label{AlgoSC1}
    \KwData{The T-sequence $\sigma_r$, which is feasible in $S$; the system $(N,M_t)$ obtained by firing $t\in p^\bullet$ in $S$.}
    \KwResult{A sequence $\sigma_t$ feasible in $(N,M_t)$ 
	that fires the delayed occurrences of $\kappa_0 = K^1_{t}(\sigma_r)$.}

    $\alpha := t$\;
    \While{$\exists \, t' \in p\lbul \setminus \{t\}, \, \Parikh(\kappa_0)(t') > \Parikh(\alpha)(t')$}{
	\While{$tnext(p\lbul,\alpha,\kappa_0)$ is not enabled}{
	    Among the enabled transitions, 
	    fire the transition $t_i$ whose next occurrence after the $\Parikh(\alpha)(t_i)$-th
		appears first in $\sigma_r^\infty$\;
	    $\alpha := \alpha \, t_i$\;
	}
	Fire the transition $t_j = tnext(p\lbul,\alpha,\kappa_0)$\;
	$\alpha := \alpha \, t_j$\;
    }
    $\alpha$ is of the form $t \, \sigma_t$\;
    \Return{$\sigma_t$}
  
\caption{From the feasible T-sequence $\sigma_r$ and the transition $t$, 
with $t \in p\lbul$, 
construction of a sequence $\sigma_t$ that contains the occurrences of $p\lbul$ 
delayed by $t$ relatively to $\projection{\sigma_r^\infty}{p\lbul}$
by following the ordering induced by $\sigma_r$ for enabled transitions.}
\end{algorithm}

Algorithm~\ref{AlgoSC1} considers an initial firing of a single transition $t \in p^\bullet$
and follows the local (trivial) orderings induced by the T-sequence in each (singleton) equal-conflict set not intersecting $p^\bullet$,
as well as the local ordering of $p^\bullet$ 
(which is not always an equal-conflict set)
until all the delayed occurrences (if any) are fired.

An application of Algorithm~\ref{AlgoSC1} is given in Figure \ref{JFalgo1}.

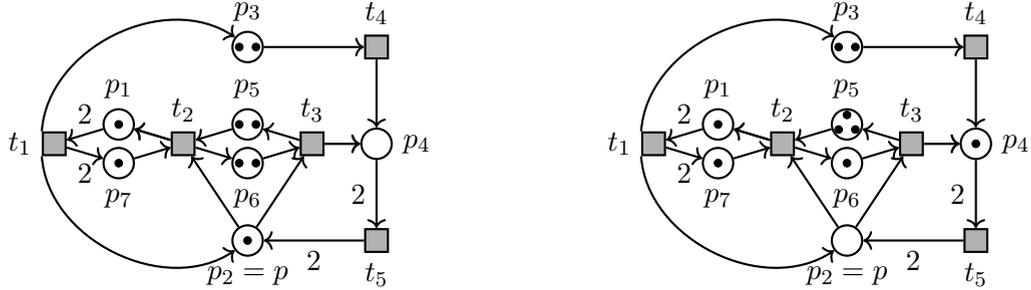
\begin{figure}[!h]

\centering
 
\begin{minipage}{0.43\linewidth}
\centering

\begin{tikzpicture}[mypetristyle,scale=0.85]

\node (p1) at (1,1.3) [place,tokens=1] {};
\node (p2) at (3,-0.5) [place,tokens=1] {};
\node (p3) at (3,2.5) [place,tokens=2] {};
\node (p4) at (5,1) [place] {};
\node (p5) at (3,1.3) [place,tokens=2] {};
\node (p6) at (3,0.7) [place,tokens=2] {};
\node (p7) at (1,0.7) [place,tokens=1] {};

\node [anchor=south] at (p1.north) {$p_1$};
\node [anchor=north] at (p2.south) {$p_2=p$};
\node [anchor=south] at (p3.north) {$p_3$};
\node [anchor=west] at (p4.east) {$p_4$};
\node [anchor=south] at (p5.north) {$p_5$};
\node [anchor=north] at (p6.south) {$p_6$};
\node [anchor=north] at (p7.south) {$p_7$};

\node (t1) at (0,1) [transition] {};
\node (t2) at (2,1) [transition] {};
\node (t3) at (4,1) [transition] {};
\node (t4) at (5,2.5) [transition] {};
\node (t5) at (5,-0.5) [transition] {};

\node [anchor=east] at (t1.west) {$t_1$};
\node [anchor=south] at (t2.north) {$t_2$};
\node [anchor=south] at (t3.north) {$t_3$};
\node [anchor=south] at (t4.north) {$t_4$};
\node [anchor=north] at (t5.south) {$t_5$};

\draw [->,thick] (p1) to node [above] {$2$} (t1);
\draw [->,thick] (t2) to node [above] {} (p1);

\draw [->,thick] (p2) to node [above right] {} (t2);
\draw [->,thick] (p2) to node [above left] {} (t3);

\draw [->,thick,bend left=60] (t1.north west) to node [above] {} (p3);
\draw [->,thick] (t2) to node [above] {} (p1);
\draw [->,thick,bend right=60] (t1.south west) to node [below] {} (p2);

\draw [->,thick] (t3) to node [above] {} (p4);
\draw [->,thick] (p3) to node [above] {} (t4);

\draw [->,thick] (t4) to node [right] {} (p4);
\draw [->,thick] (t3) to node [above] {} (p4);

\draw [->,thick] (p4) to node [left] {$2$} (t5);
\draw [->,thick] (t5) to node [below] {$2$} (p2);

\draw [->,thick] (t1) to node [below] {$2$} (p7);
\draw [->,thick] (p7) to node [below] {} (t2);

\draw [->,thick] (t3) to node [below] {} (p5);
\draw [->,thick] (p5) to node [below] {} (t2);

\draw [->,thick] (t2) to node [below] {} (p6);
\draw [->,thick] (p6) to node [below] {} (t3);

\end{tikzpicture}
\end{minipage}
\hspace*{1cm}
\begin{minipage}{0.43\linewidth}
\centering
\begin{tikzpicture}[mypetristyle,scale=0.85]

\node (p1) at (1,1.3) [place,tokens=1] {};
\node (p2) at (3,-0.5) [place,tokens=0] {};
\node (p3) at (3,2.5) [place,tokens=2] {};
\node (p4) at (5,1) [place,tokens=1] {};
\node (p5) at (3,1.3) [place,tokens=3] {};
\node (p6) at (3,0.7) [place,tokens=1] {};
\node (p7) at (1,0.7) [place,tokens=1] {};

\node [anchor=south] at (p1.north) {$p_1$};
\node [anchor=north] at (p2.south) {$p_2=p$};
\node [anchor=south] at (p3.north) {$p_3$};
\node [anchor=west] at (p4.east) {$p_4$};
\node [anchor=south] at (p5.north) {$p_5$};
\node [anchor=north] at (p6.south) {$p_6$};
\node [anchor=north] at (p7.south) {$p_7$};

\node (t1) at (0,1) [transition] {};
\node (t2) at (2,1) [transition] {};
\node (t3) at (4,1) [transition] {};
\node (t4) at (5,2.5) [transition] {};
\node (t5) at (5,-0.5) [transition] {};

\node [anchor=east] at (t1.west) {$t_1$};
\node [anchor=south] at (t2.north) {$t_2$};
\node [anchor=south] at (t3.north) {$t_3$};
\node [anchor=south] at (t4.north) {$t_4$};
\node [anchor=north] at (t5.south) {$t_5$};

\draw [->,thick] (p1) to node [above] {$2$} (t1);
\draw [->,thick] (t2) to node [above] {} (p1);

\draw [->,thick] (p2) to node [above right] {} (t2);
\draw [->,thick] (p2) to node [above left] {} (t3);

\draw [->,thick,bend left=60] (t1.north west) to node [above] {} (p3);
\draw [->,thick] (t2) to node [above] {} (p1);
\draw [->,thick,bend right=60] (t1.south west) to node [below] {} (p2);

\draw [->,thick] (t3) to node [above] {} (p4);
\draw [->,thick] (p3) to node [above] {} (t4);

\draw [->,thick] (t4) to node [right] {} (p4);
\draw [->,thick] (t3) to node [above] {} (p4);

\draw [->,thick] (p4) to node [left] {$2$} (t5);
\draw [->,thick] (t5) to node [below] {$2$} (p2);

\draw [->,thick] (t1) to node [below] {$2$} (p7);
\draw [->,thick] (p7) to node [below] {} (t2);

\draw [->,thick] (t3) to node [below] {} (p5);
\draw [->,thick] (p5) to node [below] {} (t2);

\draw [->,thick] (t2) to node [below] {} (p6);
\draw [->,thick] (p6) to node [below] {} (t3);

\end{tikzpicture}

\end{minipage}

\caption{
On the left, a live, homogeneous system $S=(N,M_0)$ with a single shared place. 
The T-sequence $\sigma_r = t_2 \, t_1 \, t_3 \, t_4 \, t_5 \, t_2 \, t_3 \, t_5 \, t_3 \, t_2 \, t_1 \, t_2 \, t_4 \, t_5 \, t_3$
is initially feasible.
A single firing of $t_3$ leads to the system on the right,
in which one occurrence of $t_2$ is delayed by $t_3$ relatively to 
$\projection{\sigma_r^\infty}{p_2\lbul} = (t_2 \, t_3 \, t_2 \, t_3 \, t_3 \, t_2 \, t_2 \, t_3)^\infty$.
Thus, $\kappa_0 = t_2 \, t_1$. 
Algorithm~\ref{AlgoSC1} fires $t_4$.
The enabled transitions are now $t_4$ and $t_5$.
Since $\alpha = t_3 \, t_4$, the transition $t_i \in \{t_4, t_5\}$ whose next occurrence after the $\Parikh(\alpha)(t_i)$-th
appears first in $\sigma_r^\infty$ is $t_5$.
A firing of $t_5$ enables $p_2 = p$ and disables $p_4$. The inner loop stops with $\alpha = t_3 \, t_4 \, t_5$.
Then,
the transition $t_j = tnext(p\lbul,\alpha,\kappa_0) = t_2$ is fired, implying that $\alpha = t_3 \, t_4 \, t_5 \, t_2$.
Since $\Parikh(\kappa_0)(t_2) = \Parikh(\alpha)(t_2)$, the outer loop stops and no delayed occurrence remains.
}
\label{JFalgo1}
\end{figure}

The termination and validity of this algorithm are shown in Lemma \ref{FirstLoop} below,
which is a variant of Lemma $2$ in \cite{Hujsa2015} for H$1$S systems.
We do not need here the notion of fairness exploited in~\cite{Hujsa2015}, as detailed in the following proof.

\begin{lemma}[Termination and validity of Algorithm~\ref{AlgoSC1}]\label{FirstLoop}
Let $S=(N,M_0)$ be a live H$1$S system in which a T-sequence $\sigma_r$ is feasible.
Then, 
for every transition~$t$ enabled by $M_0$, with $\myarrow{M_0}{M_t}{t}$,
Algorithm$\,$\ref{AlgoSC1} terminates and computes a sequence $\sigma_t$ feasible at $M_t$
such that $t \, \sigma_t$ does not induce any delayed occurrence relatively to the local ordering of $p^\bullet$ induced by $\sigma_r$.
\end{lemma}

\begin{proof}
Comparing with the proof of Lemma~$2$ in~\cite{Hujsa2015},
we replace the equal-conflict set $E^t$ with $p\lbul$,
we use explicitly $tnext(p\lbul,\alpha,\kappa_0)$ in the stopping condition of the inner loop,
we get rid of the fairness results
and
we reason on the H$1$S structure.

Notice that, since the system is live, it is always possible to fire some transition.
Let us show that the algorithm terminates and is correct.
To achieve it, denote by $\alpha_n$ the sequence $\alpha$ at the beginning of the $n$-th iteration of the outer loop, 
just before entering the inner loop.
Denote by~$\ell$ the number of transition occurrences in $\kappa_0$ that belong to $p\lbul$:
the algorithm stops when these~$\ell$ occurrences have been fired.
Besides, for any positive integer $n < \ell$, $\alpha_n$ contains exactly~$n$ occurrences of transitions that belong to $p\lbul$:
before the first iteration of the inner loop, $\alpha_1 = t$.
Let us prove by induction on this number $n < \ell$ the following property $\mathcal{P}(n)$:

"At the beginning of the $n$-th iteration of the inner loop, 
let us write $\alpha_n = \tau_1 \, t_1 \, \tau_2 \, t_2 \, \ldots \, \tau_n \, t_n$
where
each $t_i$ is an occurrence of some transition in $p\lbul$, with $t_1 = t$,
and 
each $\tau_i$ is a sequence whose transitions do not belong to $p\lbul$, with $\tau_1 = \epsilon$.
Also, let us write $\kappa_0 = \tau_1' \, \tau_2' \, t_2' \, \ldots \, \tau_\ell' \, t_\ell'$
where
$\tau_1' = \epsilon$,
for each $i \in [1,n]$, $t_i = t_i'$,
and
for each $j \in [1,\ell]$, $\tau_j'$ contains only occurrences of transitions that do no belong to $p\lbul$.
Then $\alpha_n$ is constructible and fulfills $\Parikh(\tau_1 \ldots \tau_n) \ge \Parikh(\tau_1' \ldots \tau_n')$."

\noindent $-$ Base case: $n=1$: $\alpha_1=t$ is trivially constructed at the beginning, and $\tau_1' = \epsilon$,
so that
$\Parikh(\tau_1) = \zero \ge \Parikh(\tau_1') = \zero$.

\noindent $-$ Inductive case: $n>1$: We suppose $\mathcal{P}(n)$ to be true,
thus
$\alpha_n$ is constructible and fulfills $\Parikh(\tau_1 \ldots \tau_n) \ge \Parikh(\tau_1' \ldots \tau_n')$.
Since $\kappa_0$ is a prefix of $\sigma_r$
and
since the inner loop fires enabled transitions following the order of $\sigma_r^\infty$,
it fires first the occurrences (if any) that remain to be fired in $\tau_{n+1}'$, yielding a sequence $\gamma$
that leads to the marking $M_{n+1}$ (such occurrences remain feasible since they are non-choice, 
i.e.\ since the P-subsystem $S_{\setminus p}$ obtained by deleting $p$ in $S$ is persistent).
Let us show that $M_{n+1}$ enables at least all the non-shared input places of $t_{n+1} = tnext(p\lbul,\alpha_n,\kappa_0)$.

Let us define $\kappa_0(n) = \tau_1' \, \tau_2' \, t_2' \, \ldots \, \tau_n' \, t_n'$.
From the above, we deduce $\Parikh(\alpha_n \, \gamma) \ge \Parikh(\kappa_0(n) \, \tau_{n+1}')$.
Moreover, $\Parikh(\alpha_n \, \gamma)(t_{n+1}) = \Parikh(\kappa_0(n) \, \tau_{n+1}')(t_{n+1})$.
Consequently, for each input place $p'$ of $t_{n+1}$, the inputs of $p'$ have been fired at least as many times as in $\kappa_0(n) \tau_{n+1}'$,
and
the single output of $p'$ has been fired the same number of times.
Thus, $M_{n+1}$ enables each such place $p'$.

If $M_{n+1}$ enables $p$, we are done. Otherwise, further firings occur in the inner loop,
following the order of $\sigma_r^\infty$,
knowing that $\support(\sigma_r)$ is the set $T$ of all transitions of the net.
Let us show that $\gamma$ is completed to a sequence $\tau_{n+1}$ that leads to a marking enabling $p$.

Since $S$ is live, there is a sequence $\sigma$ feasible at $M_{n+1}$ that leads to some marking $M$ enabling $p$
and
such that $\sigma$ does not contain any outgoing transition of $p$, since $S$ is homogeneous.
Since $\support(\sigma_r) = T$, firings occur in the inner loop until either $\Parikh(\sigma)$ is exceeded, enabling $p$, or $p$ is enabled before.

We proved that a finite sequence $\tau_{n+1}$ is fired in the inner loop, leading to a marking that enables 
$t_{n+1} = tnext(p\lbul,\alpha_n,\kappa_0)$.
Thus, $\mathcal{P}(n+1)$ is true: $\alpha_{n+1}$ is constructible 
and 
fulfills $\Parikh(\tau_1 \ldots \tau_{n+1}) \ge \Parikh(\tau_1' \ldots \tau_{n+1}')$.

We proved $\mathcal{P}(n)$ to be true for each $n \in [1,\ell-1]$.
At the end of the $\ell$-th iteration of the outer loop, there is no delayed occurrence anymore.

We deduce the lemma.
\end{proof}

Before studying Algorithm~\ref{AlgoSC2}, we give in Lemma~\ref{PropTau} 
a property of the sequence $\alpha = t \, \sigma_t$ obtained at the end of Algorithm~\ref{AlgoSC1}.
This result, derived from Lemma~$3$ in \cite{Hujsa2015}, compares the number of occurrences
in $\alpha$ with other ones in prefixes of $\kappa$, the latter being defined as a finite concatenation of the T-sequence $\sigma_r$.
This will prove useful for the study of Algorithm~\ref{AlgoSC2}.

\begin{lemma}[Property of $\alpha = t \, \sigma_t$]\label{PropTau}
Let $S = (N,M_0)$ be a live H$1$S system in which a T-sequence $\sigma_r$ is feasible.
Consider the sequence $\sigma_t$ constructed by Algorithm~\ref{AlgoSC1} after the firing of any transition $t$ in $p\lbul$.
Consider the sequences $\alpha = t \, \sigma_t$ and $\kappa = \sigma_r^\ell$ 
where 
$\ell \ge 1$ is the smallest integer 
such that $\Parikh(\alpha) \le \ell \cdot \Parikh(\sigma_r)$.
For every transition $t'$, denote by $E^{t'}$ the set $\{t'\}$ if $t' \not\in p\lbul$, the set $p\lbul$ otherwise;
then for each transition $t'$ such that $t_u = {tnext}(E^{t'}, \alpha, \kappa)$ is defined,
with 
$m = \Parikh(\alpha)(t_u) + 1$ 
and $K_u = K^m_u(\kappa)$,
we have that $\Parikh(\alpha)(t') = \Parikh(K_u)(t')$.
Besides, for each other transition $t''$, i.e.\ such that ${tnext}(E^{t''}, \alpha, \kappa)$ is not defined,
we have that 
$\Parikh(\alpha)(t'') = \Parikh(\kappa)(t'')$.
\end{lemma}
    
\begin{proof}
By Lemma \ref{FirstLoop}, Algorithm~\ref{AlgoSC1} terminates and is valid.
Consequently, at the end,
there is no delayed occurrence of any transition in the system obtained.
The equalities are then deduced as in the proof of Lemma~$3$ in \cite{Hujsa2015} by replacing $E^t$ with $p\lbul$.
\end{proof}

\subsubsection{Second part of the sequence construction}

Algorithm~\ref{AlgoSC2}, presented next, 
completes the sequence constructed by Algorithm~\ref{AlgoSC1}
and leads to the initial marking,
as illustrated in Figure~\ref{JFalgo2} with the particular case that $\ell = 1$.

\begin{algorithm}\label{AlgoSC2}
   \KwData{The sequences $\alpha = t \, \sigma_t$ and $\kappa = \sigma_r^\ell$, 
	    where $\ell \ge 1$ is the smallest integer 
such that $\Parikh(\alpha) \le \ell \cdot \Parikh(\sigma_r)$;
	    the marking $M_t'$ such that $\myarrow{M_0}{M_t'}{\alpha}$}
   \KwResult{A completion sequence $\sigma_t'$ that is feasible in $(N,M_t')$ such that $\myarrow{M_t'}{M_0}{\sigma_t'}$}

    \While{$\Parikh(\alpha) \neq \Parikh(\kappa)$}{
	Fire the transition $t_i$ whose next occurrence after its $\Parikh(\alpha)(t_i)$-th 
	appears first in $\kappa$\;
	$\alpha := \alpha \, t_i$\;
    }
    $\alpha$ is of the form $t \, \sigma_t \, \sigma_t'$\;
    \Return{$\sigma_t'$}

    \caption{Computation of the feasible sequence $\sigma_t'$ after the end of Algorithm~\ref{AlgoSC1}.}

\end{algorithm}

The following theorem shows that Algorithm~\ref{AlgoSC2} is valid and terminates.
It is Theorem $2$ in \cite{Hujsa2015} applied to our systems.

\begin{theorem}\label{ThSC}
Let $S=(N,M_0)$ be a live H$1$S system, with $N=(P,T,W)$, in which a T-sequence $\sigma_r$ is feasible.
For every transition $t$ enabled by $M_0$ such that $\myarrow{M_0}{M_t}{t}$ ~,
consider the sequence $\sigma^{\star} = \sigma_t \, \sigma_t'$
where $\sigma_t$ is constructed from $M_t$ by Algorithm~\ref{AlgoSC1}
and $\sigma_t'$ is built by Algorithm~\ref{AlgoSC2} after the execution of Algorithm~\ref{AlgoSC1}.
Then, $\sigma = t \, \sigma^{\star}$ is a T-sequence feasible in $S$
that satisfies $\Parikh(\sigma) = k \cdot \Parikh(\sigma_r)$ for some integer $k \ge 1$.
\end{theorem}

\begin{proof}
We use Lemma \ref{PropTau} and the rest of the proof remains the same as in~\cite{Hujsa2015}.
\end{proof}

\begin{figure}[!h]
%
%
\centering
\begin{tikzpicture}[mypetristyle,scale=0.8]

\node (p1) at (1,1.3) [place,tokens=2] {};
\node (p2) at (3,-0.5) [place,tokens=1] {};
\node (p3) at (3,2.5) [place,tokens=1] {};
\node (p4) at (5,1) [place] {};
\node (p5) at (3,1.3) [place,tokens=2] {};
\node (p6) at (3,0.7) [place,tokens=2] {};
\node (p7) at (1,0.7) [place,tokens=0] {};

\node [anchor=south] at (p1.north) {$p_1$};
\node [anchor=north] at (p2.south) {$p_2=p$};
\node [anchor=south] at (p3.north) {$p_3$};
\node [anchor=north east] at (p4.south west) {$p_4$};
\node [anchor=south] at (p5.north) {$p_5$};
\node [anchor=north] at (p6.south) {$p_6$};
\node [anchor=north] at (p7.south) {$p_7$};

\node (t1) at (0,1) [transition] {};
\node (t2) at (2,1) [transition] {};
\node (t3) at (4,1) [transition] {};
\node (t4) at (5,2.5) [transition] {};
\node (t5) at (5,-0.5) [transition] {};

\node [anchor=east] at (t1.west) {$t_1$};
\node [anchor=south] at (t2.north) {$t_2$};
\node [anchor=south] at (t3.north) {$t_3$};
\node [anchor=south] at (t4.north) {$t_4$};
\node [anchor=north] at (t5.south) {$t_5$};

\draw [->,thick] (p1) to node [above] {$2$} (t1);
\draw [->,thick] (t2) to node [above] {} (p1);

\draw [->,thick] (p2) to node [above right] {} (t2);
\draw [->,thick] (p2) to node [above left] {} (t3);

\draw [->,thick,bend left=60] (t1.north west) to node [above] {} (p3);
\draw [->,thick] (t2) to node [above] {} (p1);
\draw [->,thick,bend right=60] (t1.south west) to node [below] {} (p2);

\draw [->,thick] (t3) to node [above] {} (p4);
\draw [->,thick] (p3) to node [above] {} (t4);

\draw [->,thick] (t4) to node [right] {} (p4);
\draw [->,thick] (t3) to node [above] {} (p4);

\draw [->,thick] (p4) to node [right] {$2$} (t5);
\draw [->,thick] (t5) to node [below] {$2$} (p2);

\draw [->,thick] (t1) to node [below] {$2$} (p7);
\draw [->,thick] (p7) to node [below] {} (t2);

\draw [->,thick] (t3) to node [below] {} (p5);
\draw [->,thick] (p5) to node [below] {} (t2);

\draw [->,thick] (t2) to node [below] {} (p6);
\draw [->,thick] (p6) to node [below] {} (t3);
\end{tikzpicture}
%
%
\hspace*{2cm}
%
%
\begin{tikzpicture}[mypetristyle,scale=0.8]

\node (p1) at (1,1.3) [place,tokens=1] {};
\node (p2) at (3,-0.5) [place,tokens=1] {};
\node (p3) at (3,2.5) [place,tokens=2] {};
\node (p4) at (5,1) [place] {};
\node (p5) at (3,1.3) [place,tokens=2] {};
\node (p6) at (3,0.7) [place,tokens=2] {};
\node (p7) at (1,0.7) [place,tokens=1] {};

\node [anchor=south] at (p1.north) {$p_1$};
\node [anchor=north] at (p2.south) {$p_2=p$};
\node [anchor=south] at (p3.north) {$p_3$};
\node [anchor=north east] at (p4.south west) {$p_4$};
\node [anchor=south] at (p5.north) {$p_5$};
\node [anchor=north] at (p6.south) {$p_6$};
\node [anchor=north] at (p7.south) {$p_7$};

\node (t1) at (0,1) [transition] {};
\node (t2) at (2,1) [transition] {};
\node (t3) at (4,1) [transition] {};
\node (t4) at (5,2.5) [transition] {};
\node (t5) at (5,-0.5) [transition] {};

\node [anchor=east] at (t1.west) {$t_1$};
\node [anchor=south] at (t2.north) {$t_2$};
\node [anchor=south] at (t3.north) {$t_3$};
\node [anchor=south] at (t4.north) {$t_4$};
\node [anchor=north] at (t5.south) {$t_5$};

\draw [->,thick] (p1) to node [above] {$2$} (t1);
\draw [->,thick] (t2) to node [above] {} (p1);

\draw [->,thick] (p2) to node [above right] {} (t2);
\draw [->,thick] (p2) to node [above left] {} (t3);

\draw [->,thick,bend left=60] (t1.north west) to node [above] {} (p3);
\draw [->,thick] (t2) to node [above] {} (p1);
\draw [->,thick,bend right=60] (t1.south west) to node [below] {} (p2);

\draw [->,thick] (t3) to node [above] {} (p4);
\draw [->,thick] (p3) to node [above] {} (t4);

\draw [->,thick] (t4) to node [right] {} (p4);
\draw [->,thick] (t3) to node [above] {} (p4);

\draw [->,thick] (p4) to node [right] {$2$} (t5);
\draw [->,thick] (t5) to node [below] {$2$} (p2);

\draw [->,thick] (t1) to node [below] {$2$} (p7);
\draw [->,thick] (p7) to node [below] {} (t2);

\draw [->,thick] (t3) to node [below] {} (p5);
\draw [->,thick] (p5) to node [below] {} (t2);

\draw [->,thick] (t2) to node [below] {} (p6);
\draw [->,thick] (p6) to node [below] {} (t3);

\end{tikzpicture}


\caption{
The marking reached at the end of Algorithm~\ref{AlgoSC1}
by firing $\alpha = t_3 \, t_4 \, t_5 \, t_2$
in the system on the left of Figure~\ref{JFalgo1}
is depicted on the left.
Applying Algorithm~\ref{AlgoSC2} to this marking,
the sequence $t_1 \, t_2 \, t_3 \, t_5 \, t_3 \, t_2 \, t_1 \, t_2 \, t_4 \, t_5 \, t_3$ is fired,
completing $\alpha$ such that the initial marking is reached again, as shown on the right.
In this example, the sequence of firings corresponding to $t \, \sigma_t \, \sigma_t'$ in Figure~\ref{FigSigmaC}
is $t_3 \, t_4 \, t_5 \, t_2 \, t_1 \, t_2 \, t_3 \, t_5 \, t_3 \, t_2 \, t_1 \, t_2 \, t_4 \, t_5 \, t_3$,
whose Parikh vector equals $\Parikh(\sigma_r)$.
Here,
$\ell=1$ and $\kappa=\sigma_r$.
}
\label{JFalgo2}
\end{figure}
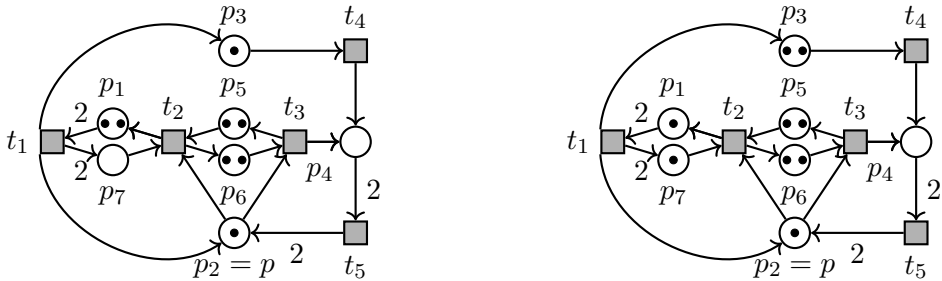

\subsubsection{Deriving the reversibility characterization}\label{CharRev}

We are now able to derive the next condition for reversibility, 
whose proof is illustrated in Figure~\ref{FigReturnsM0}.
It is a variant of Corollary~$1$ in~\cite{Hujsa2015}.

\begin{corollary}[Reversibility characterization]\label{H1Srev}
Consider a live H1S system $S=(N,M_0)$.
Then, $S$ is reversible iff $S$ enables a T-sequence.
\end{corollary}

\begin{proof}
As mentioned earlier, the left-to-right direction is clear.
We prove by induction on the length~$n$ of any feasible sequence $\sigma$
that,
after the firing of this sequence, another sequence is feasible that leads to the initial marking.
More precisely, we mimic the proof of Corollary $1$ in \cite{Hujsa2015}
by replacing the property $\mathcal{P}(n)$ with the following one:

``Consider a live H$1$S system $S=(N,M_0)$ enabling a T-sequence $\sigma_r$ and a sequence $\sigma$ of length~$n$.
There exists a firing sequence $\sigma^\star$ such that $\myarrow{M_0}{M_0}{\sigma \, \sigma^\star}$.''

The main steps are illustrated in Figure~\ref{FigReturnsM0}.

The base case, with $n=0$, is clear: $\sigma^\star = \epsilon$.
For the inductive case, with $n>0$,
let us suppose that $\mathcal{P}(n-1)$ is true
and
let us write $\sigma = t \, \sigma'$.
The firing of $t$ from $M_0$ leads to a marking $M_t$, and $\sigma'$ leads to $M$ from $M_t$.
Theorem~\ref{ThSC} applies: a sequence $\sigma_t^\star$ is feasible at $M_t$, leading to $M_0$,
such that $t \, \sigma_t^\star$ is a T-sequence.

Since $M_t$ is live and $\sigma_t^\star \, t$ is a T-sequence feasible at $M_t$,
the induction hypothesis applies to $(N,M_t)$ and $\sigma'$ of length $n-1$,
so that a sequence $\sigma''$ exists that is feasible from $M$ and leads to~$M_t$.

Thus, after the firing of $\sigma$, the sequence $\sigma^\star = \sigma'' \, \sigma_t^\star$ leads to the initial marking. 
We deduce reversibility.
\end{proof}

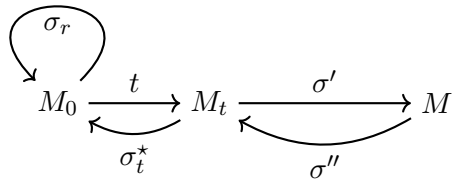
\begin{figure}[!h]

\begin{minipage}{1\linewidth}
\centering
\begin{tikzpicture}[scale=1]

\node (M0) at (0,2) {$M_0$};
\node (Mt) at (2,2) {$M_t$};
\node (M) at (5,2) {$M$};

\draw [->,thick,above] (M0) to node {$t$} (Mt);
\draw [->,thick,above] (Mt) to node {$\sigma'$} (M);

\draw [->,thick,below,bend left=30] (Mt) to node {$\sigma_t^\star$} (M0);
\draw [->,thick,below,bend left=30] (M) to node {$\sigma''$} (Mt);

\draw [->,thick, loop, below] (M0) to node {$\sigma_r$} (M0);

\end{tikzpicture}

\end{minipage}

\caption{Illustration of the general case in the proof of Corollary~\ref{H1Srev}.}
\label{FigReturnsM0}
\end{figure}

\subsection{Non-extensibility to several shared places}\label{NonExtensibility2places}

We show that Corollary~\ref{H1Srev} is no more true in the case of two shared places,
even in unit-weighted H$2$S-WMG:
we design a counter-example in Figure~\ref{NonRevTwoPlacesBis}.

\begin{figure}[!h]

\centering

\begin{tikzpicture}[scale=0.85,mypetristyle]

\node (p0) at (-0.5,2.8) [place,tokens=1] {};
\node (p1) at (2.4,2) [place,tokens=1] {};
\node (p2) at (1,0.5) [place,tokens=0] {};
\node (p3) at (3.9,-0.3) [place,tokens=0] {};
\node (p4) at (0.3,1) [place,tokens=1] {};
\node (p5) at (0,2) [place,tokens=1] {};
\node (p6) at (3.4,0.5) [place,tokens=1] {};
\node (p7) at (3,1.5) [place,tokens=1] {};

\node [anchor=east] at (p0.west) {$p_0$};
\node [anchor=south] at (p1.north) {$p_1~~~$};
\node [anchor=north] at (p2.south) {$~~~p_2$};
\node [anchor=west] at (p3.east) {$p_3$};
\node [anchor=south] at (p4.north) {$p_4$};
\node [anchor=south] at (p5.north) {$p_5$};
\node [anchor=north] at (p6.south) {$p_6$};
\node [anchor=north] at (p7.south) {$~~~p_7$};

\node (t0) at (3.9,2.8) [transition] {};
\node (t1) at (1,2) [transition] {};
\node (t2) at (-0.5,-0.3) [transition] {};
\node (t3) at (2.4,0.5) [transition] {};

\node [anchor=west] at (t0.east) {$t_0$};
\node [anchor=south] at (t1.north) {$t_1$};
\node [anchor=east] at (t2.west) {$t_2$};
\node [anchor=north] at (t3.south) {$t_3$};

\draw [->,thick,bend right=30] (t2) to node {} (p2);
\draw [->,thick,bend left=0] (p2) to node {} (t2);

\draw [->,thick] (p2) to node {} (t1);
\draw [->,thick] (t3) to node {} (p2);
\draw [->,thick] (p1) to node {} (t3);
\draw [->,thick] (t1) to node {} (p1);

\draw [->,thick] (t0) to node {} (p7);
\draw [->,thick] (p7) to node {} (t3);

\draw [->,thick,bend right=30] (t0.west) to node {} (p1);
\draw [->,thick,bend left=0] (p1) to node {} (t0.west);

\draw [->,thick, bend right=0] (p0.north east) to node {} (t0.north west);
\draw [->,thick] (t2.north west) to node {} (p0.south west);


\draw [->,thick,bend left=0] (t0.south east) to node {} (p3.north east);
\draw [->,thick,bend left=0] (p3.south west) to node {} (t2.south east);

\draw [->,thick] (t2) to node {} (p4);
\draw [->,thick] (p4) to node {} (t1);

\draw [->,thick,bend right=0] (t1) to node {} (p5);
\draw [->,thick,bend right=0] (p5) to node {} (t2.north);

\draw [->,thick,bend right=0] (t3) to node {} (p6);
\draw [->,thick,bend right=0] (p6) to node {} (t0);

\end{tikzpicture}
\hspace*{2cm}
\raisebox{0mm}{
\begin{tikzpicture}[scale=0.7,mypetristyle]
\node[ltsNode,label=above:$s_0$,minimum width=7pt](s0)at(1,2.5){};
\node[ltsNode,label=above:$s_1$](s1)at(3,2.5){};
\node[ltsNode,label=left:$s_2$](s2)at(1,1){};
\node[ltsNode,label=right:$s_3$](s3)at(3,1){};
\node[ltsNode,label=below:$s_4$](s4)at(1,-0.5){};
\node[ltsNode,label=below:$s_5$](s5)at(3,-0.5){};
\node[ltsNode,label=left:$s_6$](s6)at(0,1){};
\node[ltsNode,label=right:$s_7$](s7)at(4,1){};
\draw[-{>[scale=2.5,length=2,width=2]}](s0)to node[above]{$t_3$}(s1);
\draw[-{>[scale=2.5,length=2,width=2]}](s1)to node[left]{$t_1$}(s3);
\draw[-{>[scale=2.5,length=2,width=2]}](s3)to node[left]{$t_0$}(s5);
\draw[-{>[scale=2.5,length=2,width=2]}](s5)to node[below right]{$t_3$}(s7);
\draw[-{>[scale=2.5,length=2,width=2]}](s0)to node[right]{$t_0$}(s2);
\draw[-{>[scale=2.5,length=2,width=2]}](s2)to node[right]{$t_3$}(s4);
\draw[-{>[scale=2.5,length=2,width=2]}](s4)to node[above]{$t_1$}(s5);
\draw[-{>[scale=2.5,length=2,width=2]}](s4)to node[below left]{$t_2$}(s6);
\draw[-{>[scale=2.5,length=2,width=2]}](s6)to node[above left]{$t_1$}(s0);
\draw[-{>[scale=2.5,length=2,width=2]}](s7)to node[above right]{$t_2$}(s1);
\end{tikzpicture}
}
\caption{On the left, a unit-weighted, live, structurally bounded H$2$S-WMG system, where the two shared places are $p_1$ and $p_2$.
The system enables the T-sequence $t_0 \, t_3 \, t_2 \, t_1$ but is not reversible.
On the right, its non strongly connected reachability graph is pictured, with initial state $s_0$.
}
\label{NonRevTwoPlacesBis}

\end{figure}
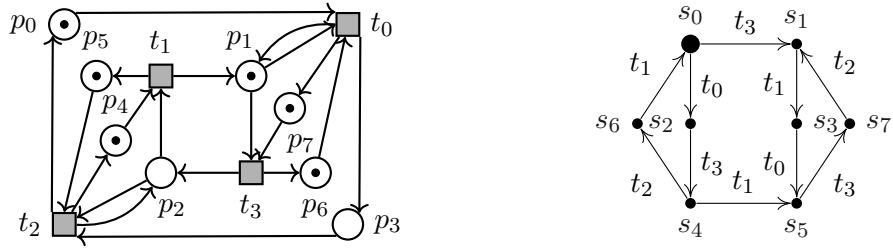

Once an H$1$S system is known to be live, 
checking reversibility thus amounts to checking the existence of a feasible T-sequence.
In next section,
we illustrate on use-cases how our results can be used to simplify the checking of liveness and reversibility.

\section{Use-cases}\label{UseCase}

In this section, we apply our results to two use-cases  modeled with H$1$S-WMG,
allowing to check liveness and reversibility more efficiently.
The first use-case is the \emph{Swimming pool protocol}, extracted from the Model Checking Contest.
The second one is the logotype of the \emph{Petri nets} conference 
(\emph{International Conference on Application and Theory of Petri Nets and Concurrency}).

\subsection{The Swimming pool protocol}\label{UseCaseSP}

The Swimming pool protocol, extracted from the Model Checking Contest database\footnote{\url{https://mcc.lip6.fr/models.php}},
stems from~\cite{latteux1980synchronisation}.
Variants of the protocol can be found in~\cite{Brams83}.
It is encoded by the system $S=(N,M_0)$ on the left of Figure~\ref{OneChoiceExample}, parameterized by the variables $a,b,c \ge 1$.
Initially, there are $a$ users outside the building, $b$ free bags and $c$ free cabins.
The protocol for one user follows: 
a user gets into the building, 
asks for the key of a cabin,
asks for a bag to pack his clothes,
uses the cabin to undress and to put his swimming suit on,
returns the key then swims,
gets out of the swimming pool,
asks for the key of a cabin,
dresses,
gives the bag back,
gives the key back, 
and finally leaves the building.\\
So as to check the liveness of this system $S$ for given values of $a$, $b$ and $c$,
one can use Proposition~\ref{LivenessOfHAC}, since $S$ is H$1$S.
It amounts to check that no minimal siphon is deadlocked at any reachable marking.
However,
since the number of these siphons is often exponential in the number of places, this method is generally too costly.
To alleviate this difficulty, we exploit our new results in the following.

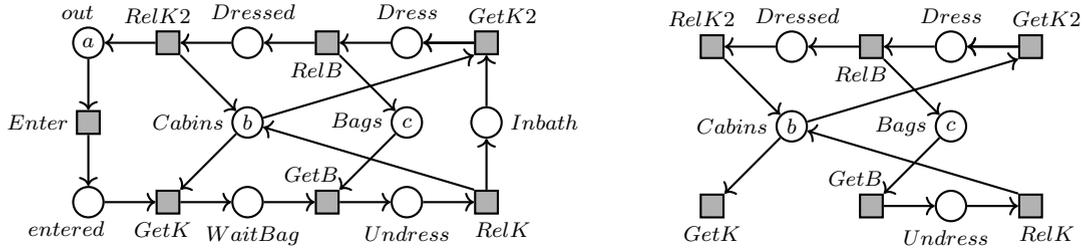
\begin{figure}[!h]
\centering
\begin{tikzpicture}[scale=0.7,mypetristyle]

\scriptsize

\node (out) at (1,3) [place,tokens=0,inner sep=1pt] {$a$};
\node [anchor=south] at (out.north) {$out~~~~$};

\node (entered) at (1,0) [place,tokens=0] {};
\node [anchor=north] at (entered.south west) {$entered~~~~$};

\node (WaitBag) at (4,0) [place] {};
\node [anchor=north] at (WaitBag.south) {$~~WaitBag$};

\node (Dressed) at (4,3) [place,tokens=0] {};
\node [anchor=south] at (Dressed.north) {$~~Dressed$};

\node (Cabins) at (4,1.5) [place,tokens=0,inner sep=1pt] {$b$};
\node [anchor=east] at (Cabins.west) {$Cabins$};

\node (Bags) at (7,1.5) [place,tokens=0,inner sep=1pt] {$c$};
\node [anchor=east] at (Bags.west) {$Bags$};

\node (Dress) at (7,3) [place] {};
\node [anchor=south] at (Dress.north) {$Dress$};

\node (Undress) at (7,0) [place] {};
\node [anchor=north] at (Undress.south) {$Undress$};

\node (InBath) at (8.5,1.5) [place,tokens=0] {};
\node [anchor=west] at (InBath.east) {$Inbath$};

\node (Enter) at (1,1.5) [transition,thick] {};
\node (RelK2) at (2.5,3) [transition,thick] {};
\node (GetK) at (2.5,0) [transition,thick] {};
\node (GetB) at (5.5,0) [transition,thick] {};
\node (RBag) at (5.5,3) [transition,thick] {};
\node (RelK) at (8.5,0) [transition,thick] {};
\node (GetK2) at (8.5,3) [transition,thick] {};

\node [anchor=east] at (Enter.west) {$Enter$};
\node [anchor=south] at (RelK2.north) {$RelK2~~~~$};
\node [anchor=north] at (GetK.south) {$GetK~~$};
\node [anchor=south] at (GetB.north) {$GetB~~~~~~$};
\node [anchor=north] at (RBag.south) {$RelB~~~~$};
\node [anchor=north] at (RelK.south) {$~~~~~~RelK$};
\node [anchor=south] at (GetK2.north) {$~~~~~~GetK2$};

\draw [->,thick] (out) to node [] {} (Enter);
\draw [->,thick] (Enter) to node [] {} (entered);
\draw [->,thick] (entered) to node [] {} (GetK);
\draw [->,thick] (GetK) to node [] {} (WaitBag);
\draw [->,thick] (WaitBag) to node [] {} (GetB);
\draw [->,thick] (GetB) to node [] {} (Undress);
\draw [->,thick] (Undress) to node [] {} (RelK);
\draw [->,thick] (RelK) to node [] {} (InBath);
\draw [->,thick] (InBath) to node [] {} (GetK2);
\draw [->,thick] (GetK2) to node [] {} (Dress);
\draw [->,thick] (Cabins) to node [] {} (GetK);
\draw [->,thick] (Dress) to node [] {} (RBag);
\draw [->,thick] (RBag) to node [] {} (Dressed);
\draw [->,thick] (Dressed) to node [] {} (RelK2);
\draw [->,thick] (RelK2) to node [] {} (out);
\draw [->,thick] (RelK2) to node [] {} (Cabins);
\draw [->,thick] (Cabins) to node [] {} (GetK2.south west);
\draw [->,thick] (GetK2) to node [] {} (Dress);
\draw [->,thick] (RelK.north west) to node [] {} (Cabins);
\draw [->,thick] (RBag) to node [] {} (Bags);
\draw [->,thick] (Bags) to node [] {} (GetB);
\end{tikzpicture}
%
%
\hspace*{8mm}
\begin{tikzpicture}[scale=0.7,mypetristyle]

\scriptsize

%


\node (Dressed) at (4,3) [place,tokens=0] {};
\node [anchor=south] at (Dressed.north) {$~~Dressed$};

\node (Cabins) at (4,1.5) [place,tokens=0,inner sep=1pt] {$b$};
\node [anchor=east] at (Cabins.west) {$Cabins$};

\node (Bags) at (7,1.5) [place,tokens=0,inner sep=1pt] {$c$};
\node [anchor=east] at (Bags.west) {$Bags$};

\node (Dress) at (7,3) [place] {};
\node [anchor=south] at (Dress.north) {$Dress$};

\node (Undress) at (7,0) [place] {};
\node [anchor=north] at (Undress.south) {$Undress~~$};


\node (RelK2) at (2.5,3) [transition,thick] {};
\node (GetK) at (2.5,0) [transition,thick] {};
\node (GetB) at (5.5,0) [transition,thick] {};
\node (RBag) at (5.5,3) [transition,thick] {};
\node (RelK) at (8.5,0) [transition,thick] {};
\node (GetK2) at (8.5,3) [transition,thick] {};

\node [anchor=south] at (RelK2.north) {$RelK2~~~~$};
\node [anchor=north] at (GetK.south) {$GetK$};
\node [anchor=south] at (GetB.north) {$GetB~~~~~~$};
\node [anchor=north] at (RBag.south) {$RelB~~~~$};
\node [anchor=north] at (RelK.south) {$~~~~~~RelK$};
\node [anchor=south] at (GetK2.north) {$~~~~~~GetK2$};

%
%
\draw [->,thick] (GetB) to node [] {} (Undress);
\draw [->,thick] (Undress) to node [] {} (RelK);
%
\draw [->,thick] (GetK2) to node [] {} (Dress);
\draw [->,thick] (Cabins) to node [] {} (GetK);
\draw [->,thick] (Dress) to node [] {} (RBag);
\draw [->,thick] (RBag) to node [] {} (Dressed);
\draw [->,thick] (Dressed) to node [] {} (RelK2);
\draw [->,thick] (RelK2) to node [] {} (Cabins);
\draw [->,thick] (Cabins) to node [] {} (GetK2.south west);
\draw [->,thick] (GetK2) to node [] {} (Dress);
\draw [->,thick] (RelK.north west) to node [] {} (Cabins);
\draw [->,thick] (RBag) to node [] {} (Bags);
\draw [->,thick] (Bags) to node [] {} (GetB);
\end{tikzpicture}
\caption{
The swimming-pool protocol is modeled by the parameterized H$1$S-WMG system $S=(N,M_0)$ on the left, with $a,b,c \ge 1$.
Consider the minimal siphon $Q=\{Dressed, Dress, Cabins, Bags, Undress\}$ in $S$:
it induces the non-live, unbounded system $S_Q$ on the right.
All the other minimal siphons of the system $S$ induce live strongly connected state-machine P-subsystems.
Each firing of transition $RelB$ generates one new token in~$Q$, and each firing of transition $GetK$ removes one token from~$Q$.
The system $S$ is live for $a=b=c=1$, but is no more live if $a=2$ and $b=c=1$ since it permits to remove all tokens from $Q$.
There are various other live instantiations, e.g.\ $a=14,b=10,c=5$ and $a=15,b=10,c=6$
and
various other non-live ones, e.g.\ $a=15,b=10,c=5$.
Instead of checking each minimal siphon, it suffices to apply our results of Section~\ref{SecLiveness},
yielding an ILP of polynomial size that checks the existence of a potentially reachable deadlock, thus non-liveness, with complexity NP.
}
\label{OneChoiceExample}
\end{figure}

\noindent {\bf Checking liveness more efficiently using Corollary~\ref{CheckLivenessH1SWMG} and Proposition~\ref{PropTh5.4}.}
The net $N$ is conservative (which can be checked in polynomial-time with a linear program over the rational numbers), 
thus structurally bounded (see~\cite{LAT98,HDM2014}) and consistent (which is also checked in polynomial-time).
Moreover, it is a strongly connected H1S-WMG whose shared place deletion yields a strongly connected WMG,
hence Corollary~\ref{CheckLivenessH1SWMG} and Proposition~\ref{PropTh5.4} can be used with the associated ILP to check non-liveness 
with complexity NP in $|P|\cdot(|T|+|P|)\cdot m$, 
where $m$ is the maximal length of the net binary-encoded numbers 
(which are the arc weights and the given upper bounds on the structural bounds).\\
Besides, notice that the strongly connected, unit-weighted WMG $S_\mathrm{WMG}$ 
obtained by deleting the shared place $Cabins$ is also conservative, consistent, live and bounded.
If needed, its liveness can be checked in polynomial-time~\cite{March09}.\\
Now, so as to check reversibility more efficiently, 
let us recall the next two results, extracted from Theorem~4.10.2 and Proposition~4.7 in~\cite{WTS92}.

\begin{proposition}[Reversibility of live and bounded WMG~\cite{WTS92}]\label{LRWMG}
Each consistent and live WMG is reversible.
\end{proposition}

\begin{proposition}[Fireability of the minimal T-semiflow~\cite{WTS92}]\label{WMGfireY}
If $S=(N,M_0)$ is a live WMG with incidence matrix $I$, then for each T-vector $Y \ge \one$ such that $I \cdot Y \ge 0$
there exists
$\sigma_Y$ feasible in $S$ such that $\Parikh(\sigma_Y) = Y$.
\end{proposition}

\noindent {\bf Checking reversibility under the liveness assumption.}
Applying Corollary~\ref{H1Srev}, we only need to check the existence of a feasible T-sequence.
Moreover, the system $S$ of Figure~\ref{OneChoiceExample} has minimal T-semiflow $\mathcal{Y} = \one$
which is a (minimal) consistency vector, which can be computed in polynomial time.
The WMG $S_\mathrm{WMG}$ obtained by deleting the shared place $Cabins$ is well-formed and live 
(as recalled earlier, it is also a polynomial-time problem), 
thus reversible by Proposition~\ref{LRWMG}
and it enables a T-sequence $\sigma_\mathcal{Y}$ with Parikh vector $\mathcal{Y}$ by Proposition~\ref{WMGfireY}.
If, in addition, the shared place $Cabins$ contains at least $|Cabins\lbul|=2$ tokens,
then $\sigma_\mathcal{Y}$ is also feasible in $S$, implying its reversibility by Corollary~\ref{H1Srev}.
This is a polynomial-time sufficient condition of reversibility.\\
Otherwise, when $Cabins$ initially contains a single token, the checking complexity may be higher
and will make the subject of a future work for further improvement.
The difficulty comes from the conflict resolution policy for the shared place:
when constructing a T-sequence, 
a series of "bad" choices can lead to a longer T-sequence (whose Parikh vector is a multiple of $\mathcal{Y}$),
and one has to stop this increase, possibly by backtracking.

\subsection{The Petri nets conference emblem}\label{UseCasePNF}

On the left of Figure~\ref{PNconfLOGO}, we depict the emblem of the
\emph{International Conference on Application and Theory of Petri Nets and Concurrency}.
On the right, a simplified version with isomorphic reachability graph (hence preserving liveness, for instance)
obtained by applying some of the efficient reduction rules of~\cite{BLD2018} under some assumptions on the initial marking. 
Both nets are structurally bounded and fulfill the conditions of Corollary~\ref{CheckLivenessH1SWMG}
and 
can consequently benefit from the ILP of Section~\ref{SubsecSmallILP} to check (non-)liveness.

Recall that the number of inequalities and variables of this ILP is linear in the number of places and transitions,
and the length of each inequality is linear in the number of places, transitions and in the number of bits 
in the largest binary-encoded number among weights and structural bounds;
this describes the input size of the ILP.
The general problem of solving an ILP is NP-complete (even for the decision version) in the input size with binary-encoded numbers,
thus the problem of solving our ILP belongs to NP, implying that the non-liveness problem in this class belongs to co-NP.

Once liveness has been checked for a given initial marking $M_0$, reversibility can be checked using Corollary~\ref{H1Srev},
i.e.\ 
by checking the existence of a feasible T-sequence.
Since there is a unique minimal T-semiflow in both nets, which is $\one$, i.e.\ a consistency vector,
each T-sequence, if any, has a Parikh vector equal to a multiple of $\one$.
If, moreover, $m_5 \ge 2$, we know (as in the Swimming pool example) that a T-sequence with Parikh vector $\one$ is feasible, implying reversibility.

\begin{figure}[!h]

\centering

\begin{tikzpicture}[scale=1,mypetristyle]

\node (p0) at (4.5,1.5) [place,inner sep=0pt] {\tiny $m_0$};
\node (p1) at (1.5,0) [place,inner sep=0pt] {\tiny $m_1$};
\node (p2) at (2.05,0) [place,inner sep=0pt] {\tiny $m_2$};
\node (p3) at (3,0) [place,inner sep=0pt] {\tiny $m_3$};
\node (p4) at (3.95,0) [place,inner sep=0pt] {\tiny $m_4$};
\node (p5) at (4.5,0) [place,inner sep=0pt] {\tiny $m_5$};
\node (p6) at (6,0) [place,inner sep=0pt] {\tiny $m_6$};
\node (p7) at (7.5,0) [place,inner sep=0pt] {\tiny $m_7$};
\node (p8) at (4.5,-1.5) [place,inner sep=0pt] {\tiny $m_8$};

\node [anchor=south] at (p0.north) {$p_0$};
\node [anchor=north] at (p1.south east) {$~~p_1$};
\node [anchor=west,inner sep=0pt] at (p2.east) {$p_2$};
\node [anchor=west,inner sep=0pt] at (p3.east) {$p_3$};
\node [anchor=south, inner sep=1pt] at (p4.north) {$p_4$};
\node [anchor=west] at (p5.east) {$p_5$};
\node [anchor=west,inner sep=1pt] at (p6.east) {$p_6$};
\node [anchor=east,inner sep=1pt] at (p7.west) {$p_7$};
\node [anchor=north] at (p8.south) {$p_8$};

\node (t0) at (3,1) [transition] {};
\node (t1) at (6,1) [transition] {};
\node (t2) at (3,-1) [transition] {};
\node (t3) at (6,-1) [transition] {};

\node [anchor=south east] at (t0.north west) {$t_0$};
\node [anchor=south west] at (t1.north east) {$t_1$};
\node [anchor=north east] at (t2.south west) {$t_2$};
\node [anchor=north west] at (t3.south east) {$t_3$};

\draw [->,thick,bend right=0,rounded corners=1mm] (t0.west) -- (1.5,1) -- (p1);
\draw [->,thick,bend right=0,rounded corners=1mm] (p1) -- (1.5,-1) -- (t2.west);

\draw [->,thick,bend right=0] (t0.south west) to node {} (p2);
\draw [->,thick,bend right=0] (p2) to node {} (t2.north west);

\draw [->,thick,bend right=0] (t2) to node {} (p3);
\draw [->,thick,bend right=0] (p3) to node {} (t0);

\draw [->,thick,bend right=0] (t2.north east) to node {} (p4);
\draw [->,thick,bend right=0] (p4) to node {} (t0.south east);

\draw [->,thick,bend left=30,rounded corners=1mm] (t0) -- (3,1.5) -- (p0);
\draw [->,thick,bend left=30,rounded corners=1mm] (p0) -- (6,1.5) -- (t1);

\draw [->,thick,bend right=0] (t1) to node {} (p6);
\draw [->,thick,bend right=0] (p6) to node {} (t3);

\draw [->,thick,bend left=0,rounded corners=1mm] (t1) -- (7.5,1) -- (p7);
\draw [->,thick,bend left=0,rounded corners=1mm] (p7) -- (7.5,-1) -- (t3);

\draw [->,thick,bend left=0,rounded corners=1mm] (t3) -- (6,-1.5) -- (p8);
\draw [->,thick,bend left=0,rounded corners=1mm] (p8) -- (3,-1.5) -- (t2);

\draw [->,thick,bend right=40,rounded corners=1mm] (t2) -- (4.35,-1) -- (p5.south west);
\draw [->,thick,bend right=40,rounded corners=1mm] (p5.north west) -- (4.35,1) -- (t0);

\draw [->,thick,bend left=40,rounded corners=1mm] (t3) -- (4.65,-1) -- (p5.south east);
\draw [->,thick,bend left=40,rounded corners=1mm] (p5.north east) -- (4.65,1) -- (t1);

\end{tikzpicture}
\hspace*{15mm}
\begin{tikzpicture}[scale=1,mypetristyle]

\node (p0) at (4.5,1.5) [place,inner sep=0pt] {\tiny $m_0$};
\node (p1) at (1.5,0) [place,inner sep=0pt] {\tiny $m_1$};
\node (p3) at (3,0) [place,inner sep=0pt] {\tiny $m_3$};
\node (p5) at (4.5,0) [place,inner sep=0pt] {\tiny $m_5$};
\node (p6) at (6,0) [place,inner sep=0pt] {\tiny $m_6$};
\node (p8) at (4.5,-1.5) [place,inner sep=0pt] {\tiny $m_8$};

\node [anchor=south] at (p0.north) {$p_0$};
\node [anchor=north] at (p1.south east) {$~~p_1$};
\node [anchor=west,inner sep=0pt] at (p3.east) {$p_3$};
\node [anchor=west] at (p5.east) {$p_5$};
\node [anchor=west,inner sep=1pt] at (p6.east) {$p_6$};
\node [anchor=north] at (p8.south) {$p_8$};

\node (t0) at (3,1) [transition] {};
\node (t1) at (6,1) [transition] {};
\node (t2) at (3,-1) [transition] {};
\node (t3) at (6,-1) [transition] {};

\node [anchor=south east] at (t0.north west) {$t_0$};
\node [anchor=south west] at (t1.north east) {$t_1$};
\node [anchor=north east] at (t2.south west) {$t_2$};
\node [anchor=north west] at (t3.south east) {$t_3$};

\draw [->,thick,bend right=0,rounded corners=1mm] (t0.west) -- (1.5,1) -- (p1);
\draw [->,thick,bend right=0,rounded corners=1mm] (p1) -- (1.5,-1) -- (t2.west);


\draw [->,thick,bend right=0] (t2) to node {} (p3);
\draw [->,thick,bend right=0] (p3) to node {} (t0);


\draw [->,thick,bend left=30,rounded corners=1mm] (t0) -- (3,1.5) -- (p0);
\draw [->,thick,bend left=30,rounded corners=1mm] (p0) -- (6,1.5) -- (t1);

\draw [->,thick,bend right=0] (t1) to node {} (p6);
\draw [->,thick,bend right=0] (p6) to node {} (t3);


\draw [->,thick,bend left=0,rounded corners=1mm] (t3) -- (6,-1.5) -- (p8);
\draw [->,thick,bend left=0,rounded corners=1mm] (p8) -- (3,-1.5) -- (t2);

\draw [->,thick,bend right=40,rounded corners=1mm] (t2) -- (4.35,-1) -- (p5.south west);
\draw [->,thick,bend right=40,rounded corners=1mm] (p5.north west) -- (4.35,1) -- (t0);

\draw [->,thick,bend left=40,rounded corners=1mm] (t3) -- (4.65,-1) -- (p5.south east);
\draw [->,thick,bend left=40,rounded corners=1mm] (p5.north east) -- (4.65,1) -- (t1);

\end{tikzpicture}
\caption{On the left, the emblem of the Petri net conference, in which we parameterize the initial marking. 
This net is a unit-weighted H$1$S-WMG.
It is $1$-conservative hence structurally bounded.
It is also strongly connected, and the WMG obtained by deleting $p_5$ is also strongly connected,
so that the conditions of Corollary~\ref{CheckLivenessH1SWMG} are fulfilled,
implying that the ILP of Section~\ref{SubsecSmallILP} can be used to check liveness.
Suppose that $m_2 \ge m_1$, $m_4 \ge m_3$ and $m_7 \ge m_6$ (w.l.o.g.): we deduce from the reduction rules of~\cite{BLD2018}
that $p_2$, $p_4$ and $p_7$ are redundant and that their removal yields a system
with isomorphic reachability graph, depicted on the right.
Now, if $m_5 \ge 2$,   
a T-sequence with Parikh vector $\one$ is feasible;
for other markings, the computation might be more costly.
}

\label{PNconfLOGO}
\end{figure}
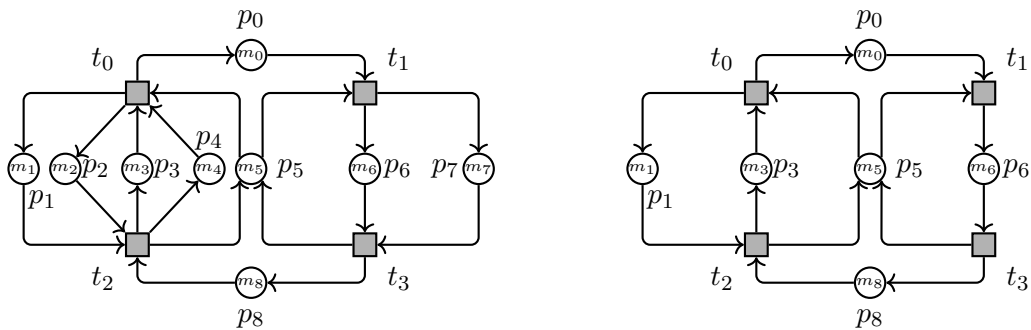

\section{Toward modular systems with \texorpdfstring{H$1$S}{} agents}\label{ModularH1C}

H$1$S systems, although more expressive than CF systems, can only model one shared place.
Moreover, our results on liveness and reversibility do not extend to H$2$S systems.
To increase the range of applicability of our results,
we can exploit modular approaches, such as the work of~\cite{DSSP98}
on Deterministically Synchronized Sequential Process (DSSP) Petri nets.
The latter are modular systems composed by a set of distributed agents 
communicating with asynchronous message passing.
Each agent is a $1$-bounded, live state machine.
In~\cite{DSSP98}, the authors provide structural results related to liveness
and derive a liveness enforcement technique.

In~\cite{SCECS96}, DSSP systems are generalized to $\{$SC$\}^*$ECS systems, which allow weights and shared places in a restricted fashion,
while retaining some strong structural properties of the DSSP systems.
They allow agents to be HFC systems.

A perspective is to exploit such structural techniques to extend the expressiveness of H$1$S systems, 
e.g.\ by allowing agents to be H$1$S systems,
so as to obtain live and reversible systems with an arbitrary number of shared places.

\section{Conclusions and Perspectives}\label{SecConclu}

In Petri net analysis, the reachability, liveness and reversibility checking problems are well-known to be intractable.
Their importance for real-life applications triggered numerous fruitful studies.
In this work, we obtained new results alleviating their solving difficulty in particular Petri net subclasses, summarized as follows.

We introduced new notions useful to our purpose, such as initial directedness and strong liveness.
We proved a property relating these notions in weighted Petri nets.\\
For the H$1$S-\WMGineq{} subclass, we obtained a variant of the well-known Keller's theorem, exploiting initial directedness.
This sheds new light on the reachability problem in this class.\\
For the same subclass, we developed a new liveness characterization, based on the state equation.
Under further classical assumptions such as strong connectedness,
we derived from it the first liveness characterization whose complexity lies in co-NP.\\
For the larger H$1$S class, under the liveness assumption,
we proved the first characterization of reversibility that leads to the first polynomial-time and wide-ranging sufficient conditions
of liveness and reversibility in this class.\\
Thus, in most cases, our results drastically reduce the part of the reachability graph to be checked in H$1$S systems,
improving upon all previously known general methods.
We also provided several counter-examples showing that our new conditions do not apply to homogeneous systems with more than one shared place.
Hence, the H$1$S class embodies a part of the frontier between the systems whose reachability, liveness and reversibility checking
can be alleviated with such approaches, and the other ones.\\
Finally, we highlighted the effectiveness and scalability or our approach on two use-cases known to the Petri net community.

To extend this work, there are various ways in which one might proceed.
One way is to consider modular systems, as envisioned in Section~\ref{ModularH1C}.
Another way is to provide more efficient checking algorithms for liveness and reversibility in the H1S class, 
and to determine other Petri net classes that may benefit from the proof techniques developed in this work.
Since we did not aim at optimizing our ILP for liveness checking, focusing on the complexity class,
it is subject to imrovement: we might look for the minimal number of variables or inequalities, 
and possibly avoid the use of any transformation.
We also plan to implement our techniques in a model-checker and complete this work with a series of benchmarks on much more complex use-cases.\\
A complementary objective is to investigate the combination of our methods with the reduction techniques developed in~\cite{BLD2018},
as well as separability~\cite{BD2011,BHW18} and marking homothety~\cite{Homothetic2012},
always with the aim of improving their efficiency and of widening their applicative area.

} 

\bibliographystyle{fundam}
\bibliography{FI2020}

\end{document}